\documentclass[lettersize,journal]{IEEEtran} 
\usepackage{amsmath,amsfonts}
\usepackage{bm}
\usepackage{algorithmic}
\usepackage{algorithm}
\usepackage{array}
\usepackage{textcomp}   
\usepackage{stfloats}
\usepackage{url}
\usepackage{verbatim}
\usepackage{graphicx}
\usepackage{cite}
\usepackage{xcolor}

\usepackage{amsthm}
\newtheorem{theorem}{Theorem}

\usepackage[caption=false, font=footnotesize]{subfig}

\usepackage[normalem]{ulem}

\usepackage{empheq}

\allowdisplaybreaks[4]

\usepackage{xpatch}

\title{
Divergence-Based Adaptive Aggregation for Byzantine Robust Federated Learning
}

\begin{document}

\author{Bingnan Xiao,
        Feng~Zhu,
        Jingjing Zhang,~\IEEEmembership{Member,~IEEE}, 
        Wei~Ni,~\IEEEmembership{Fellow,~IEEE},
        and
        Xin~Wang,~\IEEEmembership{Fellow,~IEEE} 
\thanks{
This work was supported in part by the Innovation Program of Shanghai Municipal Science and Technology Commission under Grant 25DP1500300, and in part by the National Natural Science Foundation of China under Grant 62231010. (Corresponding author: \textit{Jingjing Zhang}, \textit{Xin Wang})
}

\thanks{B. Xiao, J. Zhang, and X. Wang are with the Key Laboratory of EMW Information (MoE), College of Future Information Technology, Fudan University, Shanghai 200433, China (e-mail: 22110720061@m.fudan.edu.cn, \{jingjingzhang, xwang11\}@fudan.edu.cn).

F. Zhu is with the Department of Electrical and Computer Engineering, North Carolina State University, Raleigh, NC, USA. (email: fzhu5@ncsu.edu).

W. Ni is with the School of Engineering, Edith Cowan University, Perth, WA 6027, Australia (e-mail: wei.ni@ieee.org).

(\textit{Bingnan Xiao} and \textit{Feng Zhu} are co-first authors.)
}

}

\maketitle

\begin{abstract}

Inherent client drifts caused by data heterogeneity, as well as vulnerability to Byzantine attacks within the system, hinder effective model training and convergence in federated learning (FL).
This paper presents two new frameworks, named DiveRgence-based Adaptive aGgregation (DRAG) and Byzantine-Resilient DRAG (BR-DRAG), to mitigate client drifts and resist attacks while expediting training. DRAG designs a reference direction and a metric named divergence of degree to quantify the deviation of local updates. Accordingly, each worker can align its local update via linear calibration without extra communication cost. 
BR-DRAG refines DRAG under Byzantine attacks by maintaining a vetted root dataset at the server to produce trusted reference directions. The workers' updates can be then calibrated to mitigate divergence caused by malicious attacks.
We analytically prove that DRAG and BR-DRAG achieve fast convergence for non-convex models under partial worker participation, data heterogeneity, and Byzantine attacks. 
Experiments validate the effectiveness of DRAG and its superior performance over state-of-the-art methods in handling client drifts, and highlight the robustness of BR-DRAG in maintaining resilience against data heterogeneity and diverse Byzantine attacks.

\end{abstract}


\section{Introduction}

With the growing complexity of machine learning (ML) tasks and exponential increase in data volumes, federated learning (FL) has received growing attention~\cite{pmlr-v54-mcmahan17a,9069945}. 
Meanwhile, FL with local stochastic gradient descent (SGD) introduces a series of challenges different from centralized learning.

One challenge is client drift due to the lack of persistent communication between workers and the parameter server (PS)~\cite{zhao2018federated}. 
Since the local data of workers can be imbalanced and heterogeneous, i.e., non-independent and identically distributed (non-IID), untimely synchronization of their local updates can cause their models to drift apart~\cite{zhao2018federated}. 
Naively averaging the models, e.g., using Federated Averaging (FedAvg)~\cite{pmlr-v54-mcmahan17a}, fails to mitigate such drifts, thwarting convergence.

Another challenge in FL stems from susceptibility to attacks launched by malicious workers due to its distributed nature. 
These attacks, commonly known as Byzantine attacks~\cite{so2020byzantine,10574838}, disrupt the global model training by tampering with the local data or model updates of the attacked workers, which seriously deteriorate training performance and system security.

\subsection{Related Work}

Client drift in FL was first reported in~\cite{zhao2018federated}, where data heterogeneity can cause local model updates to deviate systematically from the true global gradient direction, thereby introducing bias and impeding convergence.
Client drift was further analyzed in~\cite{li2019convergence,malinovskiy2020local}.
Some works incorporated variance reduction techniques~\cite{defazio2014saga} into local SGD, e.g.,~\cite{liang2019variance,li2019feddane,pathak2020fedsplit,mitra2021linear}. Yet, these methods require full worker participation, limiting their practicality when only partial workers are active.

Another mainstream strategy is to leverage control variates or explicit gradient constraints for update correction. In~\cite{karimireddy2020scaffold}, SCAFFOLD was designed with both local and global control variates to mitigate drifts at the workers. 
Similarly, FedDyn~\cite{acar2021federated} adopts first-order regularization terms on both the worker and PS sides.
MIME~\cite{karimireddy2020mime} utilizes momentum-based gradients to enhance the server-level optimization. As a method based on gradient constraints, FedProx was proposed in~\cite{li2020federated}, which introduced a proximal term into the local objective function to restrict the deviation of local updates from the global model. Built on FedProx, FedDC~\cite{gao2022feddc} decouples the local and global models, and integrates both control variates and regularization terms to jointly mitigate the drift between them.
FedSAM~\cite{pmlr-v162-qu22a} deploys the SAM optimizer~\cite{NEURIPS2023_e095c0a3}, which prompts workers to approximate local flat minima to enhance model generalization.
Unfortunately, these approaches incur additional complexity and memory overhead due to the stored control variates and incorporated regularization terms.

Recently, strategies, such as FedExP~\cite{jhunjhunwala2023fedexp} and FedACG~\cite{Kim_2024_CVPR}, have aimed to improve training convergence by refining the aggregation process. FedExP dynamically adjusts the stpdfize at the PS side with extrapolation on pseudo-gradients. FedACG stabilizes training by broadcasting a lookahead global gradient to align the local updates.
However, these methods rely on accurate estimation of global or pseudo-gradients, which can be noisy or biased in heterogeneous settings. 



On the other hand, FL with local SGD is vulnerable to Byzantine attacks~\cite{9947081}. Common Byzantine attacks include:
\begin{itemize}
    \item Noise Injection~\cite{9614992}: The malicious workers upload random vectors (e.g., sampled from a Gaussian distribution) instead of genuine updates, which corrupt the global model with unpredictable perturbations during training.
    
    \item Sign Flipping~\cite{li2019rsa}: For the malicious workers, the sign of their local updates is flipped, leading to gradient ascent rather than descent. This actively maximizes the loss and drives the global model away from convergence.
    
    \item Label Flipping~\cite{10054157}: This attack erodes local data by reversing each data label $l$ to $L - l - 1$, where $L$ is the total number of classes, producing semantically misleading updates while remaining stealthy and API-independent.
\end{itemize}

{\color{black}
    \begin{itemize}
    \item Min-Max~\cite{shejwalkar2021manipulating}: The malicious workers upload the same crafted update with an additional term
    $\gamma p^t$, where $\gamma$ is a carefully chosen coefficient to prevent the malicious update from being the farthest from any of the other (or benign) updates.
    This makes the attack difficult to detect through distance-based robust aggregation.

    \item Min-Sum~\cite{shejwalkar2021manipulating}: {\color{black}The malicious workers upload the crafted update consistent with Min-Max, with $\gamma$ chosen under a stricter, total-distance constraint relative to the benign workers compared to Min-Max.}
    This attack is often more effective than Min-Max against distance-based robust rules, as it remains geometrically similar to benign updates while biasing the aggregation.
    \end{itemize}
}

Recent works have designed Byzantine-robust aggregation strategies to remove outliers among the workers, i.e., comparing all local updates and suppressing the anomalies~\cite{blanchard2017machine,fang2020local,yin2018byzantine}.
With techniques, e.g., Krum~\cite{blanchard2017machine} and Trimmed Mean~\cite{fang2020local}, FL systems can withstand malicious attacks to some extent. In \cite{li2019rsa}, RSA was designed to penalize the difference between local and global model parameters to defend against Byzantine workers.
These methods lack robustness against Byzantine attacks when malicious workers become widespread. 

In \cite{cao2020fltrust}, FLTrust was proposed, where the PS leverages a root dataset to guide the global update direction by evaluating a cosine similarity-based trust score for each local update. 
Apart from FLTrust, other state-of-the-art trust score-based methods are FLEST \cite{geng2023better} and FLTG \cite{wen2025fltg}.
In \cite{9721118}, the geometric median~\cite{minsker2015geometric} was adopted for FedAvg due to its strong tolerance against a high proportion of malicious workers. In~\cite{9153949}, Byrd-SAGA was developed, which augments the SAGA method~\cite{defazio2014saga} with the geometric median to achieve Byzantine robustness. In~\cite{10100920}, BROADCAST was proposed, which integrates gradient difference compression with SAGA to mitigate compression and stochastic noise, enhancing resilience to Byzantine attacks.
However, these algorithms lack convergence guarantees in Byzantine FL systems with heterogeneous data, particularly for non-convex learning models. Recently, in \cite{zuo2024byzantine}, a convergence guarantee was established for the geometric median under data heterogeneity, albeit relying on the debatable assumption of bounded gradients.


\vspace{-3 mm}
\subsection{Contribution} 
This paper proposes two new aggregation frameworks, namely DiveRgence-based Adaptive aGgregation (DRAG) and Byzantine-Resilient DRAG (BR-DRAG), to tackle data heterogeneity and malicious Byzantine attacks, and accelerate model training in FL. The key contributions include:

\begin{itemize}
\item 
We interpret client drift as gradient misalignment and define the \textit{degree of divergence} (DoD) to quantify the misalignment between each worker’s update and a momentum-based global reference. 
We propose DRAG to steer local updates toward the reference direction to enhance the alignment in a decentralized manner.

\item 
We strengthen the resistance of DRAG against Byzantine attacks by producing trustworthy reference directions based on a small and vetted root dataset. By aligning and scaling the local updates with the reference direction, this BR-DRAG algorithm mitigates the adverse effect of malicious gradients while preserving their meaningful contribution to aggregations.


\item 
We analyze the convergence of DRAG and BR-DRAG in the non-convex case, accounting for data heterogeneity, partial worker participation, and Byzantine attacks. 
We reveal that, under both benign and adversarial settings, our algorithms attain the convergence rate of state-of-the-art methods operating without attacks. Also, BR-DRAG sustains convergence while relaxing the conventional assumption on the proportion of malicious workers.
\end{itemize}

     
    
    

We evaluate the performance of DRAG and BR-DRAG on three broadly adopted datasets: EMNIST, CIFAR-10, and CIFAR-100. 
DRAG effectively mitigates client drifts and accelerates convergence, and consistently outperforms state-of-the-art baselines under heterogeneity and partial worker participation. 
Under various Byzantine attacks, BR-DRAG shows its superb resilience to severely skewed data distributions and high proportions of malicious workers, compared to existing defense methods, e.g., FLTrust and geometric median. 

The rest of this paper is organized as follows. Section II describes the system model. DRAG is developed in Section III, followed by the design of BR-DRAG in Section IV. Section V analyzes the convergence of DRAG and BR-DRAG. Numerical results are provided in Section VI, followed by the conclusions in Section VII.

\textit{Notation}: Calligraphic letters represent sets; boldface lowercase letters indicate vectors; \(\mathbb{R}\) denotes the real number field; \(\mathbb{E}[\cdot]\) denotes expectation; \(\nabla\) stands for gradient. 
$\|\cdot\|_2$ is the $\ell_2$-norm; and $\langle\cdot, \cdot\rangle$ represents the inner product operation.

\section{System Model}
In this section, we elucidate the FL system and the corresponding setup under Byzantine attacks.

\vspace{-3 mm}
\subsection{Federated Learning Framework}


We study an FL system consisting of a PS and $M$ workers collected by the set $\mathcal{M}:=\{1,\cdots, M\}$; see Fig.~\ref{fig:sys model}. Each worker $m$ maintains a local dataset $\mathcal{D}_m$ with $N_m$ data samples. Let $\mathcal{D}=\bigcup_{m\in\mathcal{M}} \mathcal{D}_m$.
The objective of FL is to 
minimize the average sum of local functions of individual workers:
\begin{align} \label{problem}
    \min_{\boldsymbol{\theta}\in \mathbb{R}^d} f(\boldsymbol{\theta})&=\frac{1}{M}\sum_{m=1}^{M} F_m(\boldsymbol{\theta}), 
\end{align}
where $\boldsymbol{\theta} \in \mathbb{R}^d$ is the $d$-dimensional model parameter, and $F_m(\boldsymbol{\theta}) \buildrel \Delta \over =  \mathbb{E}_{z_m \sim \mathcal{D}_m}\left[F_m(\boldsymbol{\theta}; z_m)\right]$ is the smooth local loss function concerning the dataset $\mathcal{D}_m$ of worker $m$.

\begin{figure}[!t]
	\centering
	\includegraphics[width=0.9\linewidth]{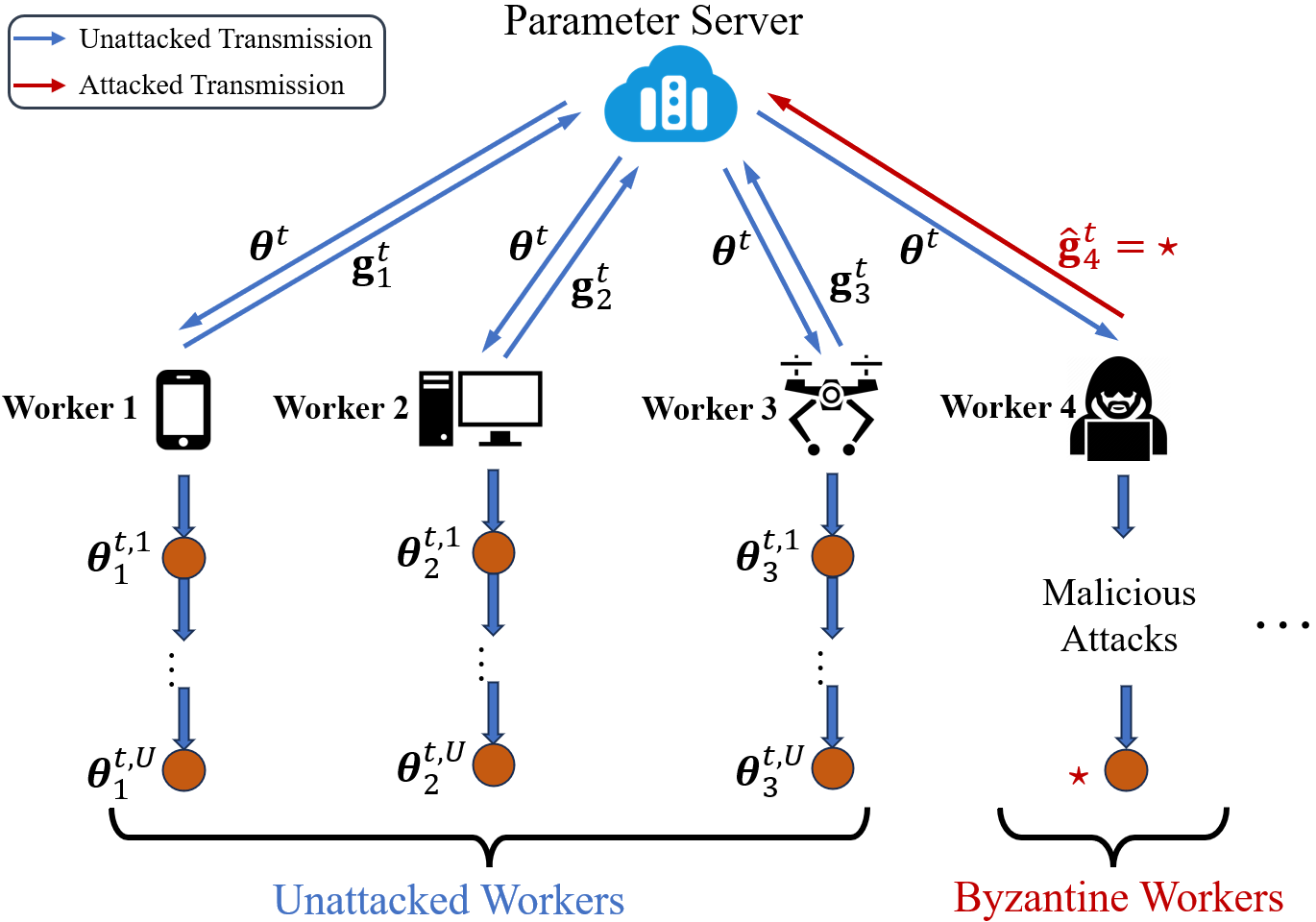}
	\caption{The architecture of a FL system with $S$ selected workers in round~$t$. For Byzantine FL systems, the attacked workers can upload arbitrary local updates to affect global model aggregation.}
	\label{fig:sys model}
\end{figure} 


With the objective function in (\ref{problem}), we attempt to solve it iteratively using local SGD. In any training round $t$, the following stpdf are executed:
\begin{itemize}
    \item Worker selection and model transmission: The PS first selects a subset $\mathcal{S}^t \subset \mathcal{M}$ with $S$ workers uniformly at random (UAR) without replacement, and broadcasts the global model $\boldsymbol{\theta}^t$ to the selected workers $m \in \mathcal{S}^t$.

    \item Local update: With $\boldsymbol{\theta}_m^{t,0}=\boldsymbol{\theta}^t$, each selected worker $m\in \mathcal{S}^t$ performs $U$ local updates via SGD updating rule:
    \begin{align} \label{update}
    \boldsymbol{\theta}_{m}^{t,u+1} \!\!=\! \boldsymbol{\theta}_{m}^{t,u} \!\!-\! \eta \nabla \!F_m(\boldsymbol{\theta}_{m}^{t,u}; z_{m}^{t,u}), \ u \!=\! 0,1,\!\cdots\!,U \!\!-\!\! 1,
    \end{align}
    where $\eta$ denotes the stpdfize, and $\nabla F_m(\boldsymbol{\theta}_{m}^{t,u}; z_{m}^{t,u})$ is the stochastic gradient of worker $m$ with a mini-batch $z_{m}^{t,u}$ sampled from the local dataset $\mathcal{D}_m$ with size $|z_{m}^{t,u}|=B$.

    \item Model uploading and aggregation: After local computation, each worker $m \in \mathcal{S}^t$ sends $\mathbf{g}_m^t = \boldsymbol{\theta}_m^{t, U} - \boldsymbol{\theta}^{t}$ to the PS, which represents the discrepancy between the latest local model after $U$ local updates and the current global model received at the beginning of the $t$-th training round. 
    The PS updates the global model $\boldsymbol{\theta}^{t+1}$ as 
    \begin{align} \label{aggregation}
  \boldsymbol{\theta}^{t+1}=\boldsymbol{\theta}^{t}+\frac{1}{S}\sum_{m\in\mathcal{S}^t}\mathbf{g}_m^t. 
    \end{align}    
\end{itemize}
Then, the next training round, i.e., the $(t+1)$-th round, starts. This training process repeats until convergence. 


\subsection{Byzantine Attack}



Due to the decentralized implementation of FL, workers are vulnerable to Byzantine attacks arising from, e.g., poisoned data.
These workers, also known as Byzantine nodes, upload malicious updates to the PS in an attempt to degrade the convergence and learning performance.

Consider $A$ Byzantine nodes out of the $M$ workers. Let $\mathcal{A}$ collect all Byzantine nodes. $\mathcal{A}\subseteq \mathcal{M}$. In the $t$-th training round, $A^t$ malicious workers, collected by $\mathcal{A}^t$, are selected, with $\mathcal{A}^t \subseteq \mathcal{S}^t$ and $\mathcal{A}^t \subseteq \mathcal{A}$. Let $w^t \in [0,1]$ be the intensity level of the Byzantine attack: $A^t = w^t S$. 

To launch Byzantine attacks, any malicious worker $m \in \mathcal{A}^t$ sends 
its malicious update $\hat{\mathbf{g}}_m^t \in \mathbb{R}^d$ 
to the PS. 
The PS aggregates the local models to update the global model:
\begin{align}
    \boldsymbol{\theta}^{t+1}=\boldsymbol{\theta}^{t}+\frac{1}{S}\big(\sum_{m\in\mathcal{A}^t}\hat{\mathbf{g}}_m^t+\sum_{m\in \mathcal{B}^t }\mathbf{g}_m^t\big), \label{byzantine aggregation}
\end{align}
where $\mathcal{B}^t = \mathcal{S}^t \setminus \mathcal{A}^t$ collects the benign workers in round~$t$.




\section{DiveRgence-based Adaptive aGgregation}

In this section, we propose the DRAG algorithm, which overcomes client drifts (i.e., the deviation of local model updates from the true global gradient direction, as described in Section I-A) and accelerates convergence. We start with the update flow of DRAG and the overall algorithm.

\subsection{Algorithm Overview}


\textbf{Algorithm~\ref{Alg_DRAG}} presents the implementation details of DRAG based on a new divergence-based aggregation strategy proposed in Section III-C.
At the beginning of each training round~$t$, after the worker subset $\mathcal{S}^t$ is specified, the PS determines the \textit{reference direction} $\mathbf{r}^t$ (as proposed in Section III-B),  with \eqref{reference}, to measure the degree of local divergence and broadcasts the global model $\boldsymbol{\theta}^t$ and the reference direction $\mathbf{r}^t$ to the selected workers; see Stpdf 4--12. 

On the local side, the workers perform local updates to generate $\mathbf{g}_m^t$, and compute the DoD $\lambda_m^t$, i.e., using \eqref{dod}, to measure the deviation from the reference direction $\mathbf{r}^t$; see Stpdf 14 and 15. To cope with data heterogeneity, the workers generate the modified gradient $\mathbf{v}_m^t$ by adaptively ``dragging" their local updates toward $\mathbf{r}^t$; see Step 16. After local computation, the selected workers upload $\mathbf{v}_m^t,\,\forall m \in \mathcal{S}^t$ to the PS for global update. $\mathbf{r}^t$ and $\Delta^t$ are retained to construct the reference direction for the next round; see Step 18. 

Unlike existing FL systems tackling client drifts by employing local control variates or gradient norms~\cite{liang2019variance,li2019feddane,pathak2020fedsplit,mitra2021linear,karimireddy2020scaffold,acar2021federated,karimireddy2020mime,li2020federated,gao2022feddc,pmlr-v162-qu22a}, we design a more effective strategy that guides each local model toward the reference direction through vector manipulation, which enforces directional consistency across clients.

\begin{algorithm}[!t]
    \caption{DRAG}
    \label{Alg_DRAG}
    \begin{algorithmic}[1]
        \STATE \textbf{Input:} $\boldsymbol{\theta}^0$, $\mathcal{M}$, $S$, $T$, $U$, $\eta$, $\alpha$, $c$;
        \FOR{$t=0,1,\dots,T-1$}
            \STATE The PS selects the worker subset $\mathcal{S}^t \subseteq \mathcal{M}$ UAR;
            \IF{$t=0$}
                \FOR{each worker $m \in \mathcal{S}^t$ in parallel}
                    \STATE Perform $U$ SGD updates by \eqref{update}, and upload $\mathbf{g}_m^t$;
                \ENDFOR
                \STATE The PS obtains the reference direction $\mathbf{r}^t$ with \eqref{reference_a};
            \ELSE
                \STATE The PS obtains $\mathbf{r}^t$ with \eqref{reference_b};
            \ENDIF
            \STATE The PS sends $\boldsymbol{\theta}^t$ and $\mathbf{r}^t$ to workers $m \in \mathcal{S}^t$
            \FOR{each worker $m \in \mathcal{S}^t$ in parallel}
                \STATE Perform $U$ local SGD updates with \eqref{update};
                \STATE Compute the DoD $\lambda_m^t$ with \eqref{dod};
                \STATE Compute the modified gradient $\mathbf{v}_m^t$ with \eqref{vm}, and upload  $\mathbf{v}_m^t$ to the PS;
            \ENDFOR
            \STATE The PS aggregates $\Delta^t$ with \eqref{drag aggregation}, retains $\mathbf{r}^t$ and $\Delta^t$, and updates $\boldsymbol{\theta}^t$ to $\boldsymbol{\theta}^{t+1}$ with \eqref{drag global update};
        \ENDFOR
    \end{algorithmic}
\end{algorithm}



\subsection{Global Reference Direction}
As mentioned in Section I, one fundamental challenge in FL is data heterogeneity, as 
the local gradients $\mathbf{g}_m^t,\, m\in \mathcal{S}^t$ tend to diverge in direction and magnitude due to discrepancies in the local objective $F_m$, leading to client drifts and global gradient deviation. 
Utilizing $\tilde{\Delta}^t=\frac{1}{S}\sum_{m\in\mathcal{S}^t}\mathbf{g}_m^t$ as the global update direction (i.e., using FedAvg, as in conventional FL algorithms) could cause unstable or biased convergence~\cite{li2019convergence}.

To address client drifts, we design a global reference direction $\mathbf{r}^t$ produced by the PS at the beginning of each training round $t$, which 
helps align the local gradients. At the beginning of training round~$t$, $\mathbf{r}^t$ is constructed as 
\begin{subequations}\label{reference}
\begin{empheq}[left={\mathbf{r}^t = \left\{
\begin{aligned}}, right={\end{aligned}\right.}]{align}
& \frac{1}{S}\sum\limits_{m \in \mathcal{S}^t} \mathbf{g}_m^t, & \text{for } t = 0 ;  \label{reference_a} \\
&  (1 - \alpha)\mathbf{r}^{t-1} + \alpha\Delta^{t-1}, & \text{for } t \geq 1 , \label{reference_b}
\end{empheq}
\end{subequations}
where $\Delta^{t-1} = \frac{1}{S}\sum_{m \in \mathcal{S}^{t-1}} \mathbf{v}_m^{t-1} $ is the aggregated modified gradient in round $t-1$, and 
$\alpha \in (0,1)$ is a hyperparameter controlling the weights of historical reference directions.
$\mathbf{r}^t,\,t\geq 1$, evolves recursively as a weighted combination of the previous reference direction $\mathbf{r}^{t-1}$ and the latest global update~$\Delta^{t-1}$.

During round $t$, the selected devices generate the modified gradients $\mathbf{v}_m^t$ based on the reference direction $\mathbf{r}^t$ (see Section~III-C). $\mathbf{v}_m^t,\! \forall m \in \mathcal{S}^t$, are aggregated at the PS, given as
\begin{align} \label{drag aggregation}
    \Delta^t = \frac{1}{S}\sum_{m\in\mathcal{S}^t}\mathbf{v}_m^t.
\end{align}
The global model is updated with $ \Delta^t$, as given by
\begin{align} \label{drag global update}
    \boldsymbol{\theta}^{t+1} = \boldsymbol{\theta}^{t} + \Delta^t.
\end{align}

We further derive the closed-form expression of $\mathbf{r}^t$ by deduction, as follows: $ \forall t \geq 1$,
\begin{align} \label{reference_cf}
    \mathbf{r}^t=\frac{(1-\alpha)^t}{S}\sum_{m\in\mathcal{S}^0}\mathbf{g}_m^0+\sum_{i=0}^{t-1}\alpha(1-\alpha)^{t-i-1}\Delta^i. 
\end{align}
Here, $\mathbf{r}^t$ is an exponential moving average over all historical modified global updates~$\Delta^i,\,i=1,\cdots,t-1$, with exponentially decaying weights $(1-\alpha)^{t-i-1}$. This momentum-style construction enables $\mathbf{r}^t$ to preserve long-term memory of update history while prioritizing those more recent. As $\alpha \to 1$, $\mathbf{r}^t$ reduces to the most recent update $\Delta^{t-1}$. As $\alpha \to 0$, a smoother reference direction integrates past variations more conservatively. By tuning~$\alpha$, $\mathbf{r}^t,\,\forall t$, can flexibly balance adaptivity and stability under different data heterogeneity levels. 

\subsection{Divergence-based Local Gradient Modification}
With the reference direction $\mathbf{r}^t$, we design the adaptive modified gradients $\mathbf{v}_m^t, \forall m \in \mathcal{S}^t$ per training round $t$ to combat client drifts.
We first specify the new metric, DoD,
to measure the extent to which the local update $\mathbf{g}_m^t$ of each worker $m$ in round $t$ diverges from $\mathbf{r}^t$. 
The angle $\angle_{m}^t$ is introduced to capture the directional deviation of $\mathbf{g}_m^t$ from $\mathbf{r}^t$ based on the normalized cosine similarity, i.e.,  
\begin{align}
    \angle_{m}^t=\arccos{\frac{\left\langle\mathbf{g}_m^t, \mathbf{r}^t\right\rangle}{\|\mathbf{g}_m^t\|\cdot\|\mathbf{r}^t\|}},
\end{align}
where $\angle_{m}^t = 0$ indicates perfect alignment, and $\angle_{m}^t = \pi$ denotes complete opposition. 
By dividing $\angle_{m}^t$ over $\frac{\pi}{2}$ and approximating the $\arccos$ function with a linear function (e.g., $y=-\frac{\pi}{2}x+\frac{\pi}{2}$), the DoD, $\lambda_m^t$, is defined as 
\begin{align}
    \lambda_m^t=:c\big(1-\frac{\left\langle\mathbf{g}_m^t, \mathbf{r}^t\right\rangle}{\|\mathbf{g}_m^t\|\cdot \|\mathbf{r}^t\|}\big)\in[0,2c],\label{dod}
\end{align}
where $c \in [0, 1]$ is a tunable hyperparameter modulating the contribution of the divergence signal. 
By adjusting $c$, one can control the extent to which local updates are aligned with $\mathbf{r}^t$.
Moreover, $\lambda_m^t$ quantifies the geometric inconsistency between $\mathbf{g}_m^t$ and $\mathbf{r}^t$ on the unit hypersphere, offering a continuous and differentiable measure of gradient divergence, as opposed to discrete thresholds or binary decisions. 



With the reference direction $\mathbf{r}^t$ and the DoD~$\lambda_m^t$, we establish the modified gradient $\mathbf{v}_m^t$, which ``drags" the local gradients $\mathbf{g}_m^t$ towards $\mathbf{r}^t$ based on $\lambda_m^t$,
as given by
\begin{align}\label{vm}
    \mathbf{v}_m^t=(1-\lambda_m^t)\mathbf{g}_m^t+\lambda_m^t\frac{\|\mathbf{g}_m^t\|}{\|\mathbf{r}^t\|}\mathbf{r}^t, 
\end{align}
If there is significant misalignment between $\mathbf{r}^t$ and $\mathbf{g}_m^t$, the local update is dragged toward the normalized reference direction $\frac{\|\mathbf{g}_m^t\|}{\|\mathbf{r}^t\|}\mathbf{r}^t$, reducing client drifts while preserving the diversity of local information.

Fig.~\ref{vectorm} illustrates the gradient modification in DRAG based on the proposed reference direction $\mathbf{r}^t$ and DoD $\lambda_m^t$. For moderate deviation ($0<\lambda_m^t\le1$) in Fig.~\ref{vectorm}(a), the modification projects $\mathbf{g}_m^t$ toward $\mathbf{r}^t$ and enlarges the aligned component to at least $\frac{\|\mathbf{g}_m^t\|}{\|\mathbf{r}^t\|}\mathbf{r}^t$, reducing client drifts.
For severe deviation ($1<\lambda_m^t\le2$) in Fig.~\ref{vectorm}(b), the local gradient diverges in the opposite direction to $\mathbf{r}^t$. We reverse $\mathbf{g}_m^t$ according to (\ref{vm}) to ensure adherence to the correct update direction.
Unlike conventional drift-mitigation methods~\cite{karimireddy2020scaffold,acar2021federated}, which rely on local control variates, DRAG aligns local gradients with $\mathbf{r}^t$ via \eqref{vm}, eliminating the need for extra control variables and incurring negligible memory and communication overhead.
Moreover, methods employing local gradient-norm regularizers, e.g.,\cite{li2020federated}, are ineffective under heterogeneous data distributions, since local gradients may deviate substantially from the global gradient. In contrast, DRAG attains closer alignment with the global gradient by steering local updates toward $\mathbf{r}^t$ using the DoD metric $\lambda_m^t$ in~\eqref{dod}, accelerating convergence even under severe heterogeneity.

\begin{figure}[t!]
\centering
\subfloat[$0<\lambda_m^t\leq 1$]{\includegraphics[width=1.485in]{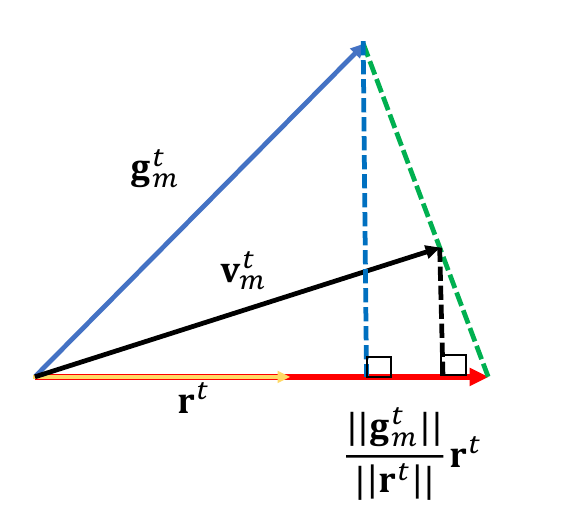}}
\subfloat[$1<\lambda_m^t$]{\includegraphics[width=1.945in]{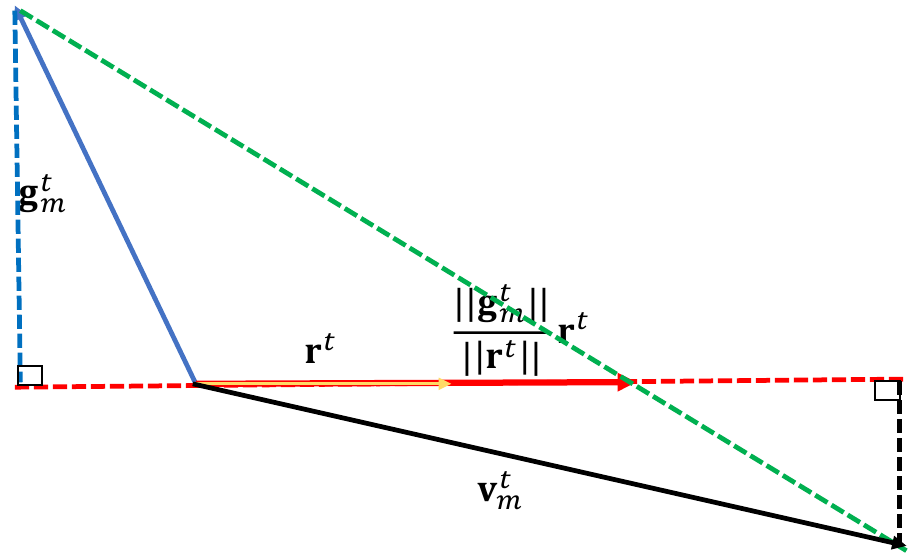}}
\caption{Illustration of vector modification of DRAG. For workers with different DoD values, the modified gradient $\mathbf{v}_m^t$ (solid black line) has a larger component on the reference direction $\mathbf{r}^{t}$ (solid yellow line) than the raw local update $\mathbf{g}_{m}^t$ (solid black line).}
\label{vectorm}
\end{figure}

\section{Byzantine-Resilient DiveRgence-based Adaptive aGgregation}

This section enhances DRAG with resistance to Byzantine attacks, and develops BR-DRAG;
see \textbf{Algorithm~\ref{Alg_BRDRAG}}.

\subsection{Algorithm Overview}


Unlike DRAG, BR-DRAG maintains a root dataset $\mathcal{D}_{\mathrm{root}}$  for the PS to generate a trusted reference direction before model training. 
At the beginning of each training round~$t$, the worker subset $\mathcal{S}^t$ is specified. The PS computes the current reference direction $\mathbf{r}^t$ based on $\mathcal{D}_{\mathrm{root}}$ (as opposed to aggregating the local updates $\mathbf{v}_m^t,\,\forall m\in \mathcal{S}^t$, from the selected workers, as in DRAG), and sends $\mathbf{r}^t$ and the current global model $\boldsymbol{\theta}^t$ to the selected workers; see Step 4.

At the selected workers, either malicious or benign, SGD is performed to generate the local updates $\mathbf{g}_m^t,\, \forall m\in \mathcal{S}^t$; see Step~6. 
The selected workers upload $\mathbf{v}_m^t$ to the PS for global update. 
Based on the DoD, $\lambda_m^t$, in \eqref{dod new}, the PS compute the modified gradients $\mathbf{v}_m^t$ with \eqref{vm1} to counteract malicious attacks; see Step 8.
The updated global model $\boldsymbol{\theta}^{t+1}$ is used to start the next round; see Step~9.
Compared to DRAG, BR-DRAG updates both the reference direction $\mathbf{r}^t$ and the module normalization operation to tolerate malicious attacks and ensure system convergence.

\begin{algorithm}[!t]
    \caption{BR-DRAG}
    \label{Alg_BRDRAG}
    \begin{algorithmic}[1]
        \STATE \textbf{Input:} $\boldsymbol{\theta}^0$, $\mathcal{M}$, $S$, $T$, $U$, $\eta$, $\alpha$, $c^t$, $\mathcal{D}_{\mathrm{root}}$;
        \FOR{$t=0,\dots,T-1$}
            \STATE The PS selects $\mathcal{S}^t \subseteq \mathcal{M}$ UAR;
            \STATE The PS obtains the reference direction $\mathbf{r}^t$ with \eqref{trust}, and sends $\boldsymbol{\theta}^t$ and $\mathbf{r}^t$ to workers $m \in \mathcal{S}^t$;
            \FOR{each worker $m \in \mathcal{S}^t$ in parallel}
                \STATE Perform $U$ local SGD updates with \eqref{update}, and upload $\mathbf{g}_m^t$ to the PS;
            \ENDFOR
            \STATE $\forall m \in \mathcal{S}^t$, the PS computes the DoD $\lambda_m^t$ and modified gradient $\mathbf{v}_m^t$ with \eqref{dod new} and \eqref{vm1};
            \STATE The PS updates $\boldsymbol{\theta}^t$ to $\boldsymbol{\theta}^{t+1}$ with \eqref{byzantine drag global update};
        \ENDFOR
    \end{algorithmic}
\end{algorithm}

\subsection{Byzantine-Resilient Global Reference Direction}

To defend against the Byzantine attacks, the key improvements of BR-DRAG over DRAG are the selection of $\mathbf{r}^t$ and the modified gradient $\mathbf{v}_m^t$.
Moreover, in the presence of malicious local updates (from the malicious workers), $\mathbf{r}^t$ needs to be constructed in a trustworthy way under BR-DRAG, as compared to DRAG.
{\color{black} As in FLTrust \cite{cao2020fltrust}, BR-DRAG maintains the root dataset $\mathcal{D}_{\mathrm{root}}$ to compute $\mathbf{r}^t$. Following the same construction principles as FLTrust, $\mathcal{D}_{\text{root}}$ is drawn from a vetted pool with de-duplication and outlier filtering, and its size (e.g., 3,000 samples in the default setting) is consistent with the typical choices in the literature, e.g.,~\cite{cao2020fltrust}. In this sense, BR-DRAG does not impose additional requirements on the root dataset beyond those adopted by FLTrust and other recent root-dataset-based methods, e.g., FLEST \cite{geng2023better} and      
  FLTG~\cite{wen2025fltg}. }

{\color{black}
Notably, the root dataset $\mathcal{D}_{\mathrm{root}}$ is not a centralized collection of the clients' private local data. It is a trusted seed dataset with a small size separately maintained at the PS, as in existing Byzantine-robust FL designs, e.g., FLTrust~\cite{cao2020fltrust} and FLEST~\cite{geng2023better}.
$\mathcal{D}_{\mathrm{root}}$ can be obtained in several ways, e.g., from a small set of task-relevant samples already available at the PS~\cite{mai2023server}, institution-owned data, or a small set of samples collected with explicit consent and verified manually~\cite{cao2020fltrust}.
$\mathcal{D}_{\mathrm{root}}$ can also be periodically audited and refreshed with standard curation procedures, e.g., through manual inspection, de-duplication, suspicious sample removal, and label-quality checking~\cite{northcutt2021confidentlearning}.
Existing data-cleaning techniques can be incorporated to improve the reliability of the trusted dataset.
}


During training round $t$, based on $\mathcal{D}_{\text{root}}$, the global model $\boldsymbol{\theta}^t$ is updated for $U$ SGD iterations to obtain $\boldsymbol{\theta}^{t,U}$: 
\begin{equation}\label{update_root_new}
\boldsymbol{\theta}^{t,u + 1} = \boldsymbol{\theta}^{t,u} - \eta \nabla f(\boldsymbol{\theta}^{t,u}; z^{t,u}), \ \ u = 0,1,\ldots,U - 1,
\end{equation}
where $\boldsymbol{\theta}^{t,0}=\boldsymbol{\theta}^{t}$, and $z^{t,u}$ is drawn independently from $\mathcal{D}_{\text{root}}$ across all batches, local iterations and training rounds.

With $\boldsymbol{\theta}^{t,U}$ obtained from the trusted root dataset $\mathcal{D}_{\text{root}}$, 
the reference direction $\mathbf{r}^t$ in BR-DRAG  is designed as 
\begin{align} \label{trust}
\mathbf{r}^t = \boldsymbol{\theta}^{t,U}-\boldsymbol{\theta}^{t} = - \eta \sum_{u=0}^{U-1} \nabla f \left(\boldsymbol{\theta}^{t, u} ; z^{t, u}\right) ,
\end{align}
which, unlike \eqref{reference} in DRAG, stems directly from $\mathcal{D}_{\text{root}}$, providing a trustworthy update direction of the global objective and thus mitigating the impact of Byzantine behaviors. 

The PS aggregates the local updates of the workers as in \eqref{byzantine aggregation}. For each uploaded local update, the PS generates the modified gradient $\mathbf{v}_m^t,\,\forall m\in \mathcal{S}^t$, to resist Byzantine attacks (as will be designed in Section IV-C). The global model is updated as 
\begin{align} \label{byzantine drag global update}
    \boldsymbol{\theta}^{t+1}=\boldsymbol{\theta}^t+\Delta^t=\boldsymbol{\theta}^t+ \frac{1}{S}\big(\sum_{m\in\mathcal{A}^t}\hat{\mathbf{v}}_m^t+\sum_{m\in \mathcal{B}^t }\mathbf{v}_m^t\big),
\end{align}
where ${\mathbf{v}}_m^t$ stands for the benign modified gradient (as in DRAG), while $\hat{\mathbf{v}}_m^t$ denotes the malicious modified gradients.


\subsection{Byzantine-Resilient Local Gradient Modification}
In FL systems subjected to Byzantine attacks, malicious workers upload distorted gradients resulting from poisoned data. The vector modification in \eqref{vm} is inapplicable since $\mathbf{v}_m^t$ can be distorted by the unbounded term $\left(1-\lambda_m^t\right) \mathbf{g}_m^t$.
To alleviate the impact of malicious local updates on the global model, we design the modified gradient as \begin{align}\label{vm1}
    \mathbf{v}_m^t:=&(1-\lambda_m^t)\frac{\|\mathbf{r}^t\|}{\|\mathbf{g}_m^t\|}\mathbf{g}_m^t+\lambda_m^t\mathbf{r}^t ,\\
\label{dod new}
\text{with }  \quad  \lambda_m^t:=&c^t\big(1-\frac{\left\langle\mathbf{g}_m^t, \mathbf{r}^t\right\rangle}{\|\mathbf{g}_m^t\|\cdot \|\mathbf{r}^t\|}\big)\in[0,2c^t],
\end{align}
where $c^t \in [0,1], \forall t$ is set to be adjustable across training rounds, as described later in Section V-B.

Unlike DRAG, which modifies the local updates by scaling the reference direction $\mathbf{r}^t$ to match the norm of $\mathbf{g}_m^t$, i.e.,  \eqref{vm}, 
BR-DRAG adopts a different scaling strategy that normalizes $\mathbf{g}_m^t$ to match $\| \mathbf{r}^t \|$; see \eqref{vm1}.
This modification is essential under Byzantine attacks, where the attackers may inflate the norm of their local updates in an attempt to dominate the aggregation process. By scaling $\mathbf{g}_m^t$ to $\|\mathbf{r}^t\|$, BR-DRAG mitigates the impact of malicious norm.
Moreover, with the DoD $\lambda_m^t$, BR-DRAG retains the advantages of divergence-based alignment from DRAG, providing directional correction under data heterogeneity.
With the reference direction $\mathbf{r}^t$ obtained from a trusted dataset and indicating a general update direction for the training process, as well as the module normalization operation in \eqref{vm1}, BR-DRAG can effectively combat the malicious local updates $\hat{\mathbf{g}}_m^t$.

    {\color{black}
  The state-of-the-art root-dataset-based methods, such as FLTrust \cite{cao2020fltrust}, FLEST \cite{geng2023better}, and FLTG~\cite{wen2025fltg}, all operate under a scalar weighting paradigm based on trust scores. 
  During the model aggregation, every local update received is assigned a scalar trust weight, 
  and its magnitude is then normalized to that of the PS's reference for weighted model averaging. 
With the alignment $x_m^t = \langle\mathbf{g}_m^t, \mathbf{r}^t\rangle/(\|\mathbf{g}_m^t\|\|\mathbf{r}^t\|)$ between
   the local update $\mathbf{g}_m^t$ and the trusted reference $\mathbf{r}^t$, in {FLTrust} \cite{cao2020fltrust}, the server computes the ReLU-clipped cosine similarity $\max(0, x_m^t)$ as the 
  trust score for each local update $\mathbf{g}_m^t$, normalizes $\|\mathbf{g}_m^t\|$ to $\|\mathbf{r}^t\|$, and aggregates via trust-score-weighted averaging.
      In {FLEST} \cite{geng2023better}, the server synthesizes a trust score by combining FLTrust-style cosine similarity
   with a $K$-means-based anomaly detection score via a dynamic ratio, and performs weighted averaging with the synthesized weights. 
In {FLTG}~\cite{wen2025fltg}, the server filters clients using ReLU-clipped cosine similarity, selects a dynamic   
  reference client to mitigate non-IID bias, assigns inversely angle-proportional weights, and aggregates via weighted averaging.  
Unlike these existing methods, BR-DRAG conducts divergence correction. It computes DoD
  $\lambda_m^t = c^t(1 - x_m^t)$ and corrects the local divergence before aggregation; see (15).
  This allows BR-DRAG to retain the directionally aligned part of the local update while suppressing its misaligned 
  part. 

Under data heterogeneity, benign local updates can deviate from $\mathbf{r}^t$ due to non-IID data, instead of malicious  
  intent. Trust-score-based methods (e.g., FLTrust, FLEST, and FLTG) treat the deviations with mechanisms, such as  
  simply dropping the updates (e.g., $\max(0, x_m^t)=\max(0,\frac{\langle \mathbf{g}_m^t,\mathbf{r}^t}{\|\mathbf{g}_m^t\|\cdot\|\mathbf{r}^t\|})$) [30, eq.(2)], which can aggressively suppress benign but heterogeneous updates. 
  In contrast, BR-DRAG is robust to data heterogeneity, inherited from DRAG. BR-DRAG adaptively corrects the update direction through vector interpolation, preserving the useful information carried by the  
  benign update while steering it toward the global trend, and thereby achieving tighter alignment with the global objective.
    
Under Byzantine attacks with arbitrary gradient norms, the trust-score-based methods (e.g., FLTrust, FLEST, and FLTG) that normalize magnitudes independently of the reference direction $\mathbf{r}^t$ can still allow malicious directional influence to persist. By contrast, BR-DRAG is robust to Byzantine attacks, coming from the trusted reference direction $\mathbf{r}^t$ generated from $\mathcal{D}_{\mathrm{root}}$ and the attack-aware normalization/correction in (15).  BR-DRAG's vector correction, combined with the norm normalization anchored to $\|\mathbf{r}^t\|$, constrains both the direction and magnitude of each received update.     }

\section{Convergence Analysis of DRAG and BR-DRAG}
In this section, we establish the convergence rate of DRAG and BR-DRAG for  prevalent non-convex objective functions.
We start with the following standard assumptions.

\smallskip
\smallskip

\noindent\textbf{Assumption 1. (Smoothness and lower boundedness)}  \textit{Each local objective 
function $ F_m(\boldsymbol{\theta})$ is $L$-Lipschitz smooth ($L>0$):
\begin{align}
    &\left\|\nabla F_m(\boldsymbol{\theta}_1)-\nabla F_m(\boldsymbol{\theta}_2)\right\| \leq L \left\|\boldsymbol{\theta}_1-\boldsymbol{\theta}_2\right\|, \, \forall m \in \mathcal{M}.  \label{lipschitz_gradient}
\end{align}
Also, assume the objective function $ f$ is lower-bounded by $ f^*$.}

\smallskip
\smallskip

\noindent \textbf{Assumption 2. (Unbiasedness and bounded variance)} \textit{For each client $m$, the local gradient estimator $\nabla F_m(\boldsymbol{\theta};z)$ is unbiased, i.e.,
\begin{align}
    \mathbb{E}\left[\nabla F_m(\boldsymbol{\theta};z)\right]=\nabla F_m(\boldsymbol{\theta}), \quad \forall m \in \mathcal{M}.
\end{align}
For both the variance of the local gradient estimator and that of the local gradient from the global one, there also exist two constants $\sigma_L^2$ and $\sigma_G^2$, such that
\begin{align}
    \mathbb{E}\left[\|\nabla F_m(\boldsymbol{\theta};z)-\nabla F_m(\boldsymbol{\theta})\|^2\right]&\leq \sigma_L^2, \quad \forall m \in \mathcal{M};\\
    \mathbb{E}\left[\|\nabla F_m(\boldsymbol{\theta})-\nabla f(\boldsymbol{\theta})\|^2\right]&\leq \sigma_G^2, \quad \forall m \in \mathcal{M}.
\end{align}}

\smallskip
\smallskip

\textbf{Assumptions 1} and \textbf{2} are standard assumptions considered in non-convex optimization and FL convergence analysis~\cite{liang2019variance,li2019feddane,pathak2020fedsplit,mitra2021linear,karimireddy2020scaffold,acar2021federated,karimireddy2020mime,li2020federated,gao2022feddc,pmlr-v162-qu22a}.
Here, $\sigma_G^2$ quantifies the degree of data heterogeneity, and $\sigma_G^2=0$ indicates IID local data among the workers.

\subsection{Convergence Analysis of DRAG}

Under \textbf{Assumptions 1} and \textbf{2}, we establish the upper bound for the expectation of the average squared gradient norm:
\begin{theorem} \label{drag theorem}
Suppose that $ F_m, \forall m \in \mathcal{M}$ is non-convex, and the stpdfize is $\eta\leq \frac{1}{8LU}$. Given an initial model $\boldsymbol{\theta}^0$, after $T$ training rounds, DRAG guarantees
\begin{align} \label{drag convergence eq}
    \frac{1}{T}\sum_{t=0}^{T-1} \mathbb{E}\left[\|\nabla f(\boldsymbol{\theta}^t)\|^2\right] \leq \frac{f(\boldsymbol{\theta}^0)-f^*}{\gamma \eta U T} + V,
\end{align}
where $V=\frac{1}{\gamma\eta U} \big(  c \eta (2\sigma_L^2+9 \sigma_G^2)+\frac{\eta^2 U^2 L}{2}\left(\sigma_L^2+3 U \sigma_G^2\right) + (\frac{\eta^3 U^2 L^2(85 c+5)}{2}+\frac{15 \eta^4 U^3 L^3}{2})(\sigma_L^2+6 U \sigma_G^2) \big)  $, and $\gamma > 0 $.
\end{theorem}

\begin{proof}
    See Appendix A.
\end{proof}

As revealed in \textbf{Theorem}~\textbf{\ref{drag theorem}}, DRAG achieves fast convergence for non-convex models under data heterogeneity and partial worker participation.
The non-vanishing term $V$ on the RHS of \eqref{drag convergence eq} characterizes the effects of both stochastic gradient variance ($\sigma_L^2$) and data heterogeneity ($\sigma_G^2$) across workers.
Meanwhile, the coefficients of $\sigma_L^2$ and $\sigma_G^2$  increase with the number of local updates, $U$.
To prevent these terms from growing unboundedly as $U$ increases, the stpdfize can be set to $\eta = \mathcal{O}(\frac{1}{U})$ to ensure bounded variance accumulation.

In DRAG, the choice of the hyperparameter $c$ in the DoD metric $\lambda_m^t$ exhibits a trade-off, as setting $c$ too small or too large can adversely impact system performance.
A larger $c$ increases $\lambda_m^t$ in $\mathbf{v}_m^t$, leading to stronger alignment with the reference direction $\mathbf{r}^t$, reducing client drifts, as discussed in Section III-C.
Several terms in $V$, e.g., $c\eta(2\sigma_L^2+9\sigma_G^2)$, increase linearly with $c$, reflecting the susceptibility of $V$ to gradient variance under overly aggressive correction.




\subsection{Convergence Analysis of BR-DRAG}
Next, we prove the convergence of BR-DRAG. The following assumptions,  characterizing the local update properties of the benign workers, are established.

\smallskip
\smallskip

\noindent\textbf{Assumption 3. (Gradient boundedness)} \textit{For the benign workers, the ratio between the norms of the global reference direction $\mathbf{r}^t$ and the local gradient $\mathbf{g}_m^t$ is bounded, i.e.,
\begin{align}
    \rho_m^t = \frac{\left\|\mathbf{r}^t\right\|}{\left\|\mathbf{g}_m^t\right\|} \in [p ,q], \quad \forall m \in  \mathcal{B}^t.
\end{align}
where $0 < p < q$ are constants.}

\smallskip
\smallskip

\noindent\textbf{Assumption 4. (Bounded divergence)} \textit{For the benign workers, define the cosine similarities $ y_m^t = \frac{\langle\mathbf{g}_m^t, \mathbf{r}^t\rangle}{\|\mathbf{g}_m^t\|\|\mathbf{r}^t\|} $.
The weighted aggregations of $ y_m^t $ is bounded, $\forall t$, i.e.,
\begin{align} 
    \rho^t &= \frac{1}{S}\sum_{m \in \mathcal{B}^t} y_m^t \rho_m^t \in [0, (1-w^t)q]; \label{assumption 4-2} 
\end{align}
where $w^t \in [0,1]$ denotes the intensity level of the Byzantine attack defined in Section II-B.}

\smallskip
\smallskip

Unlike the typical bounded gradient assumptions, e.g.,~\cite[Asm. 4]{li2019convergence} and~\cite[Asm. 5]{zuo2024byzantine}, \textbf{Assumption~3} only requires a bounded ratio between the trusted reference direction $\mathbf{r}^t$ and the benign local gradients $\mathbf{g}_m^t, \forall m$, without restricting their individual magnitudes, relaxing constraints on modeling data heterogeneity, and accommodating diverse data distributions, compared to the typical requirement of uniform gradient norms (i.e., $\|\mathbf{g}_m^t\|\leq G$~\cite{li2019convergence,zuo2024byzantine}). This relaxation better reflects practical federated settings with non-IID data.

In \textbf{Assumption~4}, \eqref{assumption 4-2} enforces that the weighted average of the cosine similarities between the benign workers’ gradients and trusted reference direction $\mathbf{r}^t$, scaled by their gradient-norm ratios $\rho_m^t$, remains within~$[0, (1-w^t)q]$. 
In other words, the benign updates are assumed to be positively aligned with $\mathbf{r}^t$ and have bounded relative magnitudes. The benign contribution to the aggregation maintains a consistent descent direction without being dominated by excessively large updates. 

Interestingly, we do not assume explicit bounds or range for the cosine similarities of the malicious workers $m \in \mathcal{A}^t$. The only assumption about the Byzantine attacks is the intensity level of attacks, i.e., $w^t$ in \eqref{assumption 4-2}. This substantially relaxes the reliance of our analysis and findings on assumptions of a-priori knowledge about the system.

Under \textbf{Assumptions 1}–\textbf{4}, the following theorem establishes.

\begin{theorem} \label{brdrag theorem}
Suppose that $ F_m, \forall m \in \mathcal{M}$ is non-convex.
Given an initial model $\boldsymbol{\theta}^0$, after $T$ training rounds,
\begin{itemize}
    \item if $c^t = \frac{w^t}{w^t - x^t} \in [\frac{1}{2},1]$ and the stpdfize satisfies $\eta \leq \frac{(1-\tilde{w})(1-\tilde{c})p}{32U+2U L^2(p+2q^2)}$, BR-DRAG guarantees
    \begin{align} \label{brdrag theorem convergence}
    \frac{1}{T}\sum_{t=0}^{T-1} \mathbb{E}\left[\|\nabla f(\boldsymbol{\theta}^t)\|^2\right] \leq \frac{f(\boldsymbol{\theta}^0)-f^*}{\chi \eta U T} + W,
    \end{align}
    where $\chi>0$, $\tilde{c} = \max _t \frac{w^t}{w^t - x^t}$, $\tilde{w}=\max _t w^t$, and $W =   \frac{2\eta \sigma_L^2}{\chi}\left(\frac{\left(1+\tilde{w}\right) \tilde{c} \eta U L^2}{2}+\frac{\eta U L^2 q^2}{2\left(1-\tilde{c}\right) p}+1\right) $.

    \item if $\eta$ and $c^t$ satisfy $\eta \leq \frac{2c^t(1-c^t)(1-w^t)(1-w^t-y^t)p}{32 p U+ c^t U L^2(p^2+q^2)}$ and $\left( (\rho^t)^2 \!-\! 7w^t \!+\! 8x^t \!-\! 1 \right)(c^t)^2 \!+\! (15w^t \!+\! 1)c^t  \geq \frac{8(w^t)^2}{w^t-x^t}, \forall t$, respectively, BR-DRAG guarantees
    \begin{align} \label{brdrag theorem convergence}
    \frac{1}{T}\sum_{t=0}^{T-1} \mathbb{E}\left[\|\nabla f(\boldsymbol{\theta}^t)\|^2\right] \leq \frac{f(\boldsymbol{\theta}^0)-f^*}{\chi \eta U T} + W,
    \end{align}
    where $\chi>0$, $\tilde{c}=\max_t c^t$, $\tilde{w}=\max _t w^t$, $x^t = \frac{1}{S}\sum_{m \in \mathcal{A}^t} x_m^t$ with $x_m^t=\frac{\langle\mathbf{g}_m^t, \mathbf{r}^t\rangle}{\|\mathbf{g}_m^t\|\|\mathbf{r}^t\|}$, $y^t=\frac{1}{S} \sum_{m \in \mathcal{B}^t } y_m^t$, and $W =   \frac{2\eta \sigma_L^2}{\chi}\left(\frac{\left(1+\tilde{w}\right) \tilde{c} \eta U L^2}{2}+\frac{\eta U L^2 q^2}{2\left(1-\tilde{c}\right) p}+H+1\right) $ with $H = \max_{t} \frac{\| w^t(1-c^t)+c^t x^t \|}{2c^t\eta U (w^t-x^t) } $.
\end{itemize}
\end{theorem}
\begin{proof}
    See Appendix B.
\end{proof}

According to \textbf{Theorem}~\textbf{\ref{brdrag theorem}}, BR-DRAG converges under Byzantine attacks in the presence of heterogeneity and partial aggregation, unlike conventional methods, e.g,~\cite{li2019rsa,fang2020local,blanchard2017machine}.
The constant $W$ specifies the non-vanishing error caused by gradient variance, data heterogeneity, and adversarial attacks.
In each training round $t$, the malicious updates induce the increased errors $ \delta^t =  { c^t  (w^t-x^t) \eta U}$ and $\kappa^t = c^t (1-w^t-y^t)\eta U$ in \eqref{one round convergence new}, impeding convergence. 
Based on \eqref{vm1}, BR-DRAG projects the malicious vectors on a bounded space $\mathbf{r}^t$, i.e., the trusted reference direction, tolerating the malicious updates.

The resistance of BR-DRAG to malicious worker participation is achieved by reducing the stpdfize $\eta$ and enhancing the gradient alignment degree. As the proportion of malicious workers increases (i.e., $w^t$ grows), the hyperparameter $c^t$ needs to increase, which in turn reduces the effective stpdfize. 
This adjustment slows convergence to maintain the conditions necessary for stability and robustness.
An excessively large $c^t$ amplifies the variance-related components, e.g., $\frac{\tilde{c}(1+\tilde{w})L^2}{\chi}\eta^2 U \sigma_L^2$ in $W$ in BR-DRAG, while an excessively small $c^t$ may degrade gradient alignment in \eqref{vm1}, heightening the vulnerability to adversarial gradients. To this end, $c^t$ should be calibrated meticulously to suppress adversarial deviation while controlling gradient variance in BR-DRAG.

Unlike a majority of the existing Byzantine-resilient FL algorithms~\cite{li2019rsa,cao2020fltrust,9721118,9153949,10100920,zuo2024byzantine}, which typically require the condition of $A^t < \frac{S}{2}, \forall t$,
BR-DRAG allows an arbitrary proportion of malicious workers in each round. 
This is enabled by the design of the reference direction in \eqref{trust} and gradient modification in \eqref{vm1}, suppressing the effects of malicious gradient aggregation.
Moreover, by plugging $\eta = \mathcal{O}(\frac{1}{\sqrt{T}})$ and $U = \mathcal{O}(1)$ into \eqref{brdrag theorem convergence}, BR-DRAG maintains the same convergence rate of $\mathcal{O}(\frac{1}{\sqrt{T}})$ as DRAG, which is consistent with FL algorithms without Byzantine attacks~\cite{li2020federated}.


\section{Experiments and Results}



This section experimentally evaluates DRAG and BR-DRAG compared to state-of-the-art FL and Byzantine-resilient FL algorithms, with the following datasets and models:
\begin{itemize}
    \item EMNIST dataset~\cite{cohen2017emnist}: For the
     digit and letter classification task, we utilize the ``balanced" data split containing 131,600 samples with 47 classes. The designed Convolutional Neural Network (CNN) comprises two $5 \times 5$ convolutional layers and two fully connected layers, with a final output layer for 47-class classification.
    
    \item CIFAR-10 dataset~\cite{krizhevsky2009learning}: 
    For the image classification task, the CIFAR-10 dataset is composed of 60,000 color images in a total of 10 classes. We employ a CNN with two $5 \times 5$ padded convolutional layers to extract hierarchical features and predict across 10 classes.

    \item CIFAR-100 dataset~\cite{krizhevsky2009learning}:
    For the fine-grained image classification task, the CIFAR-100 dataset contains 60,000 color images categorized into 100 classes.
    We adopt a deeper CNN comprising three $3 \times 3$ padded convolutional layers followed by max pooling operations, and two fully connected layers, culminating in a 100-class output layer.
    {\color{black}
    We also consider the WRN-28-4 model~\cite{zagoruyko2016wide}, which is a derivative of the standard ResNet architecture. 
        The WRN-28-4 model has 28 convolutional layers and a widening factor of 4 to expand base feature channels, while incorporating Batch Normalization within its widened residual blocks to stabilize the optimization process.
    }
\end{itemize}
Consider an FL system with $40$ workers, and each worker performs $U=5$ local updates in each training round.
To characterize data heterogeneity, we utilize the Dirichlet distribution with parameter $\beta$~\cite{pmlr-v162-qu22a,jhunjhunwala2023fedexp,Kim_2024_CVPR}. Specifically, we sample $p_k \sim \operatorname{Dir}(\beta)$ and allocate a proportion $p_{k,j}$ of image class $k$ instances to device $j$. A smaller $\beta$ represents a more imbalanced data distribution. Here, we choose $\beta=0.1$ and $0.5$ to represent strong and moderate data heterogeneity, respectively.
By default, we set the stpdfize $\eta = 0.01$ and the batch size is 10.
To simulate partial worker participation, we set the worker subset size to $S=10$ per round. The FL environment is set up in PyTorch 2.5.1 on a Windows 10 operating system with one 4090 GPU.


\subsection{Evaluation of DRAG}

We first compare DRAG with the following benchmarks, which have been successfully deployed in FL systems.
\begin{itemize}
    \item FedAvg~\cite{pmlr-v54-mcmahan17a}: The local and global models are updated as \eqref{update} and \eqref{aggregation}, respectively.


    \item FedProx~\cite{li2020federated}: In any training round $t$, the workers perform $U$ local updates with a regularization term to approximate a $\gamma$-inexact solution: $\left\|\nabla F_m(\boldsymbol{\theta}_{m}^{t,u})+\mu(\boldsymbol{\theta}_{m}^{t,u}-\boldsymbol{\theta}^t)\right\| 
        \leq \gamma\left\|\nabla F_m(\boldsymbol{\theta}^t)\right\|$.
        
    \item SCAFFOLD~\cite{karimireddy2020scaffold}:
    In any training round $t$, the local update is employed with control variates: $\boldsymbol{\theta}_{m}^{t,u}=\boldsymbol{\theta}_{m}^{t,u}-\eta ( \nabla F_m(\boldsymbol{\theta}_{m}^{t,u}; z_{m}^{t,u})  - h_{m}^t + h^t )$, where $h_{m}^t$ and $h^t$ are updated as $h_{m}^{t+1}=\nabla F_m (\boldsymbol{\theta}^{t}; z_{m}^{t,0} ) $ and $h^{t+1}=h^{t}+\frac{1}{M}\sum_{i \in \mathcal{S}^t} (\nabla F_m (\boldsymbol{\theta}^{t}; z_{m}^{t,0} )-h_{m}^t  )$.

    \item FedExP~\cite{jhunjhunwala2023fedexp}: With $\tilde{\Delta}^t = \frac{1}{S}\sum_{i \in \mathcal{S}^t}\mathbf{g}_m^t$ and a small positive constant $\epsilon$, the global model is updated as $\boldsymbol{\theta}^{t+1}=\boldsymbol{\theta}^t+\frac{\eta_g^t}{S} \sum_{m \in \mathcal{S}^t} \mathbf{g}_m^t$, where $\eta_g^t=\max \Big\{1, \frac{\sum_{i \in \mathcal{S}^t}\left\|\mathbf{g}_m^t\right\|^2}{2 S\left(\left\|\tilde{\Delta}^t\right\|^2+\epsilon\right)}\Big\}$.   
    

    \item FedACG~\cite{Kim_2024_CVPR}: The local updates are performed as 
    $\boldsymbol{\theta}_m^{t,u+1} = \boldsymbol{\theta}_m^{t,u} - \eta ( \nabla F_m(\boldsymbol{\theta}_m^{t,u};z_{m}^{t,u}) + \beta \| \boldsymbol{\theta}_m^{t,u} - (\boldsymbol{\theta}^{t-1} + \lambda m^{t-1} ) \| )$ with the lookahead gradient updated as $m^t = \lambda m^{t-1} + \frac{1}{S} \sum_{m \in \mathcal{S}^t} \mathbf{g}_m^t $.

\end{itemize} 
After fine-tuning these benchmarks, we set $\mu=0.2$ for FedProx; $\epsilon=0.001$ for FedExP; $\beta=0.2$ and $\lambda=0.85$ for FedACG. For DRAG, $\alpha=0.25$, $c=0.1$ for moderate data heterogeneity, and $c=0.25$ for strong data heterogeneity.

Figs.~\ref{fig_EMNIST_DRAG} to \ref{fig_CIFAR100_DRAG} show the convergence of the algorithms. Across different datasets and data heterogeneity levels, DRAG significantly outperforms the benchmarks, validating \textbf{Theorem~1}. 
The accuracy gap between DRAG and FedAvg is larger when data heterogeneity is more severe (from $\beta=0.5$ to $\beta=0.1$), confirming the effective client-drift mitigation of DRAG.
Meanwhile, DRAG converges considerably faster than the benchmarks, and achieves 70\% test accuracy within 600 rounds in Fig.~\ref{fig_CIFAR10_DRAG}(b). Over 1,400 rounds are needed for other algorithms, e.g., SCAFFOLD and FedExP, verifying that the adaptive aggregation of DRAG in \eqref{vm} accelerates the decrease in the global loss function.

FedProx demonstrates subpar performance due to the lack of enough local computations to achieve the required $\gamma$-inexact solution. With the additional lookahead gradient, FedACG achieves a limited accuracy improvement over FedProx. FedExP is limited by the later stage of training and degrades to FedAvg.
Since the control variate $\nabla F_m(\boldsymbol{\theta}^t) - \mathbb{E}[h_m^t] + \mathbb{E}[h^t] \neq 0$ when $\boldsymbol{\theta}^t=\boldsymbol{\theta}^*$, the local models of SCAFFOLD deviate from $\boldsymbol{\theta}^*$ and necessitate additional training rounds to mitigate this imbalance of control variate.
By balancing the local gradient direction and the reference direction, DRAG avoids using control variates or regularization terms, realizing adaptive model training without additional communication overhead.


\begin{figure}[!t]
	\centering
	\subfloat[{\centering $\beta=0.1$.}]{\includegraphics[width=1.65in]{ 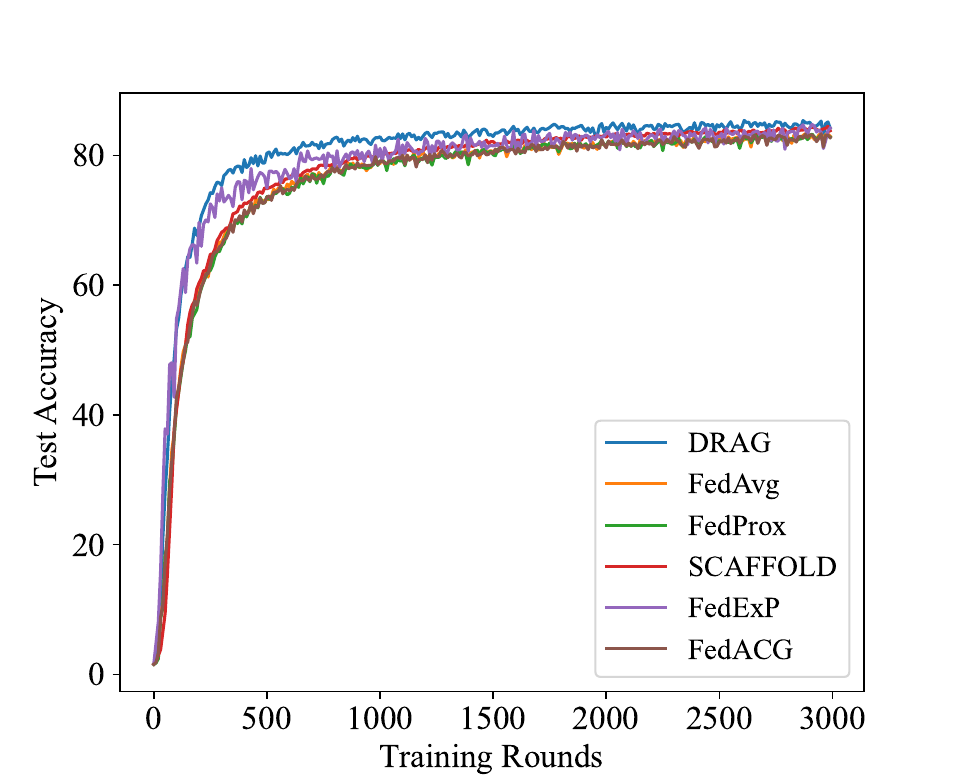}
		\label{fig_EMNIST_1}}
	\subfloat[{\centering $\beta=0.5$.}]{\includegraphics[width=1.65in]{ 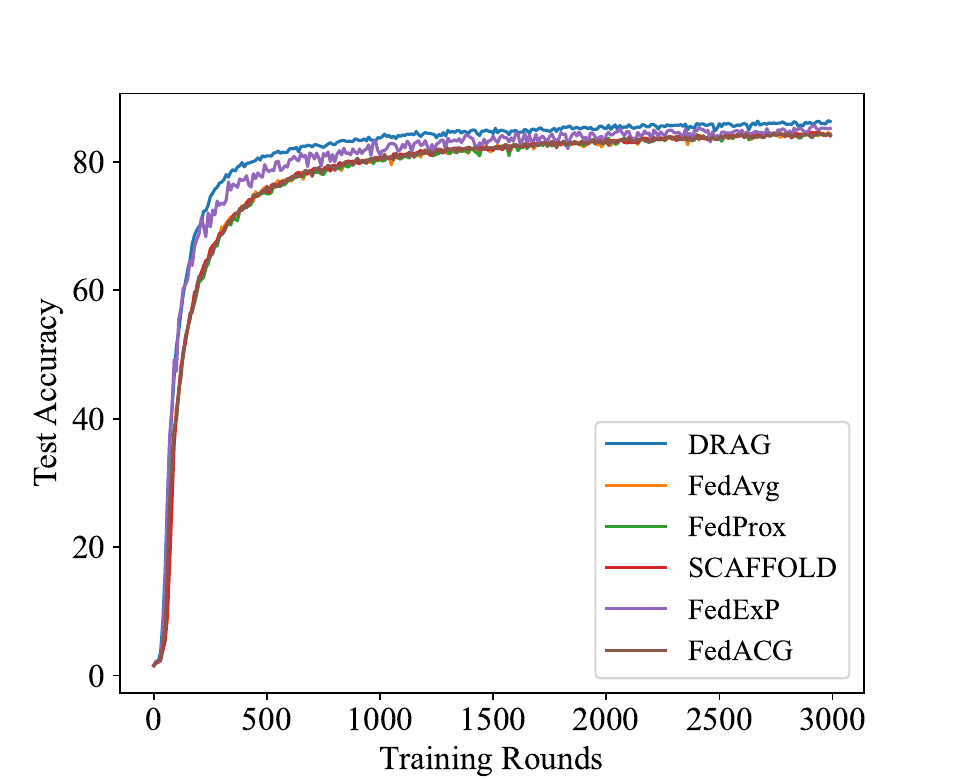 }
		\label{fig_EMNIST_2}}
	\caption{Convergence performance of different algorithms on EMNIST.}
	\label{fig_EMNIST_DRAG}
\end{figure}
\begin{figure}[!t]
	\centering
	\subfloat[{\centering $\beta=0.1$.}]{\includegraphics[width=1.65in]{ 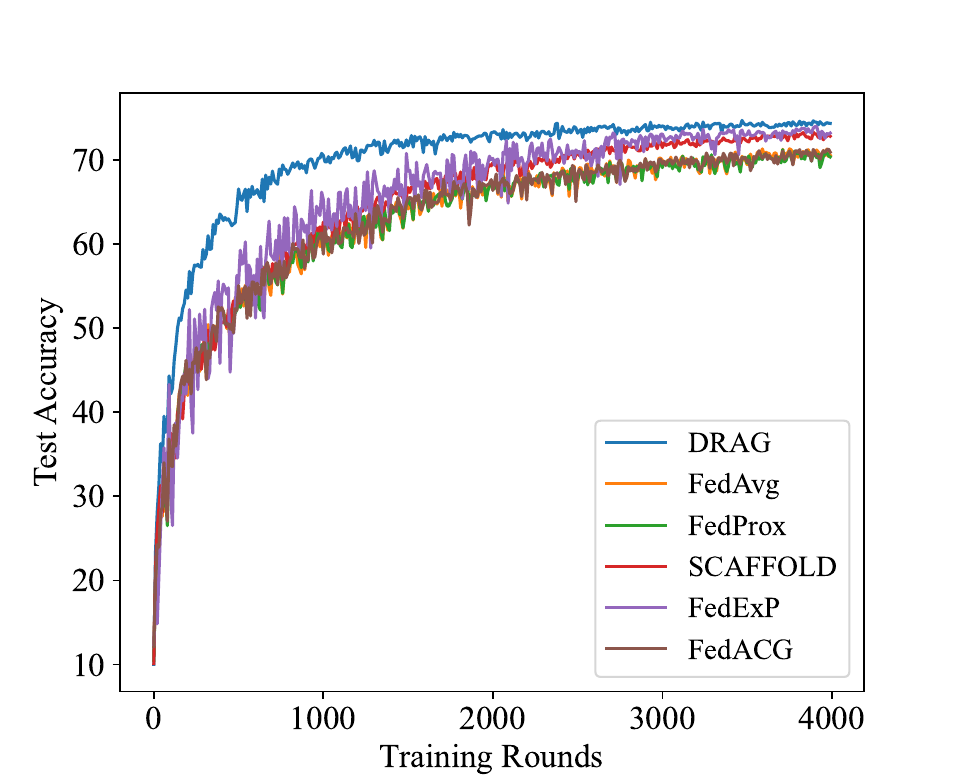}
		\label{fig_CIFAR10_1}}
	\subfloat[{\centering $\beta=0.5$.}]{\includegraphics[width=1.65in]{  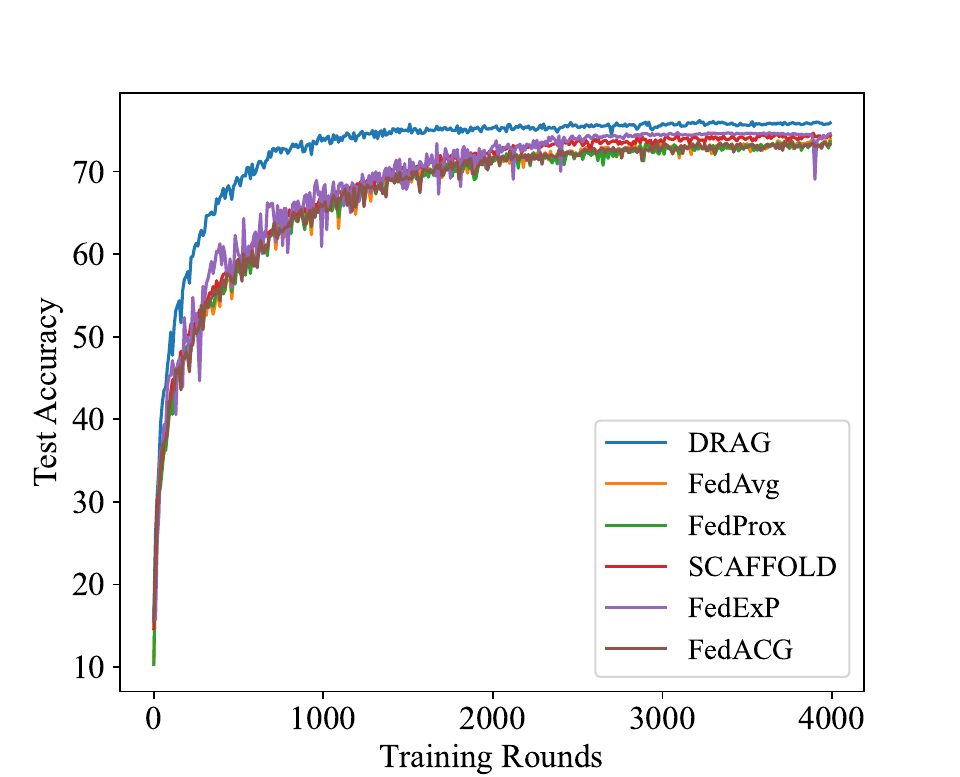}
		\label{fig_CIFAR10_2}}
	\caption{Convergence performance of different algorithms on CIFAR-10.}
	\label{fig_CIFAR10_DRAG}
\end{figure}
\begin{figure}[!t]
	\centering
	\subfloat[{\centering $\beta=0.1$.}]{\includegraphics[width=1.65in]{ 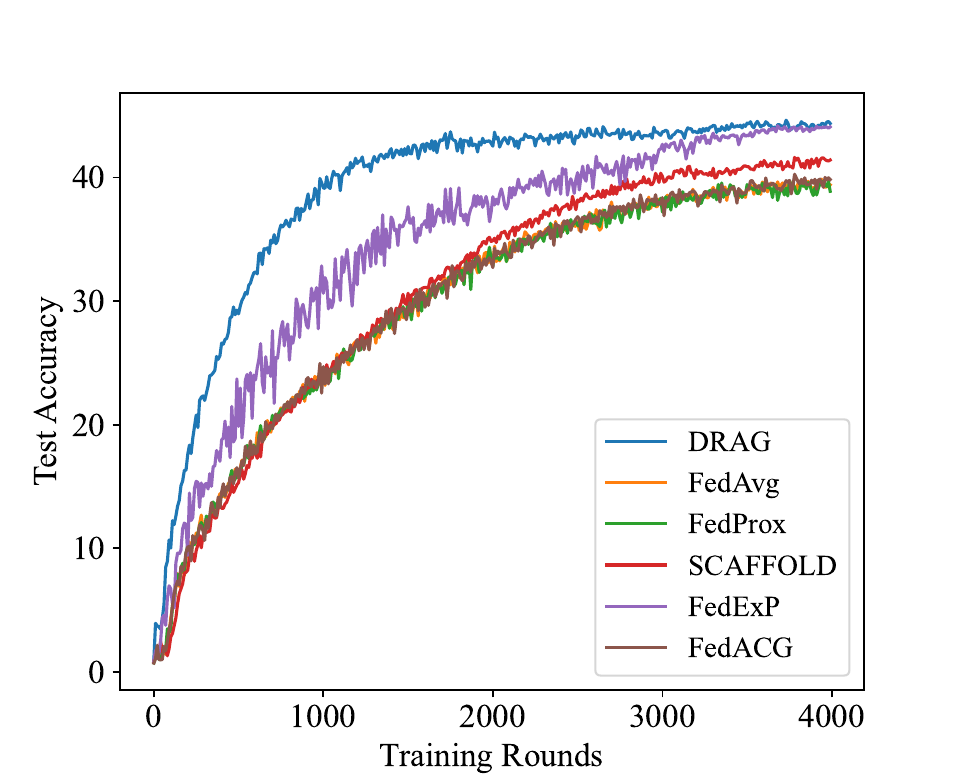}
		\label{fig_CIFAR100_1}}
	\subfloat[{\centering $\beta=0.5$.}]{\includegraphics[width=1.65in]{ 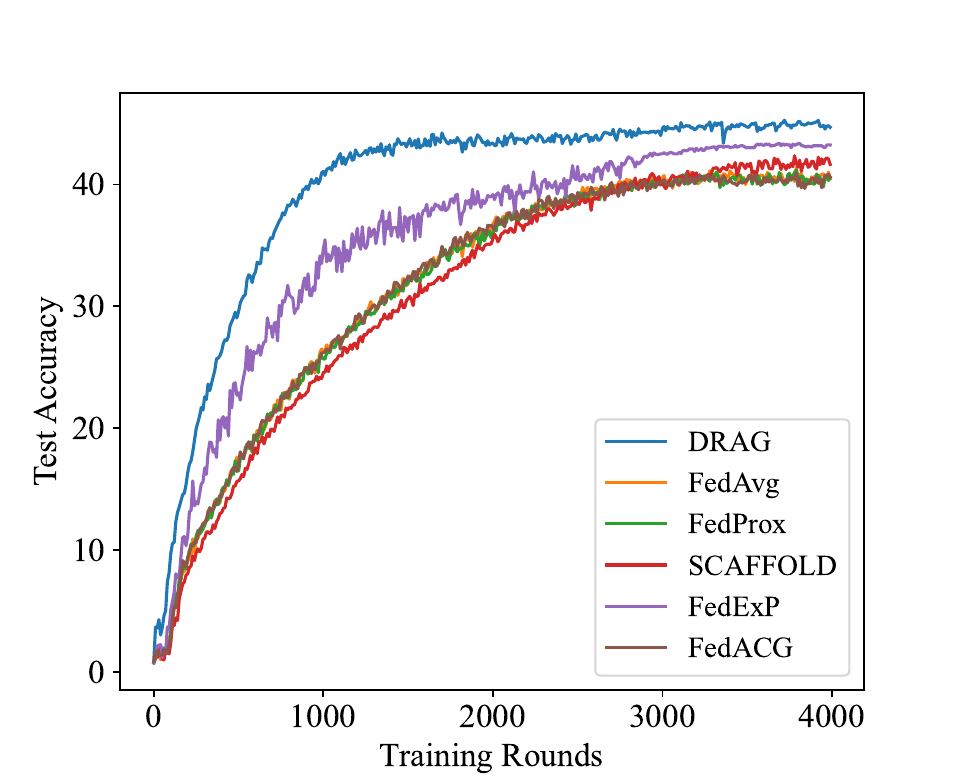}
		\label{fig_CIFAR100_2}}
	\caption{Convergence performance of different algorithms on CIFAR-100.}
	\label{fig_CIFAR100_DRAG}
\end{figure}

\begin{figure}[!t]
	\centering
	\subfloat[ $\beta=0.1$. ]{\includegraphics[width=1.65in]{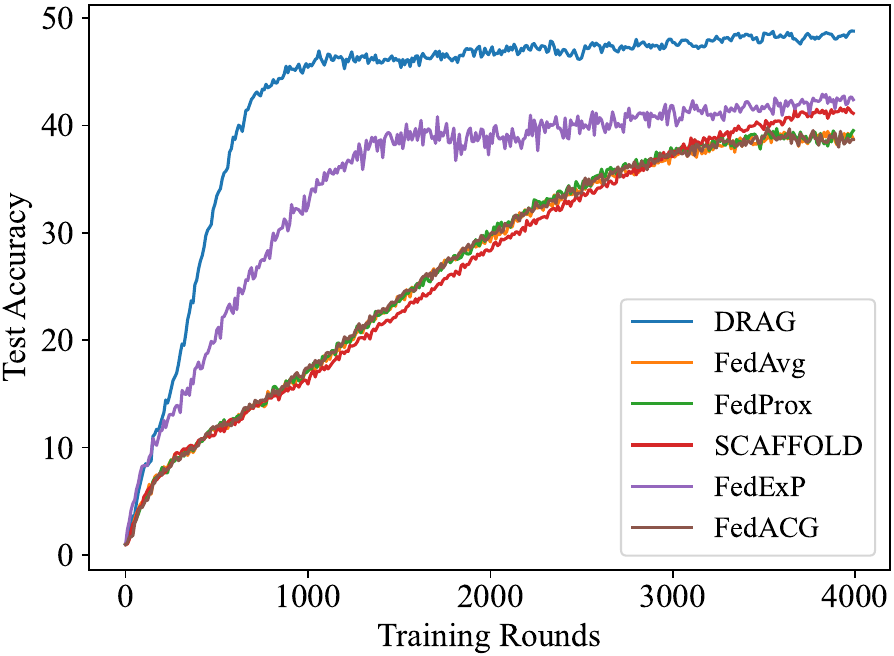}
		\label{fig_wrn1_1}}
	\hfill 
	\subfloat[$\beta=0.5$.]{\includegraphics[width=1.6in]{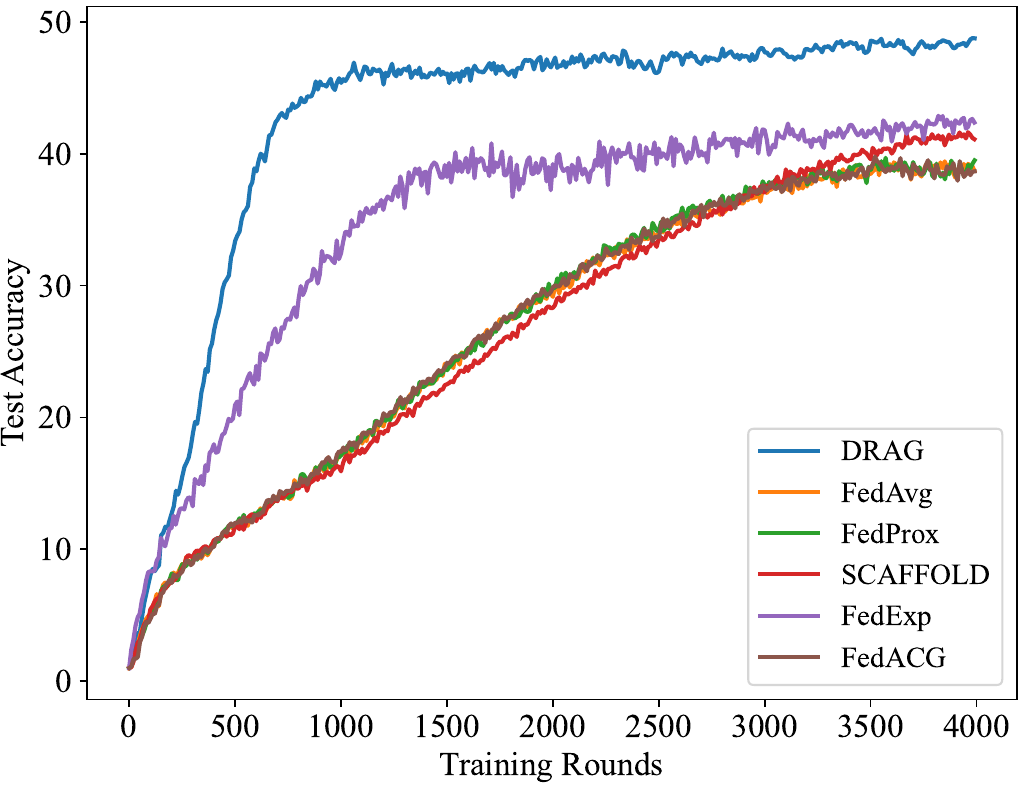}
		\label{fig_wrn1_2}}
	\caption{Convergence performance of different algorithms on CIFAR-100 with the WRN-28-4 model.}
	\label{wrn1}
    \end{figure}

{\color{black} Fig.~\ref{wrn1} plots the convergence performance of DRAG under different data heterogeneity levels with the WRN-28-4 model. It is demonstrated that DRAG maintains its convergence performance advantage over the benchmarks, validating its superior adaptability to training models with larger scales.}
Fig.~\ref{fig_diff_S_DRAG} evaluates DRAG under different numbers of participating workers, where a UAR selection of $S =5$, $15$, $25$, and $35$ out of $M=40$ workers is considered.
Clearly, more worker participation leads to faster convergence under both data distributions.
A larger number of selected workers $S$ contributes to more stable reference directions $\mathbf{r}^t, \forall t$, in \eqref{reference_cf} and helps better guide local gradient modifications.
\begin{figure}[t!]
\centering
\subfloat[\centering Accuracy on CIFAR-10 with $\beta=0.1$]{\includegraphics[width=1.7in]{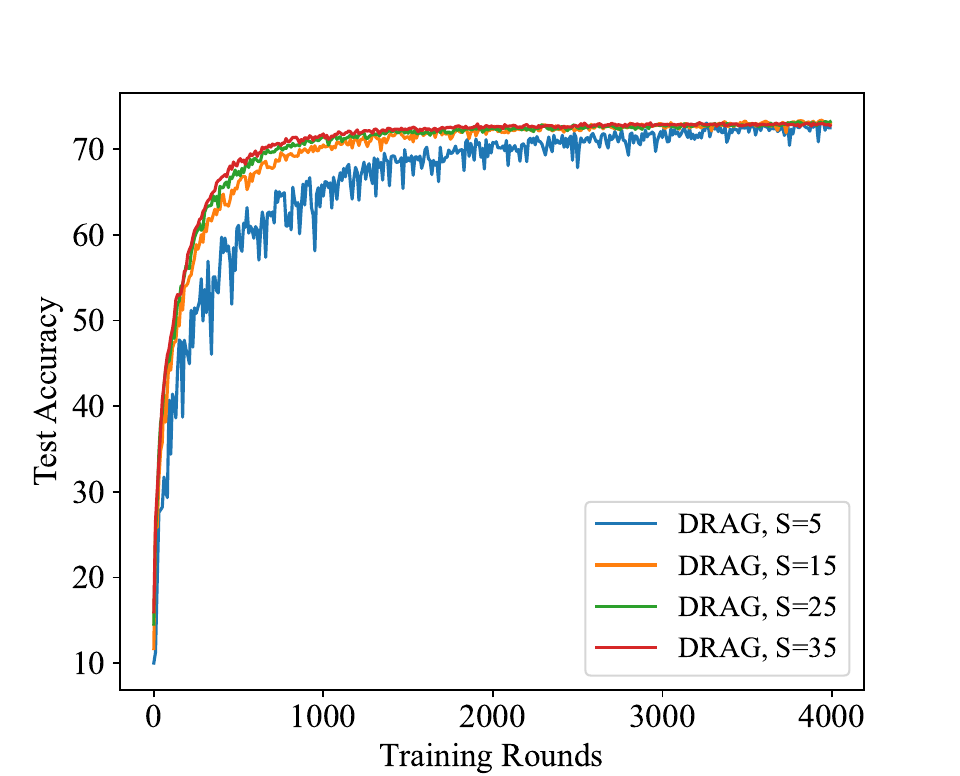}}
\subfloat[\centering Accuracy on CIFAR-10 with $\beta=0.5$]{\includegraphics[width=1.7in]{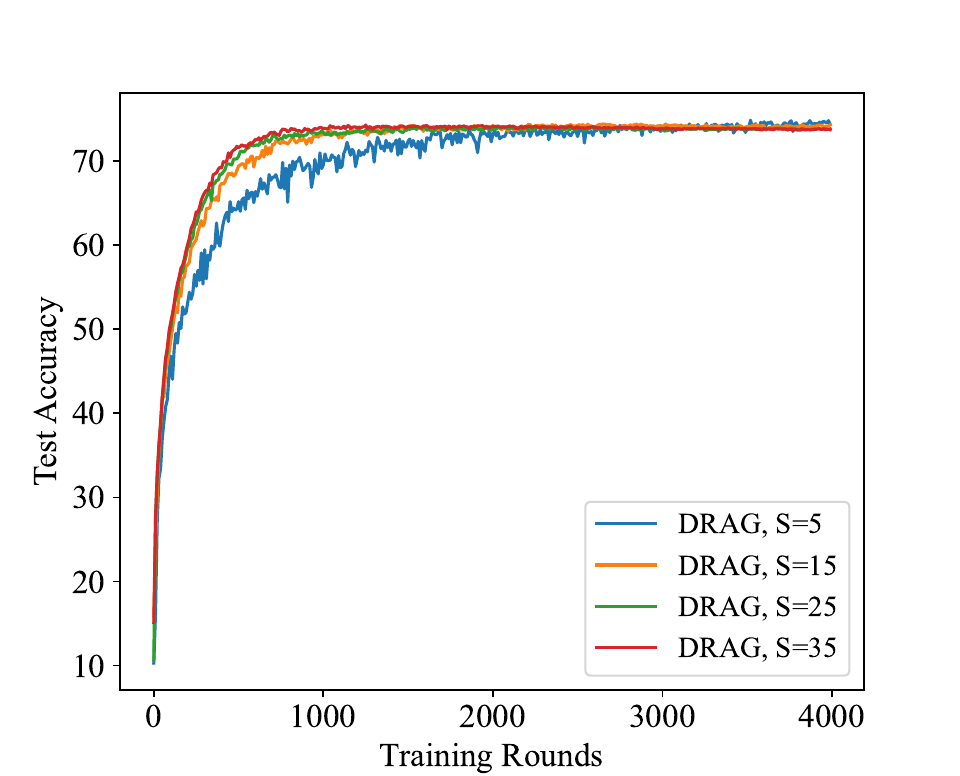}}
\caption{Convergence performance of DRAG with different numbers of participating workers on CIFAR-10.}
\label{fig_diff_S_DRAG}
\end{figure}

Figs. \ref{diff_alpha_DRAG} and \ref{diff_c_DRAG} examine the impact of the hyperparameters $\alpha$ and $c$ in \eqref{reference_cf} and \eqref{dod} on the performance of DRAG.
Fig.~\ref{diff_alpha_DRAG}(a) demonstrates an overly small $\alpha=0.01$ in \eqref{reference_cf} overuses historical gradients, where unstable gradients in the early training stage hinder convergence. When $\alpha$ is large, e.g., $\alpha > 0.25 $, the reference direction $\mathbf{r}^t$ depends excessively on $\Delta^{t-1}$, making it difficult to find a fast descent path.
In Fig.~\ref{diff_alpha_DRAG}(b), although the data heterogeneity becomes moderate, an overly large or small value of $\alpha$ slows down the convergence.
Fig.~\ref{diff_c_DRAG} shows that too large or too small $c$ significantly impairs training, since it would penalize gradient variance control or gradient alignment.
This is consistent with the analysis in Section V-B, confirming the importance of careful configuration of $c$.

\begin{figure}[t!]
\centering
\subfloat[\centering $\beta=0.1$]{\includegraphics[width=1.7in]{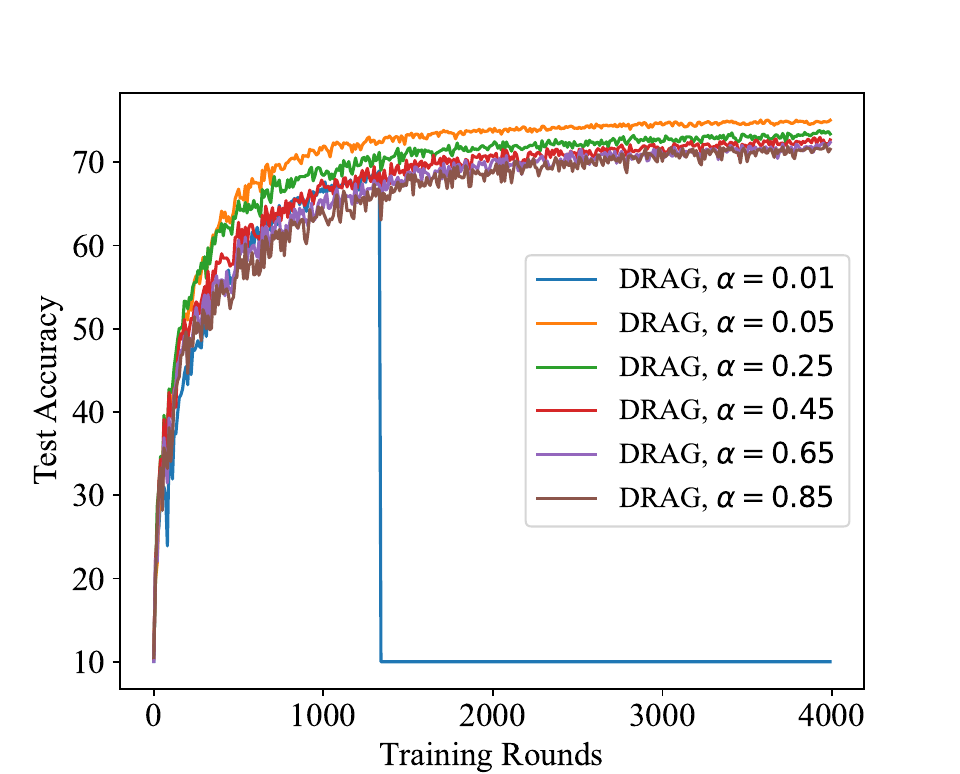}}
\subfloat[\centering $\beta=0.5$]{\includegraphics[width=1.7in]{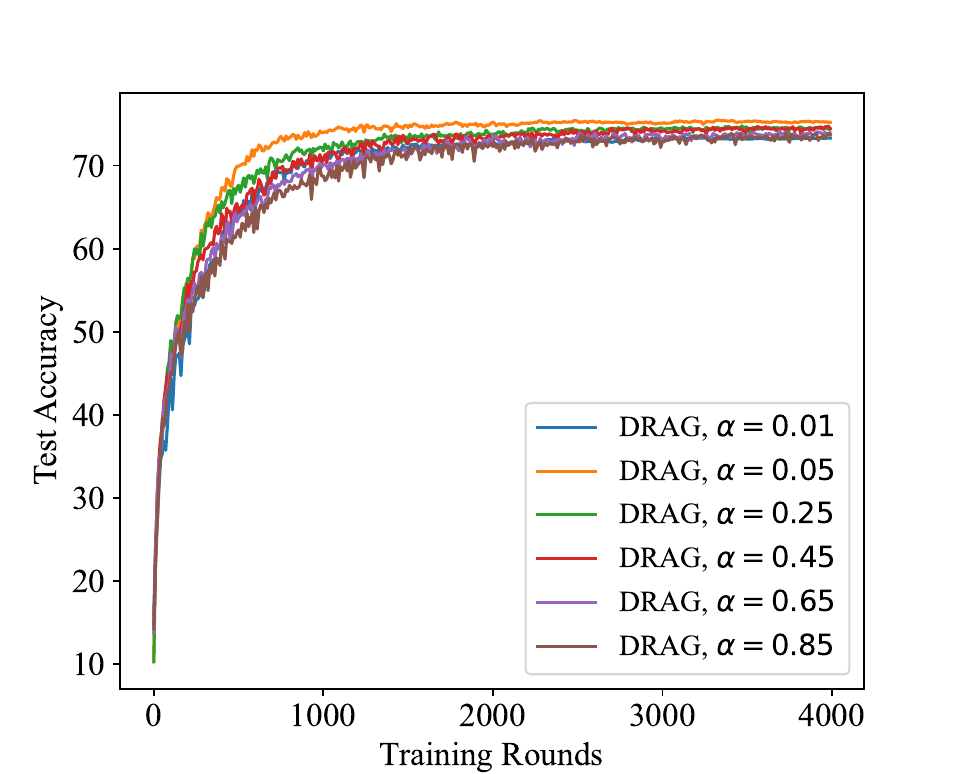}}
\caption{Convergence performance of DRAG with different values of hyperparameter $\alpha$ on CIFAR-10.}
\label{diff_alpha_DRAG}
\end{figure}

\begin{figure}[t!]
\centering
\subfloat[\centering $\beta=0.1$]{\includegraphics[width=1.7in]{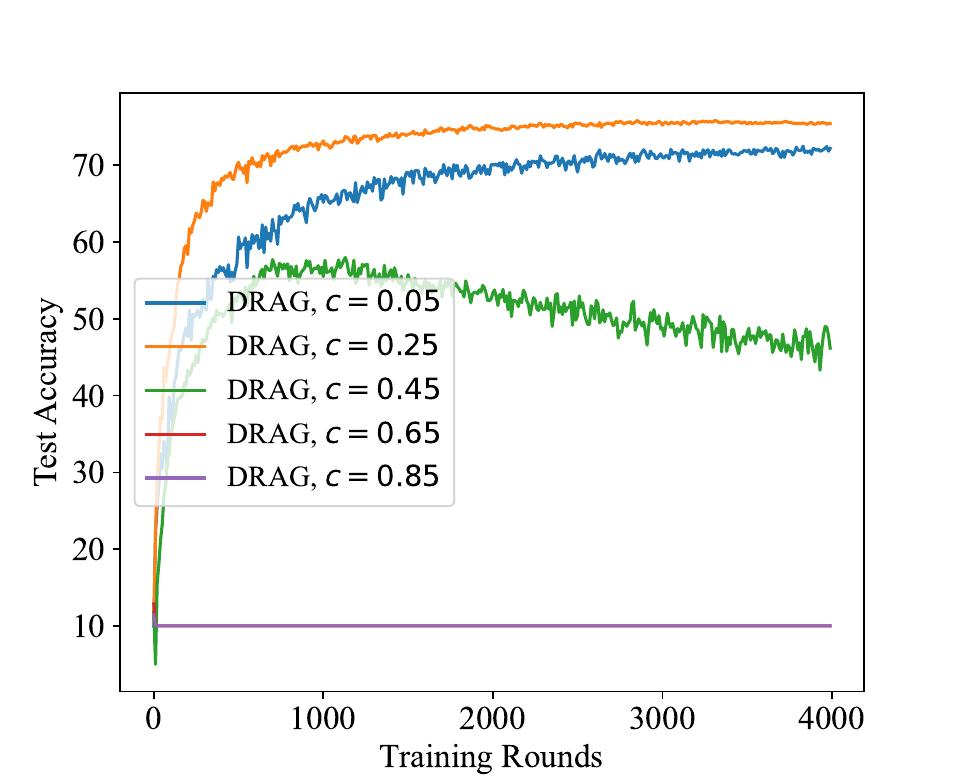}}
\subfloat[\centering $\beta=0.5$]{\includegraphics[width=1.7in]{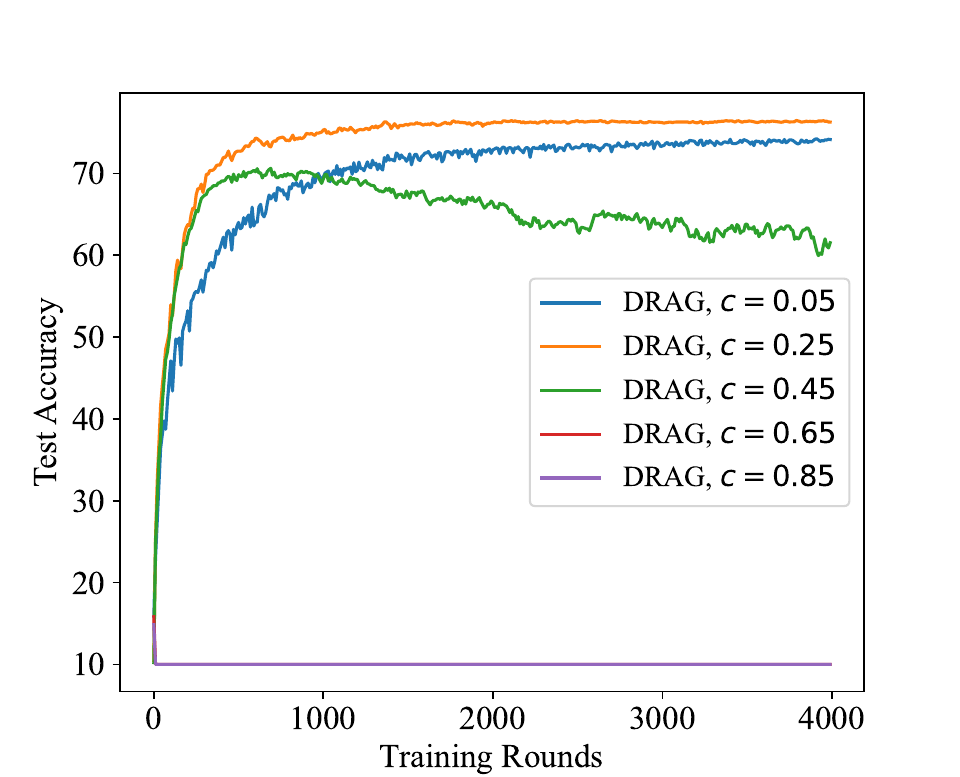}}
\caption{Convergence performance of DRAG with different values of hyperparameter $c$ on CIFAR-10.}
\label{diff_c_DRAG}
\end{figure}




\subsection{Evaluation of BR-DRAG}

We proceed to evaluate the BR-DRAG framework under Byzantine attacks by considering the following benchmarks.
\begin{itemize}
    \item FedAvg~\cite{pmlr-v54-mcmahan17a}: The local and global models are updated following \eqref{update} and \eqref{byzantine aggregation}, respectively.

    \item FLTrust~\cite{cao2020fltrust}: In any training round $t$, the local gradients are modified with ReLU-clipped cosine similarity: $\tilde{\mathbf{g}}_m^t \!=\! \max \! \left(0, \cos \left(\mathbf{g}_m^t, \mathbf{r}^t\right)\right) \! \left\|\mathbf{g}_m^t\right\| \! \frac{\mathbf{r}^t}{\left\|\mathbf{r}^t\right\|}$, where $\mathbf{r}^t$ is the reference direction designed in a similar way to \eqref{trust}.
    
    \item RFA~\cite{9721118}: The global model is updated via geometric median aggregation: $\boldsymbol{\theta}^{t+1} = \operatorname{GeoMed}\left( \{ \boldsymbol{\theta}_m^{t,U} \}_{m \in \mathcal{S}^t} \right)$, where $\operatorname{GeoMed}(\cdot)$ minimizes the sum of distances to all local models and is solved with the Weiszfeld algorithm~\cite{beck2015weiszfeld}.

    \item RAGA~\cite{zuo2024byzantine}: The global gradient $\mathbf{g}^t$ is acquired with the geometric median metric: $\mathbf{g}^t \!=\! \operatorname{GeoMed}\left( \{ \mathbf{g}_m^{t} \}_{m \in \mathcal{S}^t} \right)$, and the global model is updated as $\boldsymbol{\theta}^{t+1}=\boldsymbol{\theta}^{t} + \mathbf{g}^t$.
\end{itemize}
{\color{black}For malicious workers, we consider five different attack methods, including noise injection~\cite{9614992}, sign flipping~\cite{li2019rsa}, label flipping~\cite{10054157}, Min-Max \cite{shejwalkar2021manipulating}, and Min-Sum \cite{shejwalkar2021manipulating}; see Section I-B.} For the noise injection attack, we set malicious local gradients to $p_m^t\mathbf{g}_m^t, \forall m \in \mathcal{A}^t$, where $p_m^t$ conforms to the Gaussian distribution $\mathcal{N}(0,3)$. For the sign flipping attack, the malicious workers generate $-\mathbf{g}_m^t, \forall m \in \mathcal{A}^t$ as the local gradient. For the label flipping attack, half of local data labels of the malicious workers are reversed. 
{\color{black}
For Min-Max and Min-Sum, let $\bar{\mathbf{g}}_b^t=\frac{1}{|\mathcal{B}^t|}\sum_{i\in\mathcal{B}^t}\mathbf{g}_i^t$ denote the average benign update in round $t$. Under both attacks, all malicious workers upload the same crafted update $\hat{\mathbf{g}}_m^t=\bar{\mathbf{g}}_b^t+\gamma\mathbf{p}^t, \forall m\in\mathcal{A}^t$, where $\mathbf{p}^t$ is the attack direction and $\gamma$ is chosen to be as large as possible while the malicious updates do not stand out under distance-based measurement criteria.
In the case of Min-Max, $\gamma$ is selected so that the maximum distance from $\hat{\mathbf{g}}_m^t$ to any benign update does not exceed the maximum pairwise distance among benign updates. In the case of Min-Sum, $\gamma$ is selected to keep the total distance from $\hat{\mathbf{g}}_m^t$ to all benign updates below the maximum of such total distances within the benign group.
Under BR-DRAG, we set $c^t=0.5, \forall t$. A vetted pool $\mathcal{D}_{\mathrm{vet}}$ is first constructed by applying a verification tool based on confident learning~\cite{northcutt2021confidentlearning} and outlier filtering~\cite{10.14778/3603581.3603583} to prevent data corrupted by attacks. Then, data samples are randomly and uniformly drawn from $\mathcal{D}_{\mathrm{vet}}$ in equal proportions to form $\mathcal{D}_{\mathrm{root}}$; see Section IV-B.}


{\color{black}We first experimentally validate \textbf{Assumptions 3} and \textbf{4} in our considered simulation settings. 
        Fig.~\ref{ass1} shows the variations in the values of $\min_{m \in \mathcal{B}^t} \frac{\left\|\mathbf{r}^t\right\|}{\left\|\mathbf{g}_m^t\right\|}$, $\max_{m \in \mathcal{B}^t} \frac{\left\|\mathbf{r}^t\right\|}{\left\|\mathbf{g}_m^t\right\|}$, and $\frac{1}{S(1-w^t)} \sum_{m \in \mathcal{B}^t} y_m^t \rho_m^t$ during the training process under the CIFAR-100 dataset with different data heterogeneity levels.
        It is observed that both $\min_{m \in \mathcal{B}^t} \frac{\left\|\mathbf{r}^t\right\|}{\left\|\mathbf{g}_m^t\right\|}$ and $\max_{m \in \mathcal{B}^t} \frac{\left\|\mathbf{r}^t\right\|}{\left\|\mathbf{g}_m^t\right\|}$ converge as training progresses, confirming that the ratio between the trusted reference direction $\mathbf{r}^t$ and the benign updates $\mathbf{g}_m^t$, i.e., $\rho_m^t=\frac{\|\mathbf{r}^t\|}{\|\mathbf{g}_m^t\|},\ \forall m \in \mathcal{B}^t$ is bounded; in other words, there exist $p$ and $q$, $0 < p < q$, such that $\rho_m^t$ satisfies $p \leq \rho_m^t \leq q$, as stated in (22). 
        This validates \textbf{Assumption 3}. It is also observed that the weighted average of the cosine similarities between $\mathbf{r}^t$ and $\mathbf{g}_m^t$, i.e., $\frac{\rho^t}{1-w^t} = \frac{1}{S(1-w^t)} \sum_{m \in \mathcal{B}^t} y_m^t \rho_m^t$, remains positive in each training round $t$ and is always smaller than $\max_{m \in \mathcal{B}^t} \frac{\left\|\mathbf{r}^t\right\|}{\left\|\mathbf{g}_m^t\right\|}$. Therefore, \textbf{Assumption 4} is validated. Moreover, when data heterogeneity becomes intense, $\frac{1}{S(1-w^t)} \sum_{m \in \mathcal{B}^t} y_m^t \rho_m^t \geq 0$ still holds, which further expands \textbf{Assumption 4}.
}
\begin{figure}[!t]
	\centering
	\subfloat[ $\beta=0.1$. ]{\includegraphics[width=0.48\linewidth]{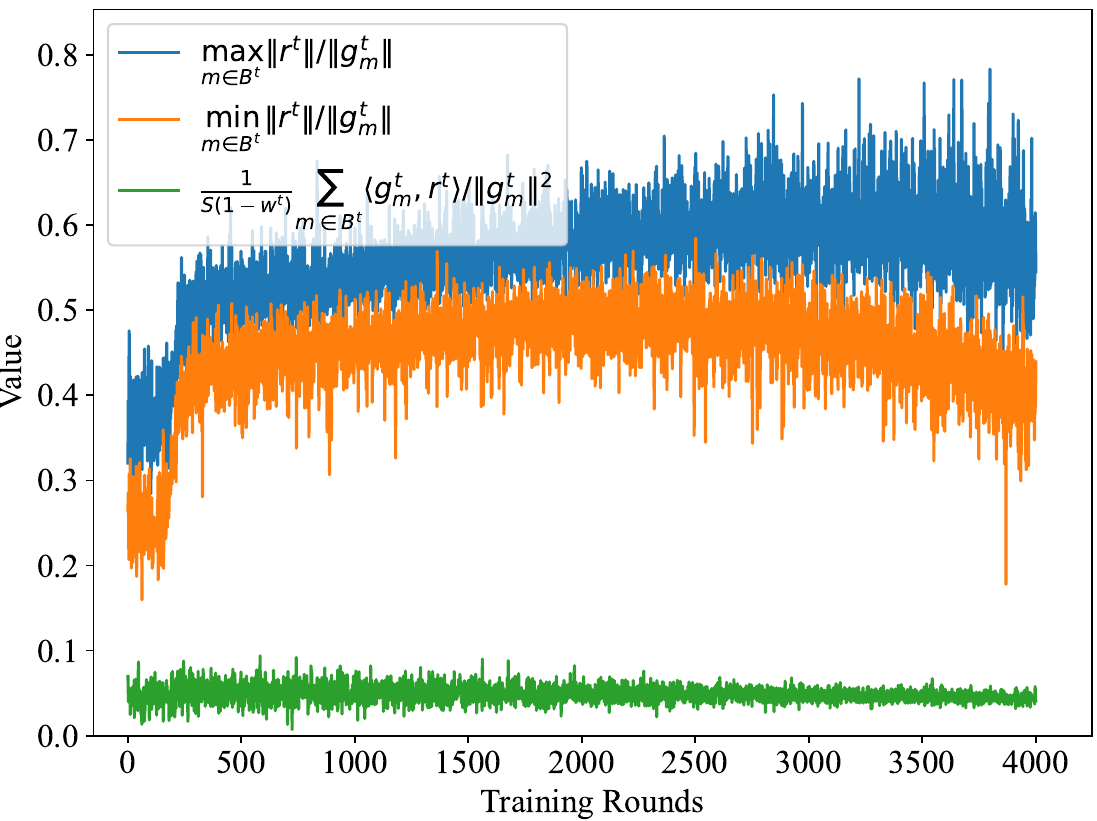}
		\label{fig_ass1_1}}
	\hfill 
	\subfloat[$\beta=0.5$.]{\includegraphics[width=0.48\linewidth]{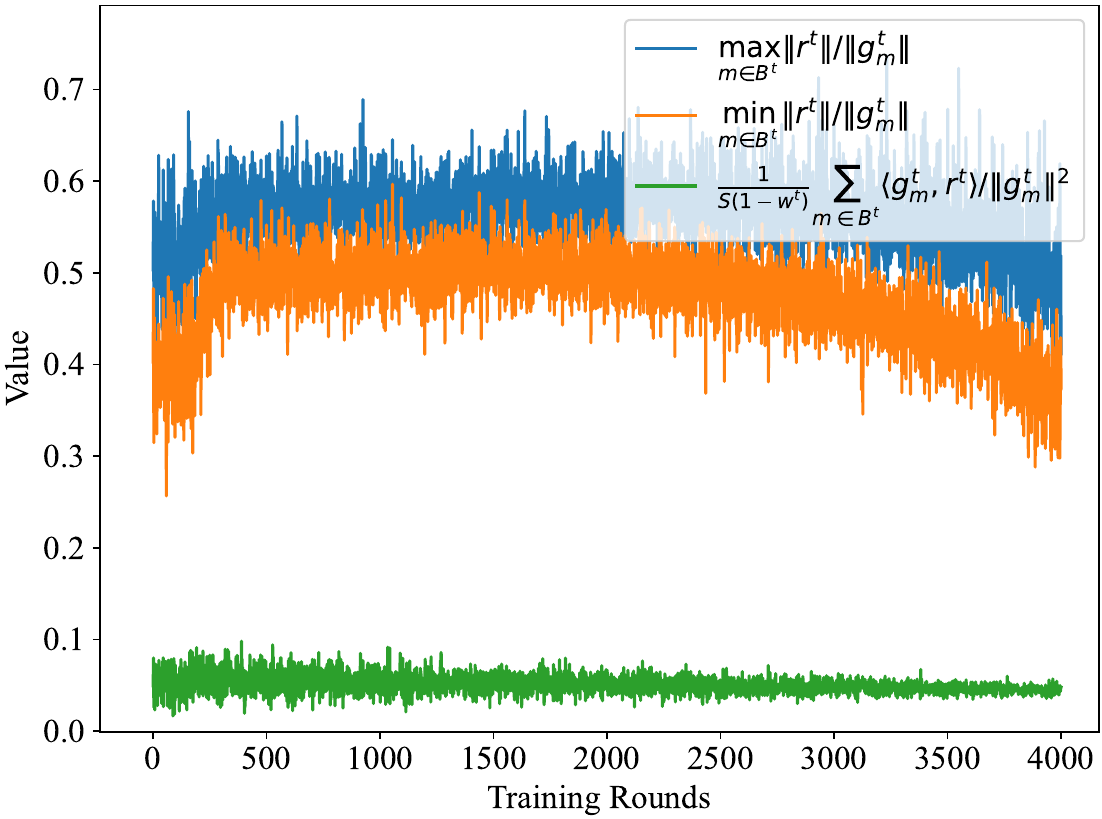}
		\label{fig_ass1_2}}
	\caption{The variation of the corresponding values in \textbf{Assumptions~3} and \textbf{4} during the training process on CIFAR-100.}
	\label{ass1}
    \end{figure}

We set the ratio of malicious workers to $ \frac{A}{M} = 30\%$ in the system.
Figs.~\ref{fig_CIFAR10_BRDRAG_NI} to \ref{fig_CIFAR100_BRDRAG_LF} evaluate the training performance of BR-DRAG under the different Byzantine attack methods, datasets, and data heterogeneity levels.
Compared to the benchmarks, BR-DRAG demonstrates superior adaptability and robustness against severe data heterogeneity and attacks, validating \textbf{Theorem 2}. 
For example, BR-DRAG surpasses FLTrust in accuracy by 13.4\% and 9.8\% in Figs.~\ref{fig_CIFAR10_BRDRAG_SF}(a) and~\ref{fig_CIFAR10_BRDRAG_SF}(b), respectively.
Based on \eqref{vm1} and using a larger $c$ than DRAG, BR-DRAG effectively suppresses malicious gradients while amplifying the impact of benign updates. 

{\color{black}
Fig.~\ref{wrn2} plots the convergence performance of BR-DRAG and its benchmarks under different data heterogeneity levels with the WRN-28-4 model. All benchmarks achieve improvements over the CNN without residual blocks. BR-DRAG demonstrates superior convergence over the benchmarks.
}

{\color{black}
Figs.~\ref{minmax1} and \ref{minsum1} plot the performance of BR-DRAG and the benchmarks under the Min-Max and Min-Sum attacks, respectively.
BR-DRAG remains more stable than the benchmarks under these two adaptive attacks across different data heterogeneity levels. Compared to other Byzantine attacks, Min-Max and Min-Sum lead to a more severe degradation for the benchmark methods. By contrast, BR-DRAG maintains effective convergence. This is consistent with the design rationale of BR-DRAG, where the trusted reference direction $\mathbf{r}^t$ is generated from $\mathcal{D}_{\mathrm{root}}$ and the uploaded updates are corrected through the divergence-based normalization/correction rule, as opposed to relying on pairwise distances among client updates; see (13) and (15).

Among the benchmarks, FLTrust remains more robust than FedAvg but is still inferior to BR-DRAG, since its ReLU-clipped trust weighting provides only coarse discrimination under these adaptive attacks, where malicious updates are deliberately designed to blend in with benign updates under distance-based measures.
RFA and RAGA cannot perform effective training under both attacks, since Min-Max and Min-Sum are specially designed to evade defense based on distance measures, and thus cause marked degradation to these distance-based robust aggregation methods. By contrast, BR-DRAG does not depend on pairwise-distance comparison among the uploaded updates. Instead, it anchors the aggregation to the trusted reference direction $\mathbf{r}^t$ and corrects both direction and magnitude of every update received before aggregation, thereby improving its robustness under these stronger adaptive attacks.
}

On the other hand, no benchmarks maintain stability under the diverse Byzantine attacks. RFA and RAGA can approach BR-DRAG in performance since the applied Weiszfeld algorithm can resist the interference of extreme gradient values. The ReLU-clipping design in FLTrust lacks fine-grained control, as it would entirely preclude the attackers or workers with excessively biased local data. 
In contrast, a key advantage of BR-DRAG is that it retains only the directionally consistent component of each received update with respect to $\mathbf{r}^t$, and suppresses misalignment.
As shown in Figs.~\ref{fig_CIFAR10_BRDRAG_NI_more_attacker} to \ref{fig_CIFAR10_BRDRAG_LF_more_attacker}, even when the proportion of Byzantine workers increases to $\frac{A}{M}=60\%$, BR-DRAG maintains convergence. This surpasses all benchmarks, which typically require $\frac{A}{M} < 50\%$, validating the analysis in Section V-C.
The benchmarks based on the geometric median, i.e., RFA and RAGA, suffer severe degradation, as the high proportion of malicious workers significantly distorts the centroid estimation.
In contrast, BR-DRAG converges by regulating the magnitude and direction of malicious updates, demonstrating its distinctive robustness under high data heterogeneity and severe Byzantine attacks.


\begin{figure}[!t]
	\centering
	\subfloat[{\centering $\beta=0.1$.}]{\includegraphics[width=1.65in]{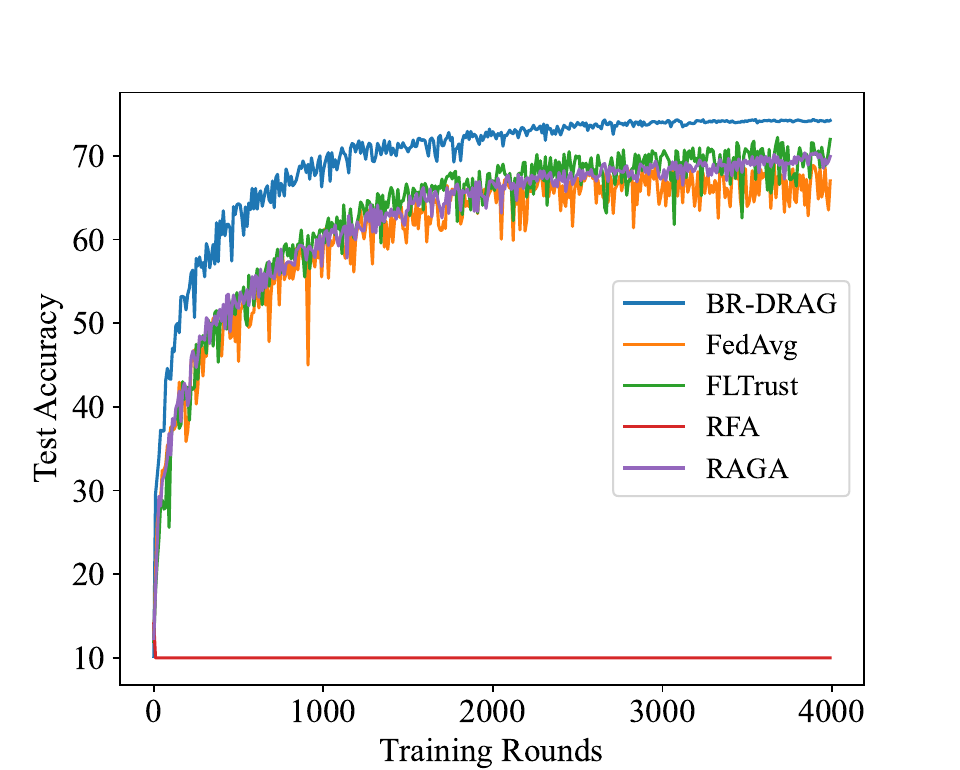}
		\label{fig_CIFAR10_BRDRAG_NI1}}
	\subfloat[{\centering $\beta=0.5$.}]{\includegraphics[width=1.65in]{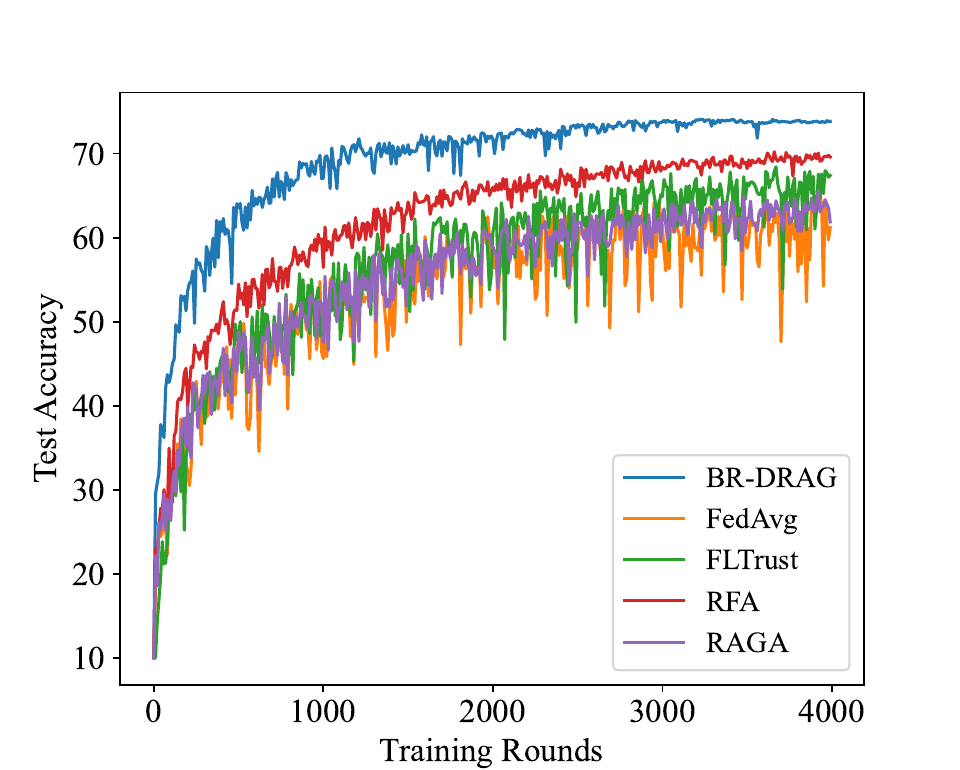}
		\label{fig_CIFAR10_BRDRAG_NI2}}
	\caption{Convergence performance of different algorithms on CIFAR-10 under the noise injection attack.}
	\label{fig_CIFAR10_BRDRAG_NI}
\end{figure}



\begin{figure}[!t]
	\centering
	\subfloat[{\centering $\beta=0.1$.}]{\includegraphics[width=1.65in]{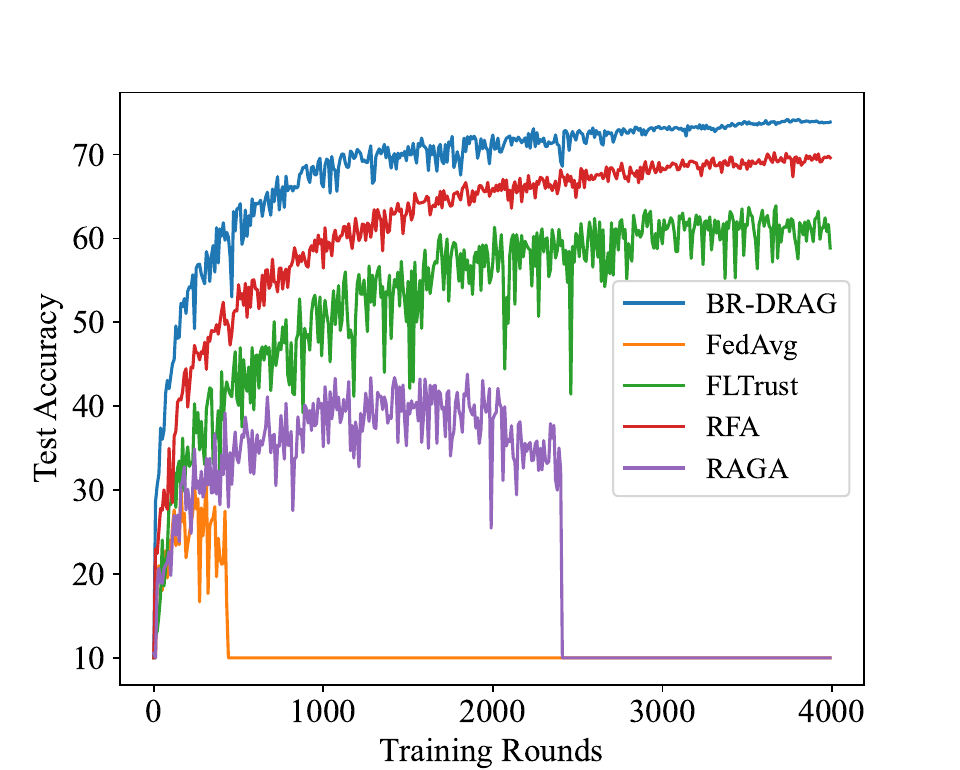}
		\label{fig_CIFAR10_BRDRAG_SF1}}
	\subfloat[{\centering $\beta=0.5$.}]{\includegraphics[width=1.65in]{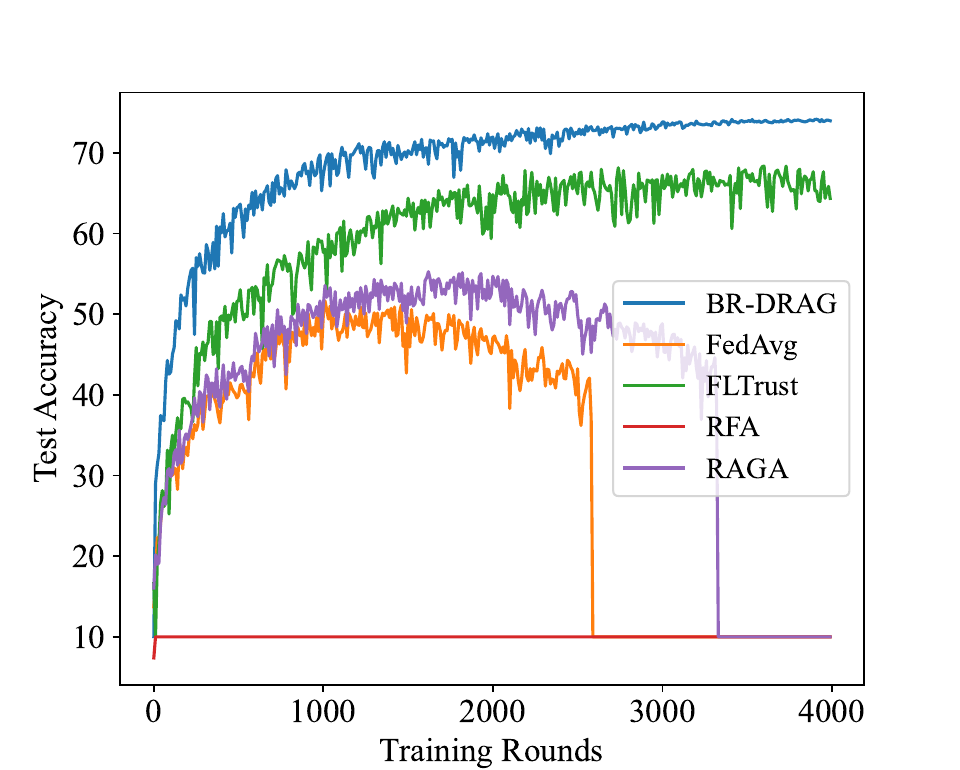}
		\label{fig_CIFAR10_BRDRAG_SF2}}
	\caption{Convergence performance of different algorithms on CIFAR-10 under the sign flipping attack.}
	\label{fig_CIFAR10_BRDRAG_SF}
\end{figure}





\begin{figure}[!t]
	\centering
	\subfloat[{\centering $\beta=0.1$.}]{\includegraphics[width=1.65in]{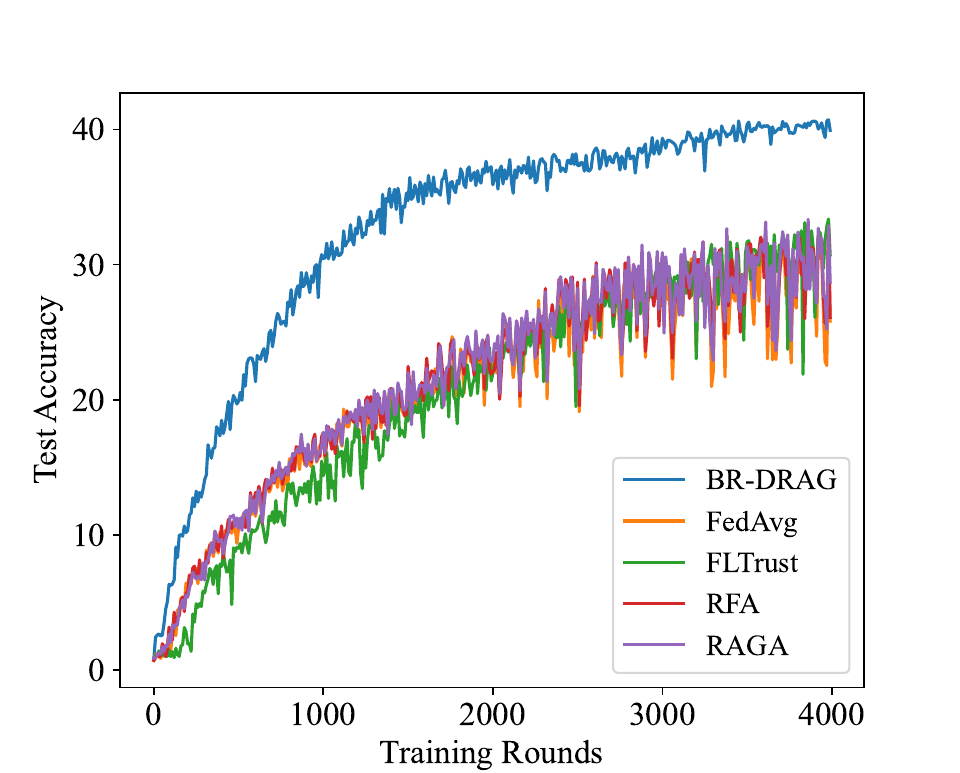}
		\label{fig_CIFAR100_BRDRAG_LF1}}
	\subfloat[{\centering $\beta=0.5$.}]{\includegraphics[width=1.65in]{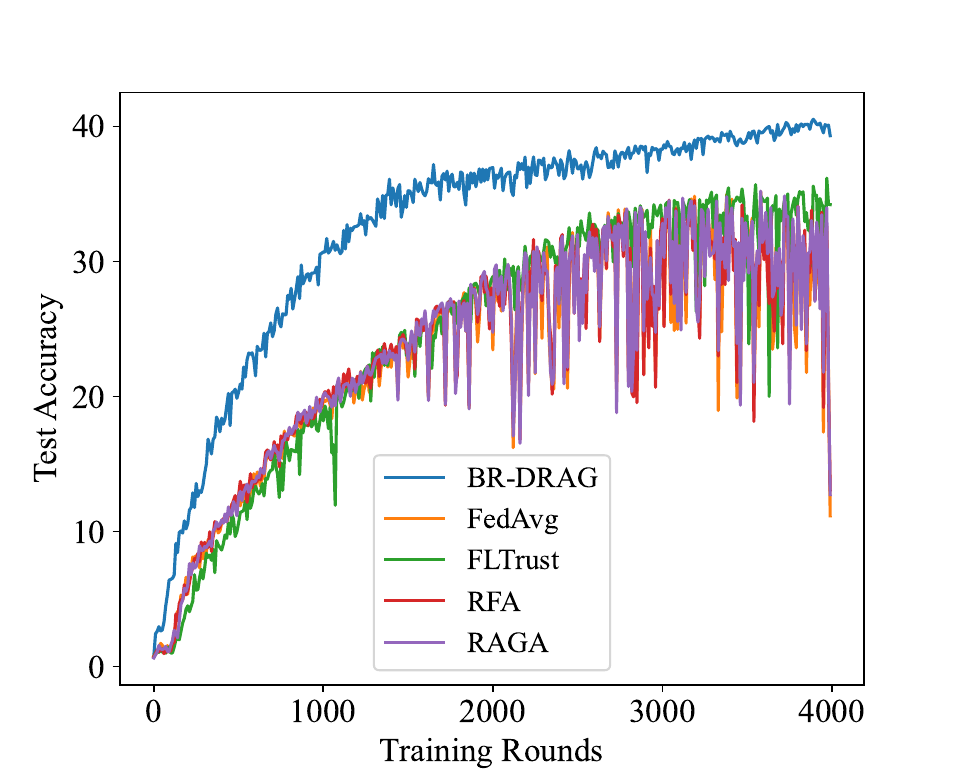}
		\label{fig_CIFAR100_BRDRAG_LF2}}
	\caption{Convergence performance of different algorithms on CIFAR-100 under the label flipping attack.}
	\label{fig_CIFAR100_BRDRAG_LF}
\end{figure}

\begin{figure}[!t]
	\centering
	\subfloat[ $\beta=0.1$. ]{\includegraphics[width=1.6in]{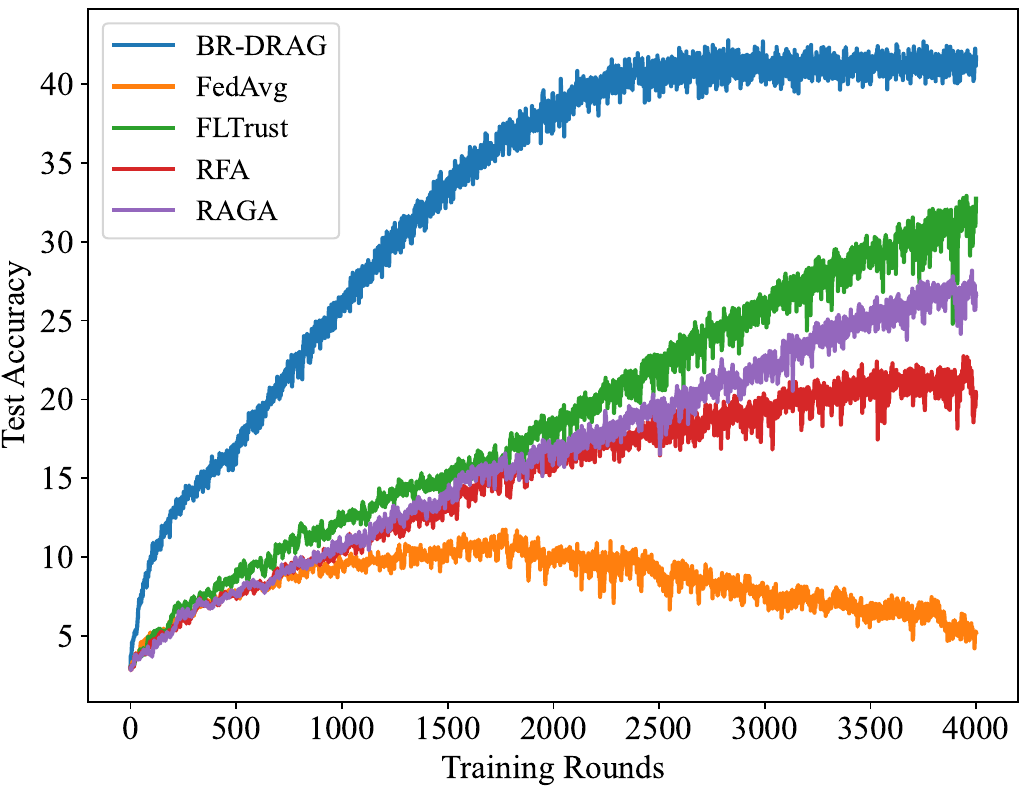}
		\label{fig_wrn2_1}}
	\hfill 
	\subfloat[$\beta=0.5$.]{\includegraphics[width=1.6in]{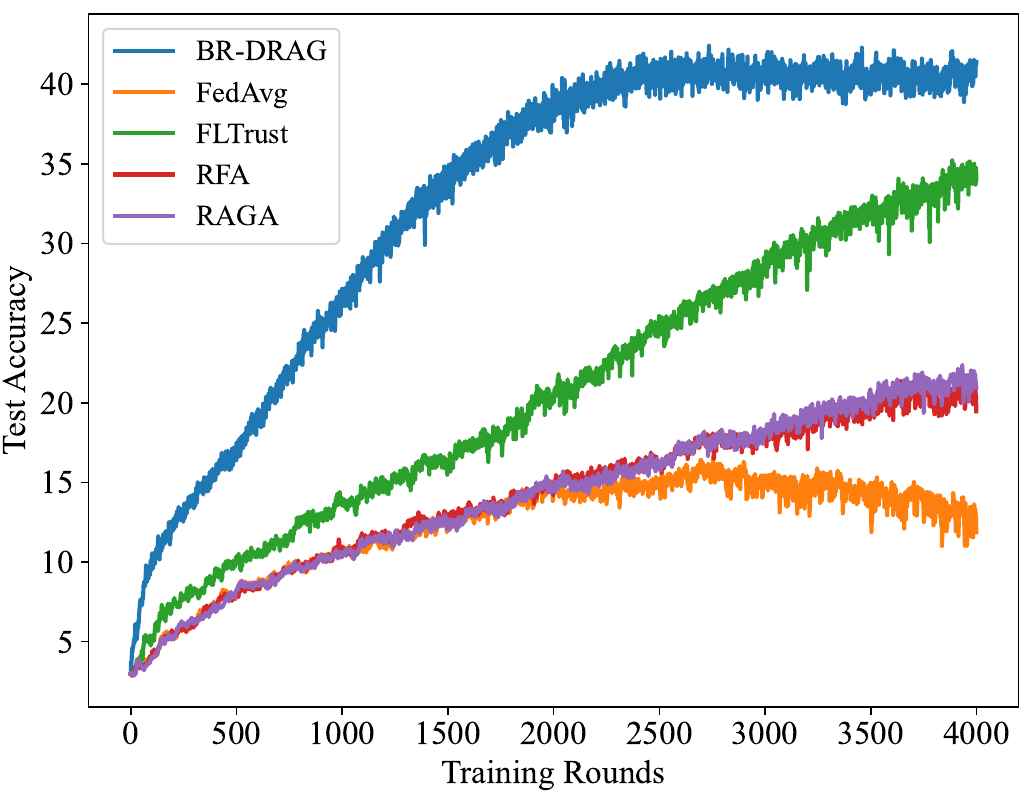}
		\label{fig_wrn2_2}}
	\caption{Convergence performance of different algorithms on CIFAR-100 with the WRN-28-4 model under the sign flipping attack.}
	\label{wrn2}
    \end{figure}

\begin{figure}[!t]
	\centering
	\subfloat[ $\beta=0.1$. ]{\includegraphics[width=1.6in]{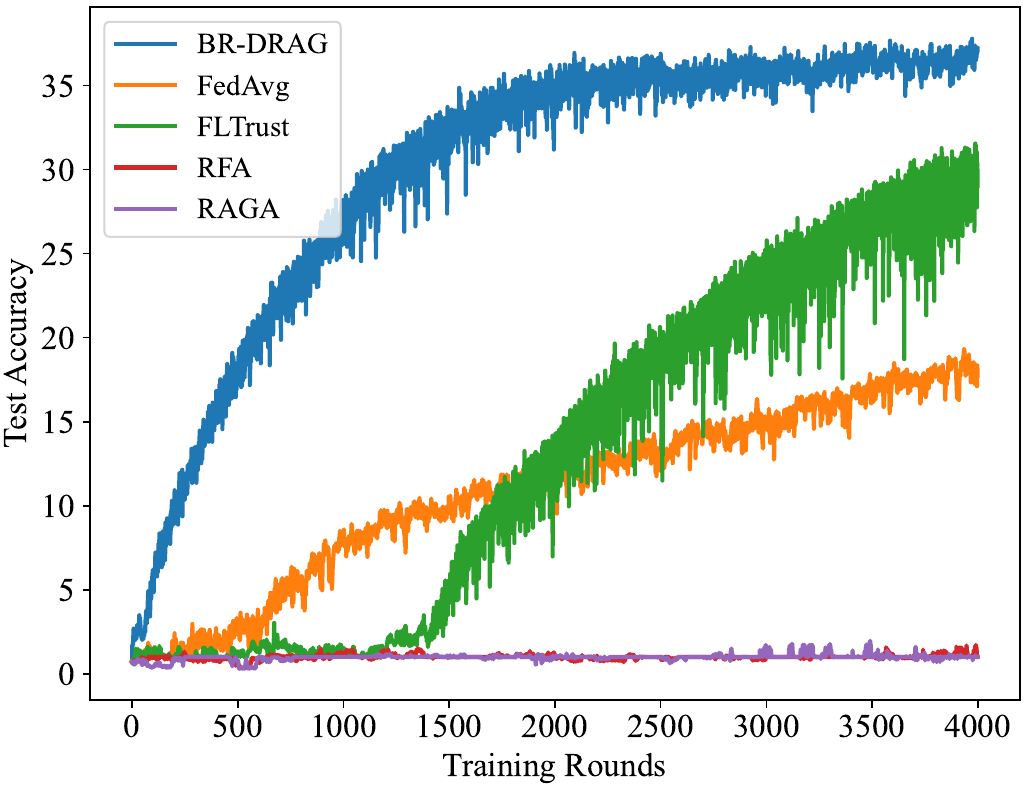}
		\label{fig_minmax1_1}}
	\hfill 
	\subfloat[$\beta=0.5$.]{\includegraphics[width=1.6in]{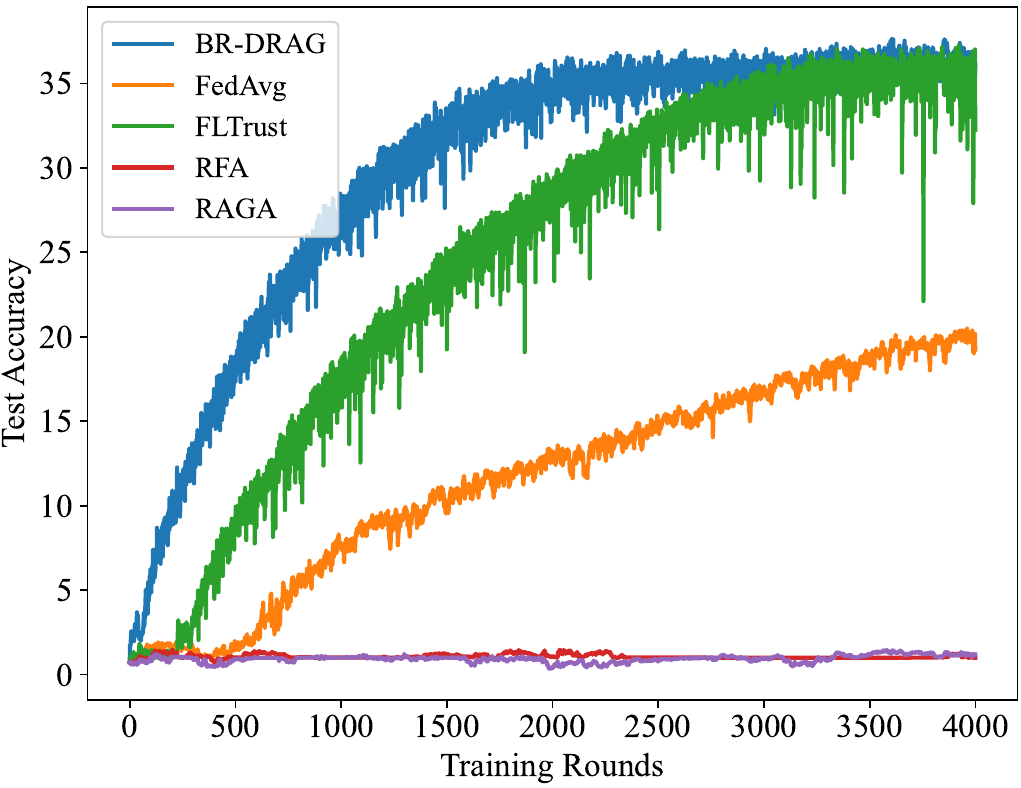}
		\label{fig_minmax1_2}}
	\caption{Convergence performance of different algorithms on CIFAR-100 under
the Min-Max attack.}
	\label{minmax1}
    \end{figure}

        \begin{figure}[!t]
	\centering
	\subfloat[ $\beta=0.1$. ]{\includegraphics[width=1.6in]{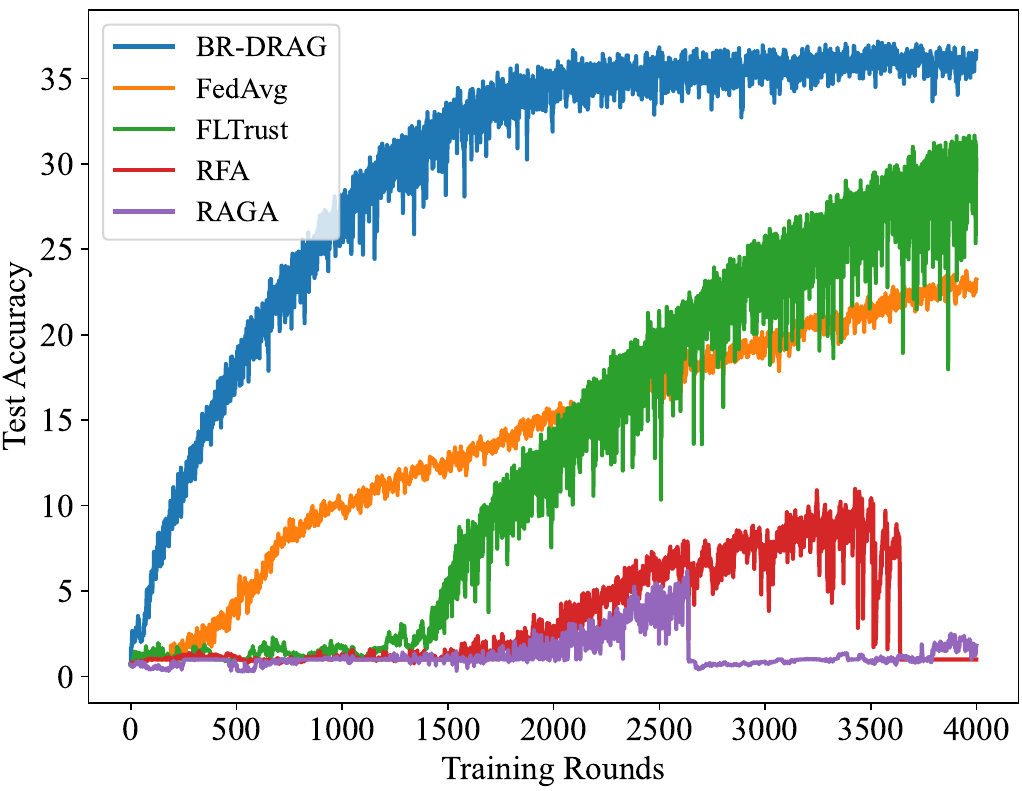}
		\label{fig_minsum1_1}}
	\hfill 
	\subfloat[$\beta=0.5$.]{\includegraphics[width=1.6in]{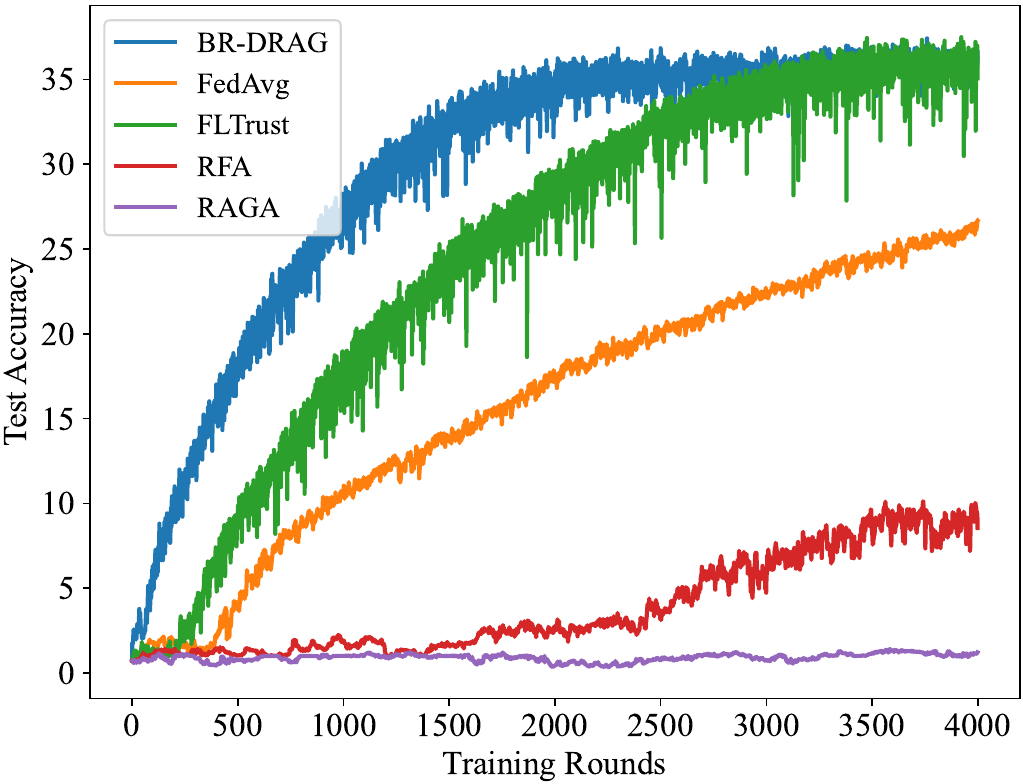}
		\label{fig_minsum1_2}}
	\caption{Convergence performance of different algorithms on CIFAR-100 under
the Min-Sum attack.}
	\label{minsum1}
    \end{figure}

\begin{figure}[!t]
	\centering
	\subfloat[{\centering $\beta=0.1$.}]{\includegraphics[width=1.65in]{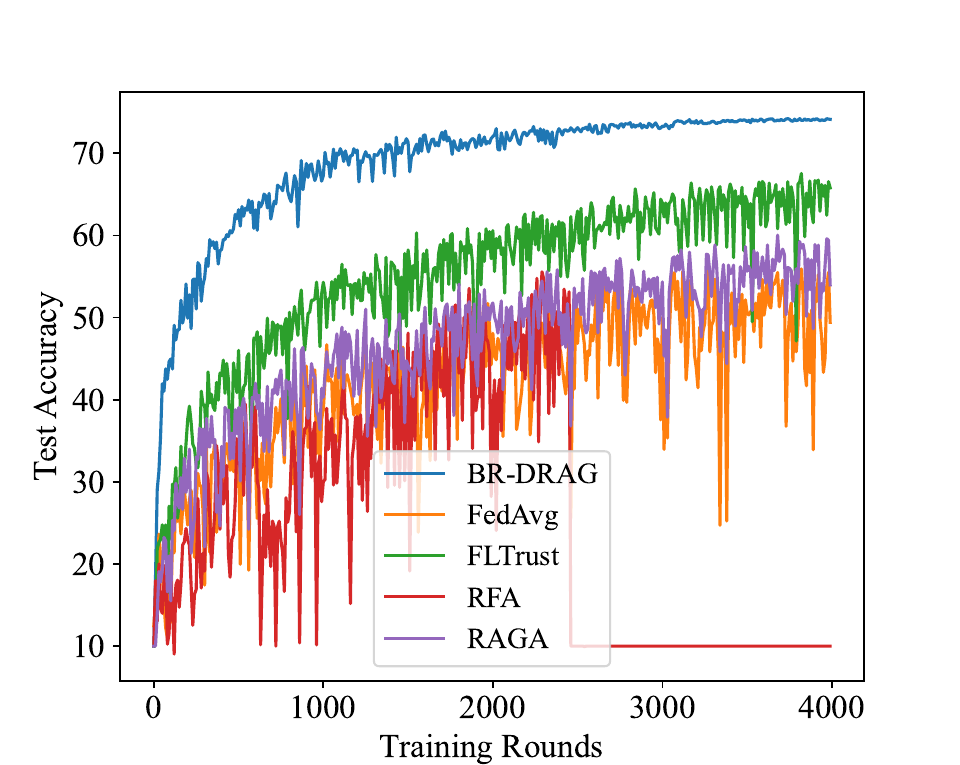}
		\label{fig_CIFAR10_BRDRAG_NI1_more_attacker}}
	\subfloat[{\centering $\beta=0.5$.}]{\includegraphics[width=1.65in]{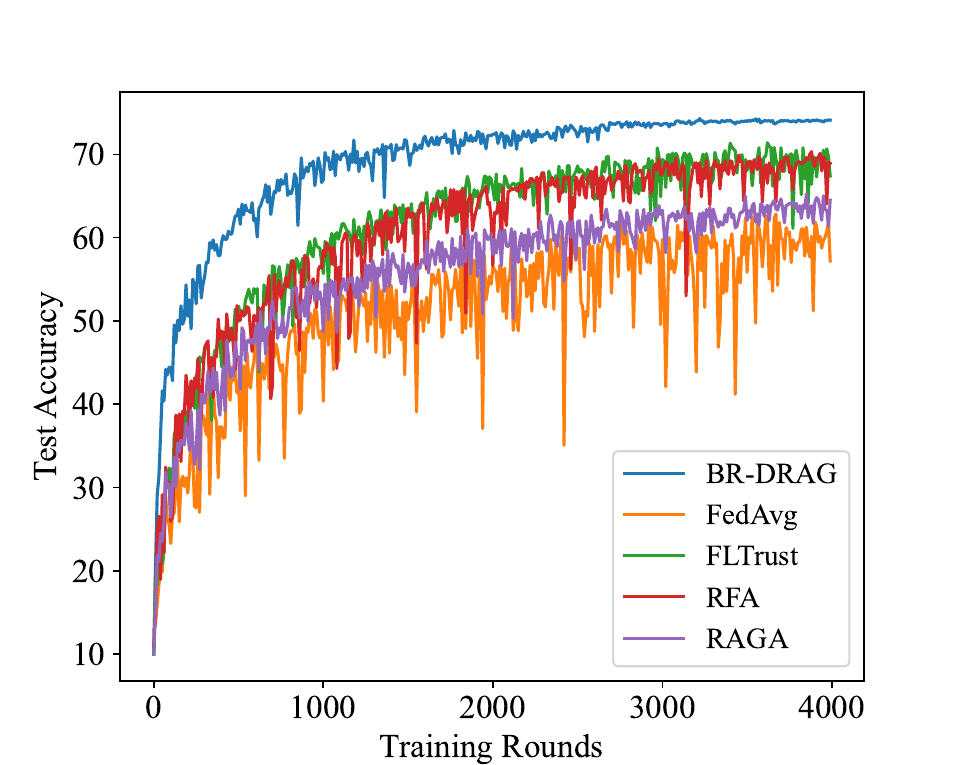}
		\label{fig_CIFAR10_BRDRAG_NI2_more_attacker}}
	\caption{Convergence performance of different algorithms on CIFAR-10 under the noise injection attack, where 60\% of the workers are attackers.}
	\label{fig_CIFAR10_BRDRAG_NI_more_attacker}
\end{figure}

\begin{figure}[!t]
	\centering
	\subfloat[{\centering $\beta=0.1$.}]{\includegraphics[width=1.65in]{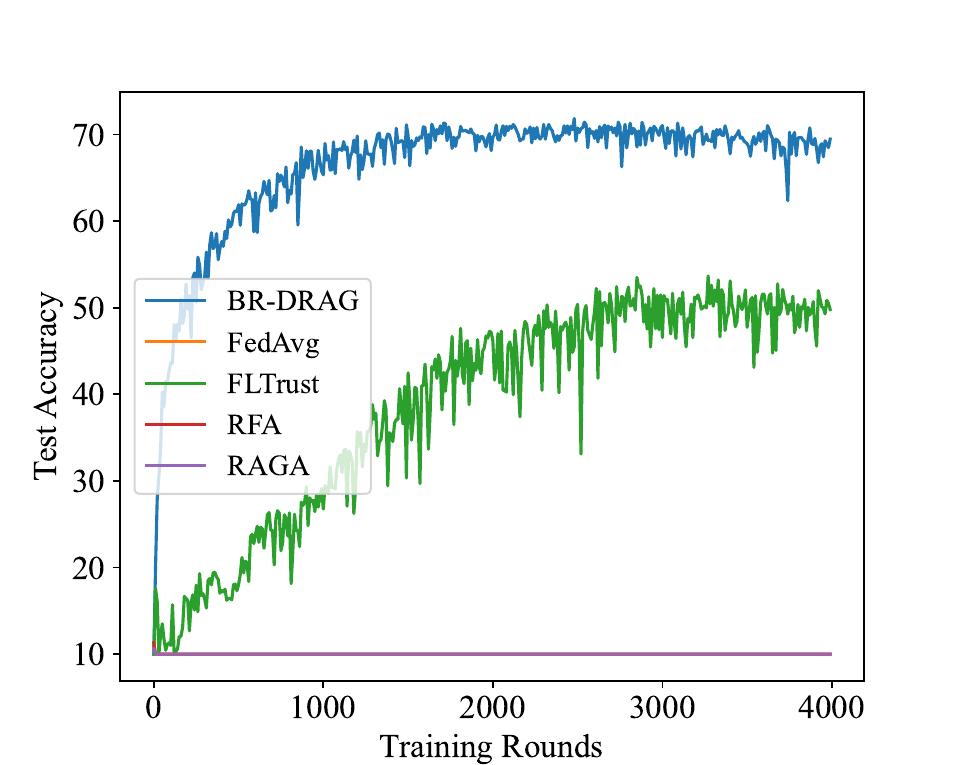}
		\label{fig_CIFAR10_BRDRAG_SF1_more_attacker}}
	\subfloat[{\centering $\beta=0.5$.}]{\includegraphics[width=1.65in]{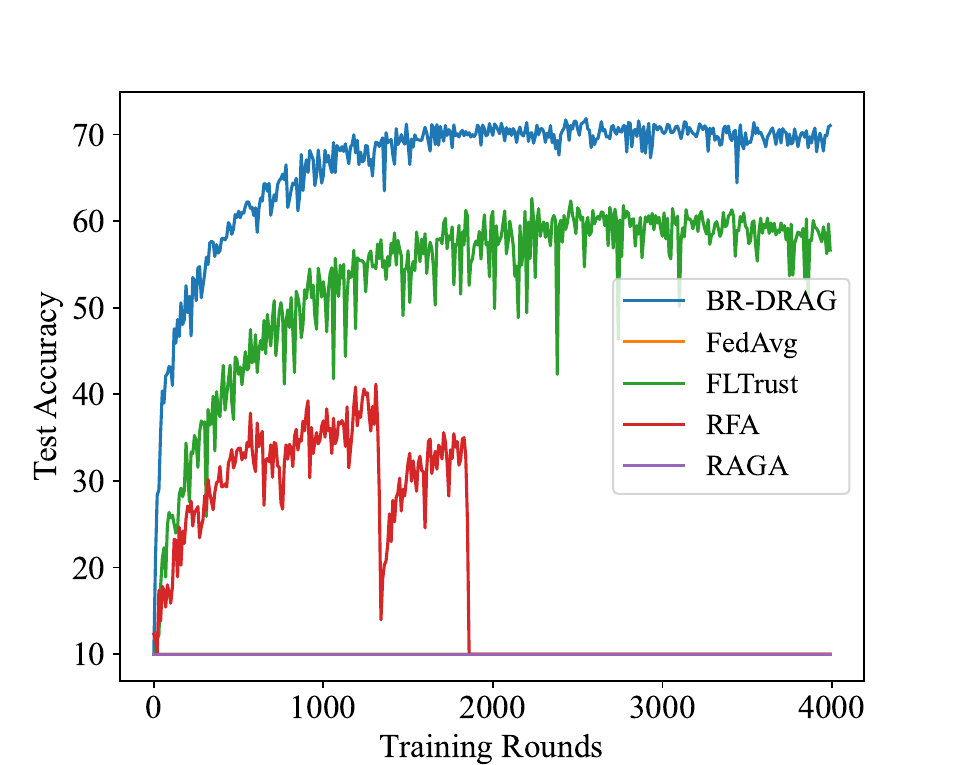}
		\label{fig_CIFAR10_BRDRAG_SF2_more_attacker}}
	\caption{Convergence performance of different algorithms on CIFAR-10 under the sign flipping attack, where 60\% of the workers in the system are attackers.}
	\label{fig_CIFAR10_BRDRAG_SF_more_attacker}
\end{figure}

\begin{figure}[!t]
	\centering
	\subfloat[{\centering $\beta=0.1$.}]{\includegraphics[width=1.65in]{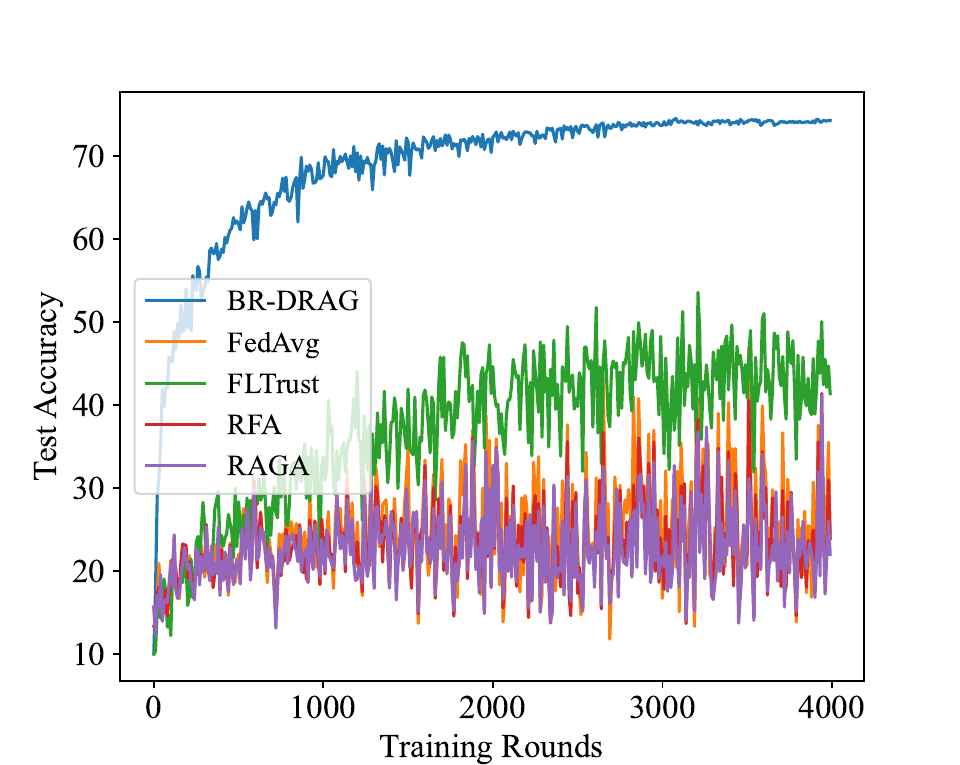}
		\label{fig_CIFAR10_BRDRAG_LF1_more_attacker}}
	\subfloat[{\centering $\beta=0.5$.}]{\includegraphics[width=1.65in]{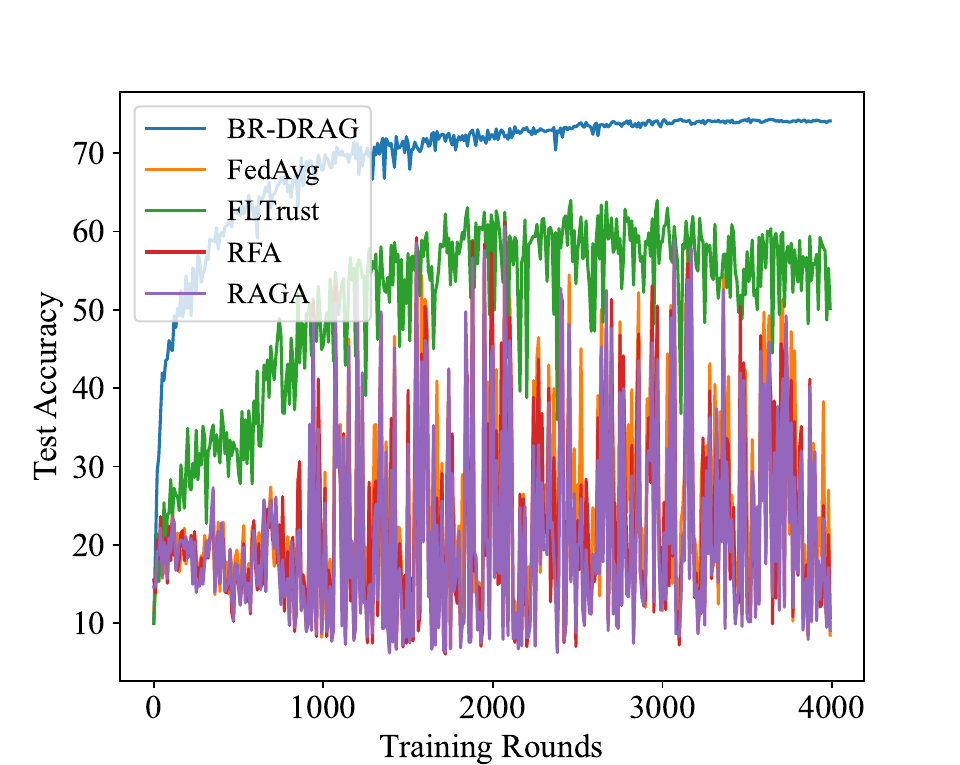}
		\label{fig_CIFAR10_BRDRAG_LF2_more_attacker}}
	\caption{Convergence performance of different algorithms on CIFAR-10 under the label flipping attack, where 60\% of the workers in the system are attackers.}
	\label{fig_CIFAR10_BRDRAG_LF_more_attacker}
\end{figure}

{\color{black}
We further experimentally corroborate that the computational cost of the PS is only slightly above that of a local client in each training round.
    Under the setting that the local clients and the PS employ computing hardware with the same capabilities, we record the computing time of the PS, and the average computing time of the selected devices per training round under different settings.
    Fig.~\ref{time} shows the ratio of the additional PS-side computation time, including the generation of the trusted reference direction $\mathbf{r}^t$ and the computation of the modified gradients $\mathbf{v}_m^t$, to the average local training time of the selected devices per round. 
    For BR-DRAG, it is noted that the time required to compute $\mathbf{r}^t$ and $\mathbf{v}_m^t$ at the PS is, on average, 139\% and 153\% of that required to compute local updates at the devices for CIFAR-10 and CIFAR-100, respectively, when the global and local computation capabilities are consistent. In practice, since the PS is typically equipped with much more powerful computing hardware than the devices, the time required to compute $\mathbf{r}^t$ would be comparatively negligible.
}
 \begin{figure}[!t]
	\centering
	\subfloat[ CIFAR-10. ]{\includegraphics[width=1.6in]{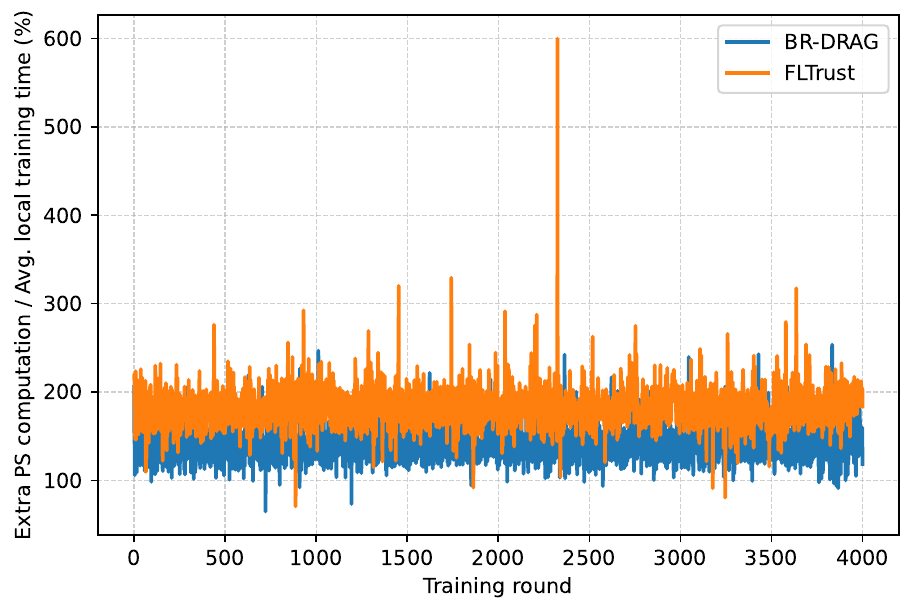}
		\label{fig_time_1}}
	\hfill 
	\subfloat[CIFAR-100.]{\includegraphics[width=1.6in]{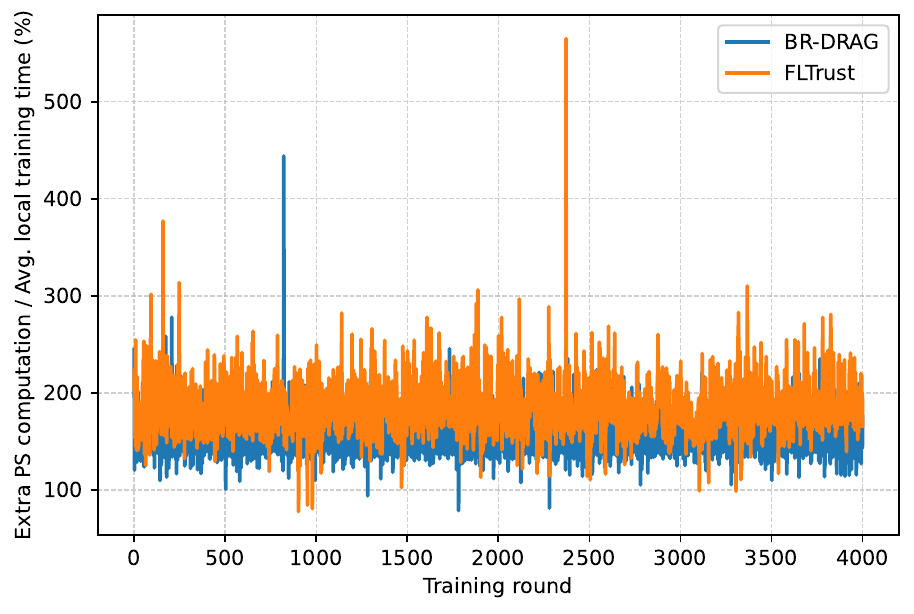}
		\label{fig_time_2}}
	\caption{The ratio of the computation time in (14)--(16) to the average local computation time of the selected devices of BR-DRAG under the noise injection attack in each training round.}
	\label{time}
    \end{figure}

To test the sensitivity of BR-DRAG to the root dataset $\mathcal{D}_{\mathrm{root}}$, for the CIFAR-100 dataset, we consider the size of $|\mathcal{D}_{\mathrm{root}}| = 2000, 3000, 4000, 5000, 6000 $. For $\mathcal{D}_{\mathrm{root}}$ with $|\mathcal{D}_{\mathrm{root}}| = 3000$ data samples, we mix portions of the data into the dataset subjected to a noise injection attack at ratios of \{0, 2.5\%, 5\%, 7.5\%, 10\%\} to further test the impact of data quality on training performance. To simulate distribution inconsistency, we also induce a non-IID data distribution by removing all samples corresponding to classes 1, 2, and 3 from the dataset.
    Fig.~\ref{sensitivity} shows the sensitivity of BR-DRAG to the size, quality, class coverage, and distribution mismatch of $\mathcal{D}_{\mathrm{root}}$.
    It is observed in Fig. \ref{sensitivity}(a) that when the size of $\mathcal{D}_{\mathrm{root}}$ decreases, the consistency between the constructed reference direction $\mathbf{r}^t$ and the true global gradient is reduced, which, in turn, penalizes the generalization ability of BR-DRAG and leads to decreased accuracy.
    
    {\color{black}
Regarding data quality, contaminated samples in $\mathcal{D}_{\mathrm{root}}$ (e.g., noisy labels) can misalign $\mathbf{r}^t$ with the true descent direction. Fig.~\ref{sensitivity}(b) evaluates this sitation by injecting noise-injection-attacked data into $\mathcal{D}_{\mathrm{root}}$  
  at increasing ratios. Corrupted data degrades convergence, as the model is forced to learn spurious features from 
  the contaminated reference direction. Nevertheless, BR-DRAG consistently outperforms FLTrust. For instance, BR-DRAG   
  with 5\% corrupted data retains comparable accuracy with FLTrust operating under a clean $\mathcal{D}_{\mathrm{root}}$. This resilience is attributable to the proposed DoD-based correction mechanism, which suppresses the misaligned component of each received update.
    
 The root dataset $\mathcal{D}_{\text{root}}$ may have incomplete class coverage or a skewed    
  distribution, compared to the clients' local data. Fig.~\ref{sensitivity}(c) simulates the scenario by removing all samples of classes 1, 2, and 3 
  from $\mathcal{D}_{\mathrm{root}}$ while keeping the client data intact. The accuracy of BR-DRAG decreases under this
   mismatch, as the reference direction $\mathbf{r}^t$ constructed based on the incomplete $\mathcal{D}_{\mathrm{root}}$ via (12) deviates from the 
  true global gradient for the omitted classes and the correction in (15) inherits this bias. Nevertheless, BR-DRAG with this skewed $\mathcal{D}_{\mathrm{root}}$ outperforms FLTrust operating even under an unbiased $\mathcal{D}_{\mathrm{root}}$. This advantage stems from the DoD-based interpolation described in (15): Rather than discarding or scaling the entire update based on a trust score, BR-DRAG retains the aligned component of each local update even when $\mathbf{r}^t$ is biased, thereby partially compensating for the distribution mismatch.

}
\begin{figure}[htbp]
        \centering
        \subfloat[Sensitivity of $\mathcal{D}_{\mathrm{root}}$ to different data sizes.]{
            \includegraphics[width=1\linewidth]{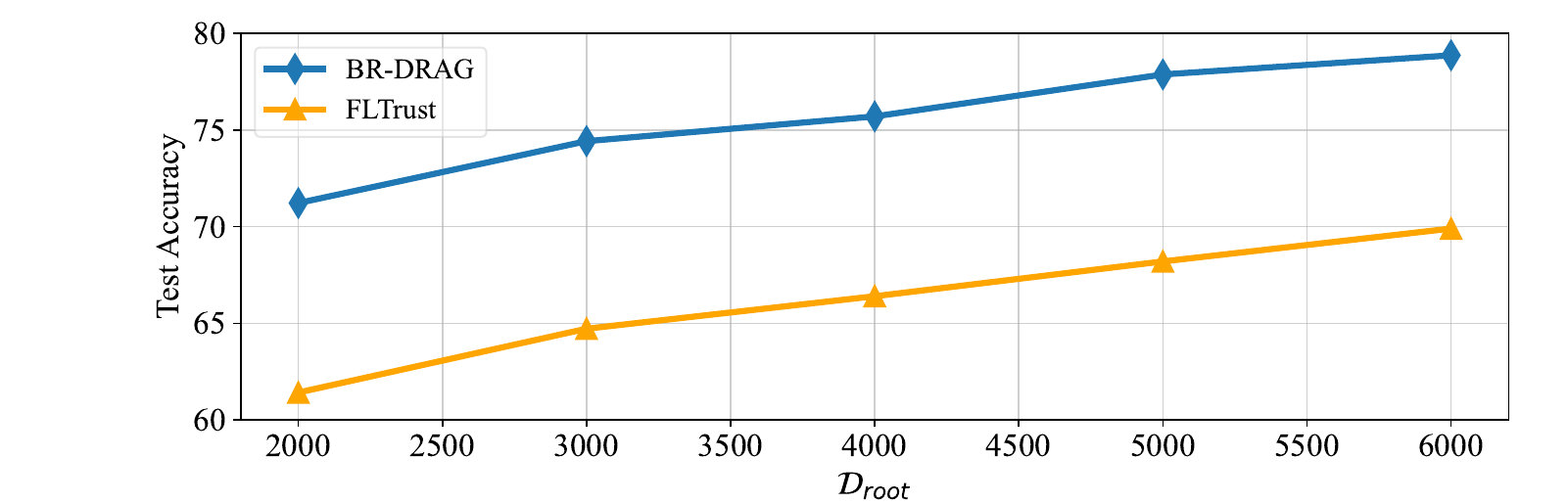}
            \label{fig:sub1}
        }
        \\ 
        \subfloat[Sensitivity of $\mathcal{D}_{\mathrm{root}}$ to corrupted data.]{
            \includegraphics[width=1\linewidth]{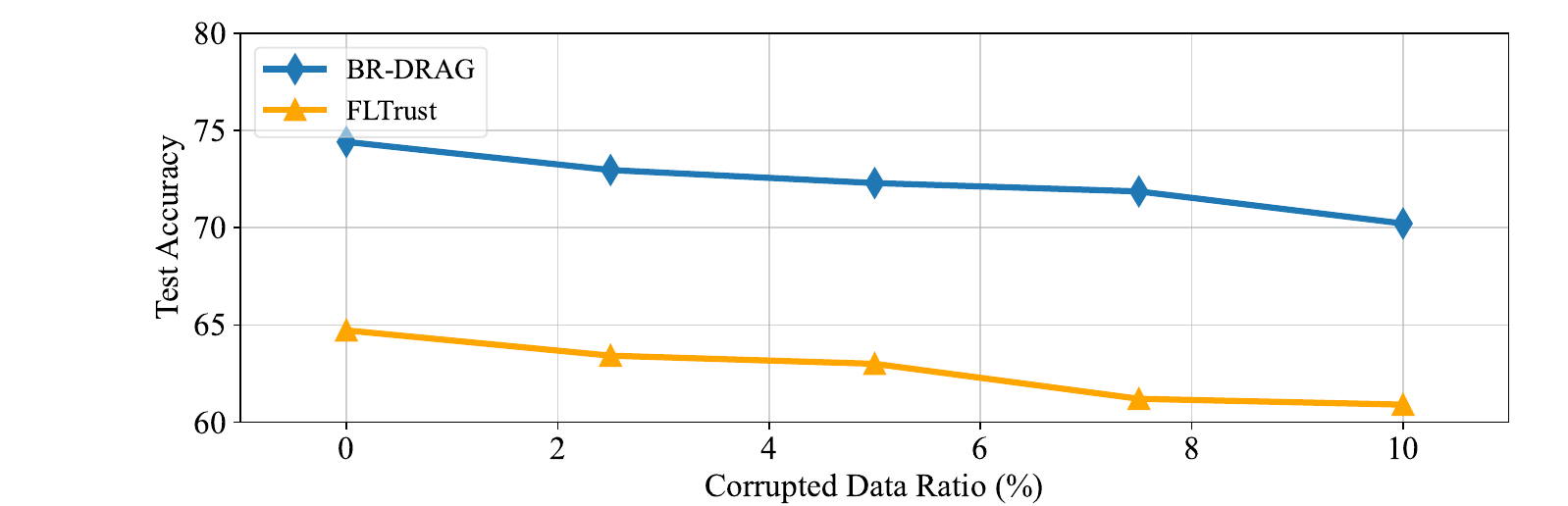}
            \label{fig:sub2}
        }
        \\ 
        \subfloat[Sensitivity of $\mathcal{D}_{\mathrm{root}}$ to non-IID distributions.]{
            \includegraphics[width=1\linewidth]{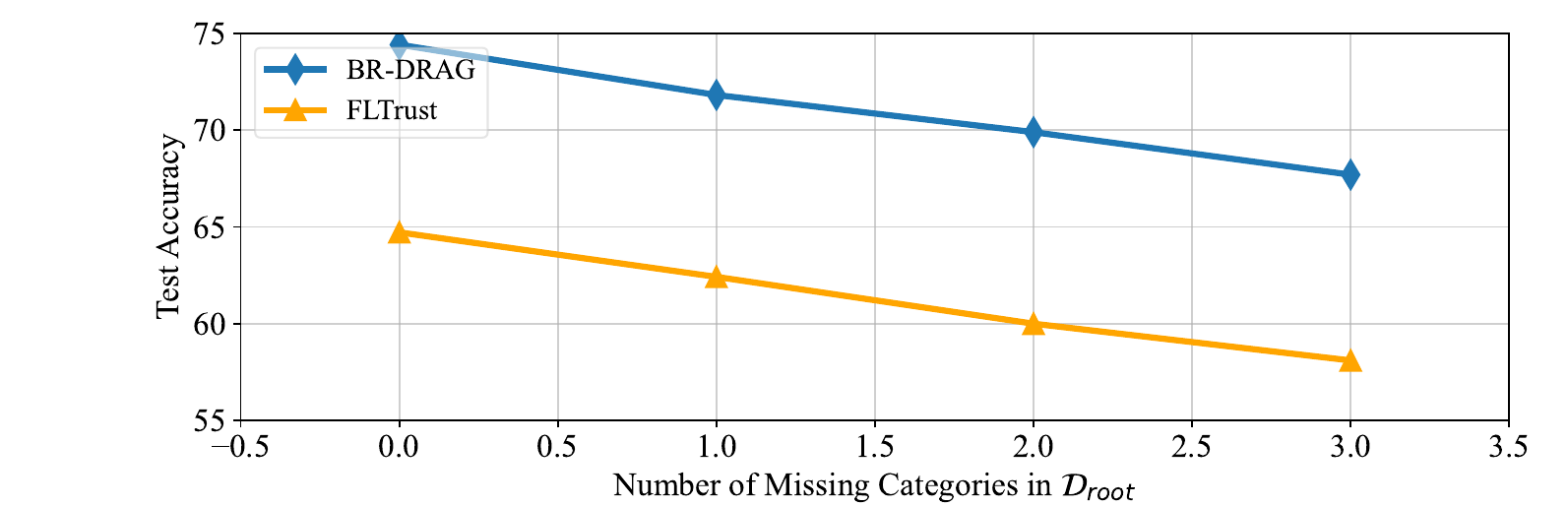}
            \label{fig:sub3}
        }
        \caption{Sensitivity of BR-DRAG to $\mathcal{D}_{\mathrm{root}}$ on CIFAR-10 with $\beta=0.1$ and noise injection attack.}
        \label{sensitivity}
    \end{figure}
 
    {\color{black}
    
To examine the benefit of the DoD-based correction in BR-DRAG, we replace it with FLTrust's trust-score mechanism; i.e., (15) is replaced by $ \mathbf{v}_m^t = \max\!\big(0, x_m^t\big) \cdot \frac{\|\mathbf{r}^t\|}{\|\mathbf{g}_m^t\|}\mathbf{g}_m^t$.
All other components of BR-DRAG, including the trusted root dataset $\mathcal{D}_{\mathrm{root}}$ construction, the trusted reference direction $\mathbf{r}^t$ computation via (13), the norm normalization $\|\mathbf{r}^t\|/\|\mathbf{g}_m^t\|$, and the global model update rule (14), remain unchanged.
Fig.~\ref{ablation} shows the convergence performance of BR-DRAG and its counterpart with DoD-based correction replaced by FLTrust's trust-score mechanism, under sign flipping attacks on the CIFAR-10 and CIFAR-100 datasets with $\beta=0.1$. 
BR-DRAG with the new DoD-based correction consistently outperforms BR-DRAG with the clipped cosine similarity-based trust score mechanism (adopted from FLTrust \cite{cao2020fltrust}) across different datasets. This is because the trust-score mechanism adopted in FLTrust scales the entire normalized update by $\max(0, x_m^t)$ without correcting its directional deviation from $\mathbf{r}^t$. When $x_m^t$ is positive but small (i.e., moderate misalignment), the trust-score mechanism retains the update in its original direction with a reduced weight, while the proposed DoD correction actively steers the update toward $\mathbf{r}^t$ through the interpolation term $\lambda_m^t\mathbf{r}^t$, thereby achieving tighter alignment with the global objective. When $x_m^t < 0$, the trust-score mechanism zeros out the update entirely, discarding potentially useful gradient information. By contrast, the DoD correction reverses the misaligned component and leverages it effectively.
    }
\begin{figure}[!t]
    \centering
    \subfloat[Sign flipping, CIFAR-10.]{\includegraphics[width=1.6in]{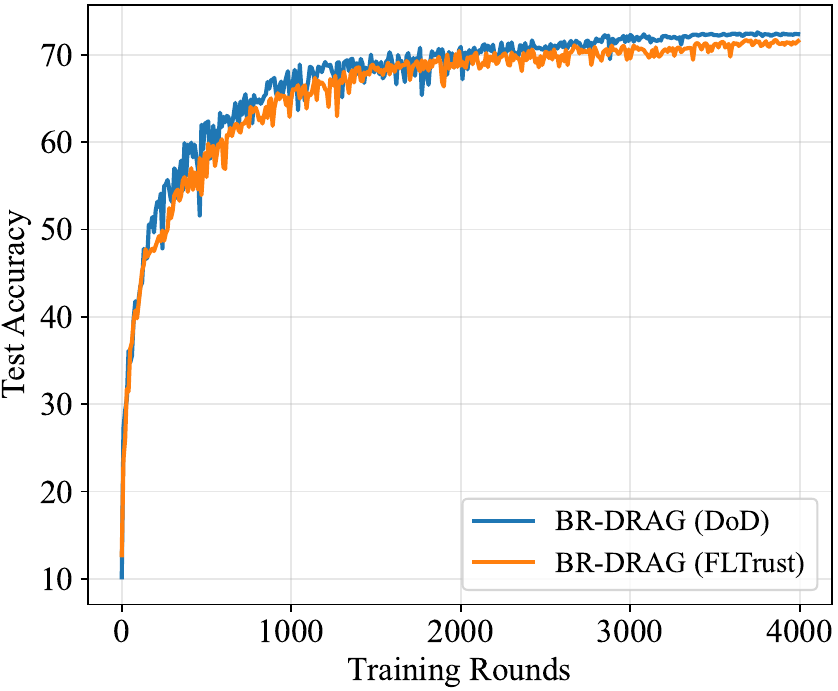}
        \label{fig_abl_1}}
    \subfloat[Sign flipping, CIFAR-100.]{\includegraphics[width=1.6in]{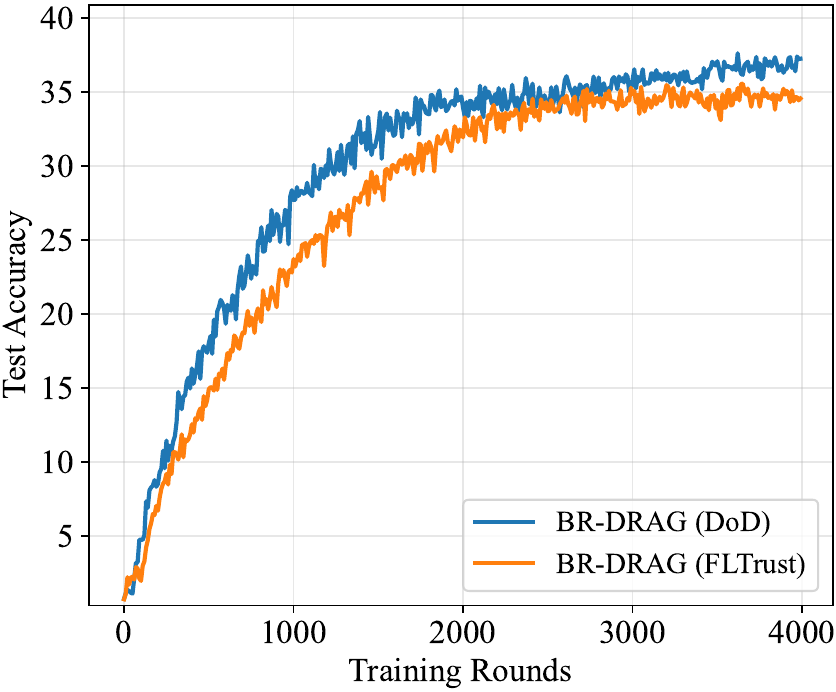}
        \label{fig_abl_2}}
    \caption{
    Comparison between BR-DRAG with DoD-based correction and BR-DRAG with FLTrust's trust-score mechanism under different datasets with $\beta=0.1$.}
    \label{ablation}
\end{figure}

\section{Conclusion and Future Work}


In this paper, we designed DRAG to address client drifts caused by data heterogeneity in FL systems. With the proposed reference direction and DoD, the local updates can be adaptively dragged towards the global update direction.
We further developed BR-DRAG, a Byzantine-resilient variant of DRAG, to resist Byzantine attacks.
Comprehensive convergence analyses of DRAG and BR-DRAG with non-convex loss functions were conducted under data heterogeneity, partial worker participation, and Byzantine attacks.
It was revealed that both DRAG and BR-DRAG achieve fast convergence compared to the existing algorithms.
Numerical results validated the superiority of DRAG in mitigating data heterogeneity and accelerating model convergence, and the robustness of BR-DRAG against malicious attacks.
In the future, we will extend DRAG and BR-DRAG to decentralized FL settings.

\section*{Appendix A \\ Proof of Theorem \ref{drag theorem}}

Based on the $L$-smoothness of $f$, we first have:
\begin{align}
    &\mathbb{E}[f(\boldsymbol{\theta}^{t+1})]  \!\leq \!\!f(\boldsymbol{\theta}^{t}) \!+\! \left\langle\nabla\! f(\boldsymbol{\theta}^{t}), \mathbb{E}\!\left[\boldsymbol{\theta}^{t+1}\!\!\!-\!\boldsymbol{\theta}^t\right]\right\rangle\!\! +\!\! \frac{L}{2} \!\mathbb{E}\!\left[\|\boldsymbol{\theta}^{t+1}\!\!\!\!-\!\!\boldsymbol{\theta}^t\|^2\right]\notag\\
    &=f(\boldsymbol{\theta}^{t}) + \left\langle\nabla f(\boldsymbol{\theta}^{t}), \mathbb{E}\left[\Delta^t+a_1-a_1\right]\right\rangle + \frac{L}{2} \mathbb{E}\left[\|\Delta^t\|^2\right]\notag\\
    & \overset{(a)}{=} f(\boldsymbol{\theta}^{t}) -(1-c)\eta U\mathbb{E}\left[\|\nabla f(\boldsymbol{\theta}^t)\|^2\right]+ \frac{L}{2}T_1+T_2,\label{one-step} 
\end{align}
where $\Delta^t = \boldsymbol{\theta}^{t+1}-\boldsymbol{\theta}^t$, and $a_1=(1-c)\eta U \nabla f(\boldsymbol{\theta}^t)$; $(a)$ is due to $T_1=\mathbb{E}\left[\|\Delta^t\|^2\right]$ and $T_2=\left\langle\nabla f(\boldsymbol{\theta}^{t}), \mathbb{E}\left[\Delta^t+a_1\right]\right\rangle$. 

Next, we bound $T_1$ and $T_2$. For $T_1$, we have
\begin{align}  \label{T1}
    & T_1 = \mathbb{E}\Bigg[\Big\|\frac{1}{S}\sum_{m\in\mathcal{S}^t}\big((1-\lambda_m^t)\mathbf{g}_m^t+\frac{\lambda_m^t\|\mathbf{g}_m^t\|}{\|\mathbf{r}^t\|}\mathbf{r}^t\big)\Big\|^2\Bigg]\notag\\
    &\overset{(a)}{\leq}\mathbb{E}\Big[\Big(\frac{1}{S}\sum_{m\in\mathcal{S}^t}\big((1-\lambda_m^t)\|\mathbf{g}_m^t\|+\frac{\lambda_m^t\|\mathbf{g}_m^t\|}{\|\mathbf{r}^t\|}\|\mathbf{r}^t\|\big)\Big)^2\Big]\notag\\
    &\overset{(b)}{\leq}\frac{\eta^2}{S}\sum_{m\in\mathcal{S}^t}\mathbb{E}\Big[\big\|\sum_{u=0}^{U-1}\nabla F(\boldsymbol{\theta}_m^{t,u};z_{m}^{t,u})\big\|^2\Big]\notag\\
    &\overset{(c)}{=}\frac{\eta^2}{S}\sum_{m\in\mathcal{S}^t} \Big( \mathbb{E}\Big[\big\|\sum_{u=0}^{U-1} a_2\big\|^2\Big] + \mathbb{E}\Big[\big\|\sum_{u=0}^{U-1}\nabla F_m(\boldsymbol{\theta}_m^{t,u})\big\|^2\Big] \Big) \notag\\
    &\overset{(d)}{=} {\eta^2U\sigma_L^2} +\frac{\eta^2}{M}\sum_{m\in\mathcal{M}}\mathbb{E}\Big[\big\|\sum_{u=0}^{U-1}\nabla F_m(\boldsymbol{\theta}_m^{t,u})\big\|^2\Big] ,
\end{align}
where $a_2=\nabla F_m(\boldsymbol{\theta}_m^{t,u};z_{m,b}^{t,u})-\nabla F_m(\boldsymbol{\theta}_m^{t,u})$; $(a)$ is due to the triangle inequality; $(b)$ is based on the definition of $\mathbf{g}_m^t$ and Cauchy-Schwartz inequality; $(c)$ follows from $\mathbb{E}[\|\boldsymbol{x}\|^2]=\mathbb{E}[\|\boldsymbol{x}-\mathbb{E}[\boldsymbol{x}]\|^2+\|\mathbb{E}[\boldsymbol{x}]\|^2]$ and \textbf{Assumption 2};
$(d)$ comes from $\mathbb{E}[\|\boldsymbol{x}_1\!\!+\!\ldots\!+\!\boldsymbol{x}_n\|^2]\!\!=\!\mathbb{E}[\|\boldsymbol{x}_1\|^2\!\!+\!\!\ldots\!+\!\!\|\boldsymbol{x}_n\|^2]$ if $\boldsymbol{x}_i,\,\forall i$, is zero-mean and independent, \textbf{Assumption 2}, and UAR worker sampling.

We can rewrite $T_2$ as
\begin{align} \label{bound of T2}
    & T_2\overset{(a)}{=}\Bigg\langle\nabla f(\boldsymbol{\theta}^{t}), \mathbb{E}\Bigg[(1-c)\mathbf{g}^t+\frac{c}{S}\sum_{m\in\mathcal{S}^t}\frac{\left\langle\mathbf{g}_m^t, \mathbf{r}^t\right\rangle}{\|\mathbf{g}_m^t\|\|\mathbf{r}^t\|}\mathbf{g}_m^t\nonumber\\
    &+\frac{c}{S}\sum_{m\in\mathcal{S}^t}\frac{\|\mathbf{g}_m^t\|\|\mathbf{r}^t\|-\left\langle\mathbf{g}_m^t, \mathbf{r}^t\right\rangle}{\|\mathbf{r}^t\|^2}\mathbf{r}^t+(1-c)\eta U \nabla f(\boldsymbol{\theta}^t)\Bigg]\Bigg\rangle\notag\\
    &\overset{(b)}{=} \underbrace{\left\langle\nabla f(\boldsymbol{\theta}^{t}), \mathbb{E}\left[(1-c)\mathbf{g}^t+a_1\right]\right\rangle}_{T_{2,1}} + T_{2,2} + T_{2,3},
\end{align}
where $T_{2,2} = \langle\nabla f(\boldsymbol{\theta}^{t}), \mathbb{E}[\frac{c}{M}\sum_{m\in\mathcal{S}^t}\frac{\left\langle\mathbf{g}_m^t, \mathbf{r}^t\right\rangle}{\|\mathbf{g}_m^t\|\|\mathbf{r}^t\|}\mathbf{g}_m^t]\rangle$, and $T_{2,3}= \langle\nabla f(\boldsymbol{\theta}^t), \mathbb{E} [\frac{c}{M}\sum_{m\in\mathcal{S}^t}\frac{\|\mathbf{g}_m^t\|\|\mathbf{r}^t\|-\left\langle\mathbf{g}_m^t, \mathbf{r}^t\right\rangle}{\|\mathbf{r}^t\|^2}\mathbf{r}^t ] \rangle$;
$(a)$ is based on the definition of $\Delta^t$, and $(b)$ comes from decomposition of $\mathbf{g}^t=\frac{1}{S}\sum_{m\in\mathcal{S}^t}\mathbf{g}_m^t$. 

We proceed to bound $T_{2,1}$, $T_{2,2}$ and $T_{2,3}$ in \eqref{bound of T2}.
The upper bound of $T_{2,1}$ is given by
\begin{align}
    &T_{2,1} 
    \overset{(a)}{=} \left\langle\nabla f(\boldsymbol{\theta}^{t}), \mathbb{E}\left[(1-c)\bar{\mathbf{g}}^t+a_1\right]\right\rangle\notag\\
    &\overset{(b)}{=}\Big\langle\nabla f(\boldsymbol{\theta}^{t}), \mathbb{E}\big[-(1-c)\frac{1}{M}\sum_{m \in \mathcal{M}}\sum_{u=0}^{U-1}\eta\nabla F_m(\boldsymbol{\theta}_m^{t,u})\nonumber\\
    &+(1-c)\eta U\frac{1}{M}\sum_{m \in \mathcal{M}} \nabla F_m(\boldsymbol{\theta}^t)\big]\Big\rangle\notag\\
    &{=}\Big\langle\sqrt{\eta U(1-c)}\nabla f(\boldsymbol{\theta}^{t}),-\frac{\sqrt{\eta(1-c)}}{M\sqrt{U}}\mathbb{E}\big[\sum_{m \in \mathcal{M}}\sum_{u=0}^{U-1}a_3\big]\Big\rangle\notag\\
    &\overset{(c)}{\leq}a_4+\frac{(1-c)\eta}{2UM^2}\mathbb{E}\Bigg[\Big\|\sum_{m \in \mathcal{M}}\sum_{u=0}^{U-1}a_3\Big\|^2\Bigg]\notag\\
    &\overset{(d)}{\leq}a_4+\frac{(1-c)\eta L^2}{2M}\sum_{m \in \mathcal{M}}\sum_{u=0}^{U-1}\mathbb{E}\left[\left\|\boldsymbol{\theta}_m^{t,u}-\boldsymbol{\theta}^t\right\|^2\right]
    ,\label{T21}
\end{align}
where $\Bar{\mathbf{g}}^t \!\!=\!\! \frac{1}{M} \! \sum_{m\in\mathcal{M}}\mathbf{g}_m^t$, $a_3 \!\!=\!\! \nabla F_m(\boldsymbol{\theta}_m^{t,u}) \!-\! \nabla F_m(\boldsymbol{\theta}^t)$, 
and $a_4 \!=\! \frac{(1-c)\eta U}{2} \mathbb{E}\left[ \|\nabla f(\boldsymbol{\theta}^{t})\|^2 \right] $; $(a)$ is due to identical sampling distribution in each round; $(b)$ comes from the definition of $\Bar{\mathbf{g}}^t$; $(c)$ follows from $\langle\mathbf{x},\mathbf{y}\rangle \! \leq \! \frac{1}{2}[\|\mathbf{x}\|^2+\|\mathbf{y}\|^2]$; and $(d)$ is due to the Cauchy-Schwartz inequality and \textbf{Assumption 1}.



The upper bound of $T_{2,2}$ is derived as  
\begin{align}
    & 
    T_{2,2}
    \overset{(a)}{=}\Bigg\langle\sqrt{c\eta U}\nabla f(\boldsymbol{\theta}^{t}), \mathbb{E}\Bigg[\frac{\sqrt{c\eta}}{S\sqrt{U}}\sum_{m\in\mathcal{S}^t}\frac{\left\langle\mathbf{g}_m^t, \mathbf{r}^t\right\rangle}{\|\mathbf{g}_m^t\|\|\mathbf{r}^t\|}a_5\Bigg]\Bigg\rangle\notag\\
    &\overset{(b)}{\leq}\frac{c\eta U}{2}\mathbb{E}\left[ \|\nabla f(\boldsymbol{\theta}^{t})\|^2 \right]+\frac{c\eta}{S^2U}\mathbb{E}\Bigg[\Bigg\|\sum_{m\in\mathcal{S}^t}\frac{\left\langle\mathbf{g}_m^t, \mathbf{r}^t\right\rangle}{\|\mathbf{g}_m^t\|\|\mathbf{r}^t\|}a_5\Bigg\|^2\Bigg]\notag\\
    &\overset{(c)}{\leq}\frac{c\eta U}{2}\mathbb{E}\left[ \|\nabla f(\boldsymbol{\theta}^{t})\|^2 \right]+\frac{c\eta}{SU}\sum_{m\in\mathcal{S}^t}\mathbb{E}\left[\left\|a_5\right\|^2\right],\label{T22}
\end{align}
where $a_5 \!\!\!=\!\!\! \sum_{u=0}^{U-1} \nabla F_m \! \left(\boldsymbol{\theta}_m^{t, u}\right)$; $(a)$ is due to the definition of $\mathbf{g}_m^t$; $(b)$ follows from $\langle\mathbf{x},\mathbf{y}\rangle \!\! \leq \!\! \frac{1}{2}[\|\mathbf{x}\|^2  \!+\! \|\mathbf{y}\|^2]$; and $(c)$ is based on the Cauchy-Schwartz inequality and triangle inequality.

The upper bound of $T_{2,3}$ is given by 
\begin{align}
    & 
    T_{2,3} 
     {=} \Bigg\langle \!\!\! \sqrt{c\eta U}\nabla \!f(\boldsymbol{\theta}^t), \mathbb{E} \! \Bigg[\!\frac{\sqrt{c}}{S\sqrt{\eta U}} \!\! \! \sum_{m\in\mathcal{S}^t} \!\!\!
 \! \frac{\|\mathbf{g}_m^t\|\|\mathbf{r}^t\| \!\!-\!\! \left\langle\mathbf{g}_m^t, \!\mathbf{r}^t\right\rangle}{\|\mathbf{r}^t\|^2}\mathbf{r}^t\! \Bigg] \! \Bigg\rangle\notag\\
    &\overset{(a)}{\leq} \frac{c \eta U}{2} \mathbb{E}\left[ \|\nabla f(\boldsymbol{\theta}^{t})\|^2 \right] \nonumber \\
    &+\frac{c}{2\eta US^2}\mathbb{E}\Bigg[\Bigg\|\sum_{m\in\mathcal{S}^t}\frac{\|\mathbf{g}_m^t\|\|\mathbf{r}^t\|-\left\langle\mathbf{g}_m^t, \mathbf{r}^t\right\rangle}{\|\mathbf{r}^t\|^2}\Bigg\|^2\|\mathbf{r}^t\|^2\Bigg]\notag\\
    &\overset{(b)}{\leq}\frac{c\eta U}{2}\mathbb{E}\left[ \|\nabla f(\boldsymbol{\theta}^{t})\|^2\right] +\frac{c}{2\eta US^2}\mathbb{E}\Bigg[\Big(2\sum_{m\in\mathcal{S}^t}\|\mathbf{g}_m^t\|\Big)^2\Bigg]\notag\\
    &\overset{(c)}{\leq}\frac{c\eta U}{2}\mathbb{E} \! \left[ \|\nabla f(\boldsymbol{\theta}^{t})\|^2 \right] \!+\! {2c\eta\sigma_L^2} \!+\! \frac{2c\eta}{MU} \!\!\! \sum_{m\in\mathcal{M}} 
    \!\! \mathbb{E} \! \left[\left\|a_5\right\|^2\right],\label{T23}
\end{align}
where $(a)$ comes from $\langle\mathbf{x},\mathbf{y}\rangle\leq\frac{1}{2}[\|\mathbf{x}\|^2+\|\mathbf{y}\|^2]$; $(b)$ follows from the triangle inequality, and $(c)$ is based on \textbf{Assumption~2} and uniform worker sampling without replacement.

We further bound $\mathbb{E}[\| a_5\|^2]$ in \eqref{T22} and \eqref{T23}; i.e.,
\begin{align}
    &\mathbb{E} \! \left[\left\| a_5 \right\|^2\right] \!=\! \mathbb{E} \! \Bigg[\Bigg\| \! \sum_{u=0}^{U-1}(a_3 \!+\! \nabla F_m(\boldsymbol{\theta}^t) \!-\! \nabla f(\boldsymbol{\theta}^t) \!+\! \nabla f(\boldsymbol{\theta}^t))\Bigg\|^2\Bigg]\notag\\    
    &\overset{(a)}{\leq} 3UL^2\sum_{u=0}^{U-1}\mathbb{E}[\|\boldsymbol{\theta}_{m}^{t,u}-\boldsymbol{\theta}^t\|^2]+3U^2\sigma_G^2+3U^2\|\nabla f(\boldsymbol{\theta}^t)\|^2\notag\\
    &\overset{(b)}{=}C_1 \mathbb{E}\left[\|\nabla f(\boldsymbol{\theta}^t)\|^2\right]+C_2,\label{ti}
\end{align}
where $C_1=90U^4L^2\eta^2+3U^2$; $C_2=15U^3L^2\eta^2(\sigma_L^2+6U\sigma_G^2)+3U^2\sigma_G^2$;
$(a)$ is due to the Cauchy-Schwartz inequality; $(b)$ is based on \cite[Lemma 3]{reddi2020adaptive}, which yields: $\forall 1 \leq u \leq U $,
\begin{align}
\mathbb{E} \! \left[ \! \left\|\boldsymbol{\theta}_{m}^{t,u} \!-\! \boldsymbol{\theta}^t\right\|^2 \! \right] \!\! \leq \! 5U\eta^2(\sigma_L^2 \!+\!  6U\sigma_G^2) \!+\!  30U^2\eta^2 \mathbb{E} \! \left[ \! \|\nabla \! f(\boldsymbol{\theta}^t)\|^2\right]. \notag 
\end{align}

Together with the upper bounds of $T_1$, $T_{2,1}$, $T_{2,2}$ and $T_{2,3}$ in (\ref{T1}), (\ref{T21}), (\ref{T22}) and (\ref{T23}), respectively, we rewrite \eqref{one-step} as
\begin{align}
    \mathbb{E}[f(&\boldsymbol{\theta}^{t+1})] 
    \!\leq \!f(\boldsymbol{\theta}^{t})\!-\!\eta U\Big(\frac{1\!\!-\!3c}{2}\!-\!\big(\frac{3c}{U}\!\!+\!\frac{\eta L}{2}\big)(90U^3L^2\eta^2\!\!+\!3U)\notag\\
    &-\frac{30(1-c)L^2U^2\eta^2}{2}\Big)\mathbb{E}\left[\|\nabla f(\boldsymbol{\theta}^t)\|^2\right]+C_3,\label{onestepnew}
\end{align}
where $C_3={2c\eta\sigma_L^2}{}+\frac{\eta^2U\sigma_L^2L}{2}+(\frac{3c\eta}{U}+\frac{\eta^2 L}{2})C_2+\frac{5(1-c)\eta^3 L^2U^2}{2}(\sigma_L^2+6U\sigma_G^2)$.
With $\eta \leq \frac{1}{8LU}$, we can prove that there exists a positive constant $\gamma$ satisfying
\begin{align}
    \gamma < \frac{1-3c}{2}-\left(\frac{3c}{U^2}+\frac{\eta L}{2U}\right)C_1-\frac{30(1-c)L^2U^2\eta^2}{2}. \label{gamma}
\end{align}
Rearranging \eqref{onestepnew} and summing from $t=0,\ldots,T-1$ yields
\begin{align} \label{drag final conv}
    \frac{1}{T}\sum_{t=0}^{T-1}\mathbb{E}\left[\|\nabla f(\boldsymbol{\theta}^t)\|^2\right] \leq \frac{f(\boldsymbol{\theta}^0)-f^*}{\gamma \eta U T} + V,
\end{align}
where $V=\frac{C_3}{\gamma U\eta}$. 
The proof is complete.

\section*{Appendix B \\ Proof of Theorem \ref{brdrag theorem}}

As with Appendix A, with the $L$-smoothness of $f$, we have
\begin{align}\label{one-step new}
    \mathbb{E}[f(\boldsymbol{\theta}^{t+1})] 
    & \! \leq \! f(\boldsymbol{\theta}^{t})\! +\! \left\langle\nabla\! f(\boldsymbol{\theta}^{t}), \mathbb{E}[\boldsymbol{\theta}^{t+1}\!\!\!-\!\boldsymbol{\theta}^t]\right\rangle \!+ \!\frac{L}{2} \mathbb{E}[\| \Delta^t \|^2]\notag\\
    & = f(\boldsymbol{\theta}^{t}) + T_1 +T_2 + \frac{L}{2} T_3 ,
\end{align}
where $T_1 = 
    \mathbb{E}\left[ \left\langle \nabla f\left(\boldsymbol{\theta}^t\right), \frac{1}{S} \sum_{m \in \mathcal{A}^t} \mathbf{v}_m^t  \right\rangle \right] $, $T_2 = \mathbb{E}\left[ \left\langle \nabla f\left(\boldsymbol{\theta}^t\right), \frac{1}{S} \sum_{m \in \mathcal{B}^t } \mathbf{v}_m^t  \right\rangle \right] $, and $T_3 =  \mathbb{E}\left[\|\Delta^t\|^2\right] $.

With the definition of $\mathbf{v}_m^t$ in (\ref{vm1}) and \textbf{Assumption 4},  the modified gradient $\mathbf{v}_m^t$ can be rewritten as 
\begin{align} \label{expansion of v}
    \mathbf{v}_m^t = h_m^t {{\bf{r}}^t} + (1 - h_m^t)\left\| {{{\bf{r}}^t}} \right\|\frac{{{\bf{g}}_m^t}}{{\left\| {{\bf{g}}_m^t} \right\|}}, \ \forall m \in \mathcal{A}^t,
\end{align}
where $h_m^t = c^t(1 - x_m^t)$. 
Then, 
$T_{1}$ can be upper bounded by 
\begin{align} \label{T11 new}
    & T_1 \!=\! \mathbb{E} \! \Bigg[ \! \Bigg\langle\nabla f\left(\boldsymbol{\theta}^t\right), \frac{1}{S}\sum_{m \in \mathcal{A}^t} h_m^t {{\bf{r}}^t} + (1 - h_m^t)\left\| {{{\bf{r}}^t}} \right\|\frac{{{\bf{g}}_m^t}}{{\left\| {{\bf{g}}_m^t} \right\|}} \Bigg\rangle \! \Bigg ] \nonumber \\
    & \overset{(a)}{=} -c^t (w^t \!-\! x^t) \eta U\left\|\nabla \! f\left(\boldsymbol{\theta}^t\right)\right\|^2 \nonumber \\
    & + \mathbb{E}\left[\left\langle\nabla \! f\left(\boldsymbol{\theta}^t\right), c^t (w^t \!-\! x^t)\eta U \nabla f\left(\boldsymbol{\theta}^t\right) \!+\! c^t (w^t \!-\! x^t) \mathbf{r}^t\right\rangle\right] \nonumber \\
    &+ \mathbb{E}\Bigg[\Bigg\langle\nabla f\left(\boldsymbol{\theta}^t\right), \frac{1}{S}\sum_{m \in \mathcal{A}^t}  (1 - h_m^t)\left\| {{{\bf{r}}^t}} \right\|\frac{{{\bf{g}}_m^t}}{{\left\| {{\bf{g}}_m^t} \right\|}}\Bigg\rangle\Bigg] , 
\end{align}
where $(a)$ is due to $A^t = w^t S$.

For the malicious workers $m \in \mathcal{A}^t$, define $x_m^t = \frac{\langle\mathbf{g}_m^t, \mathbf{r}^t\rangle}{\|\mathbf{g}_m^t\|\|\mathbf{r}^t\|}$ and $x^t = \frac{1}{S}\sum_{m \in \mathcal{A}^t} x_m^t$.
The second term on the RHS of \eqref{T11 new} is upper bounded, as given by
\begin{align} \label{second of T11 new}
    & \mathbb{E}\left[\left\langle\nabla f\left(\boldsymbol{\theta}^t\right), c^t (w^t \!-\! x^t)\eta U \nabla f\left(\boldsymbol{\theta}^t\right) + c^t (w^t \!-\! x^t) \mathbf{r}^t\right\rangle\right] \nonumber \\
    & \overset{(a)}{=} \left\langle \sqrt{\delta^t} \nabla f\left(\boldsymbol{\theta}^t\right),\frac{\sqrt{\delta^t}}{ {U} } \mathbb{E}\left[ \sum_{u=0}^{U-1} \nabla f\left(\boldsymbol{\theta}^t\right)-\nabla f\left(\boldsymbol{\theta}^{t, u}\right) \right]\right\rangle \nonumber \\
    & \overset{(b)}{\leq} \frac{\delta^t}{2} \mathbb{E}\left[ \left\|\nabla f\left(\boldsymbol{\theta}^t\right)\right\|^2 \right] + \frac{\delta^t L^2}{2U} \sum_{u=0}^{U-1} \mathbb{E}\left[\left\|\boldsymbol{\theta}^{t, u}-\boldsymbol{\theta}^t\right\|^2\right].
\end{align}
Here, $ \delta^t\! = \! c^t (w^t\!-\! x^t) \eta U$. $(a)$ follows from~\eqref{trust}; $(b)$ is due to $\langle\mathbf{x}, \mathbf{y}\rangle \leq \frac{1}{2}\left[\|\mathbf{x}\|^2+\|\mathbf{y}\|^2\right]$ and \textbf{Assumption~1}. 

The last term on the RHS of \eqref{T11 new} is bounded as 
\begin{align} \label{third of T11 new}
    & \mathbb{E}\left[\left\langle\nabla f\left(\boldsymbol{\theta}^t\right), \frac{1}{S}\sum_{m \in \mathcal{A}^t}  (1 - h_m^t)\left\| {{{\bf{r}}^t}} \right\|\frac{{{\bf{g}}_m^t}}{{\left\| {{\bf{g}}_m^t} \right\|}}\right\rangle\right]  \nonumber \\
    & = \mathbb{E}\left[\left\langle \sqrt{\delta^t} \nabla f\left(\boldsymbol{\theta}^t\right), \frac{1}{S\sqrt{\delta^t}} \sum_{m \in \mathcal{A}^t}(1 - h_m^t)\left\| {{{\bf{r}}^t}} \right\|\frac{{{\bf{g}}_m^t}}{{\left\| {{\bf{g}}_m^t} \right\|}}\right\rangle \right] \nonumber \\
    & \overset{(a)}{\leq} \! \frac{\delta^t}{2}\mathbb{E} \! \left[ \! \left\|\nabla \! f \! \left( \boldsymbol{\theta}^t\right)\right\|^2 \right] \!+\! \frac{\left\| \sum_{m \in \mathcal{A}^t}(1 \!-\! h_m^t)\right\|^2}{2S^2 \delta^t}
    \mathbb{E} \! \left[ \! \left\| \mathbf{r}^t \right\|^2 \right],
\end{align}
where $(a)$ comes from $\langle\mathbf{x}, \mathbf{y}\rangle \leq \frac{1}{2}\left[\|\mathbf{x}\|^2+\|\mathbf{y}\|^2\right]$.

Similarly, for the upper bound of $T_{2}$, we have
\begin{align} \label{T12 new}
    & T_{2} 
    = -c^t (1-w^t-y^t) \eta U \mathbb{E}\left[ \left\|\nabla f\left(\boldsymbol{\theta}^t\right)\right\|^2 \right] \nonumber \\ 
    &+ \mathbb{E}\left[\left\langle\nabla f\left(\boldsymbol{\theta}^t\right), c^t (1 \!-\! w^t \!-\! y^t) \left( \mathbf{r}^t + \eta U \nabla f\left(\boldsymbol{\theta}^t\right) \right) \right\rangle\right] \nonumber \\
    &+ \mathbb{E}\left[\left\langle\nabla f\left(\boldsymbol{\theta}^t\right), \frac{1}{S}\sum_{m \in \mathcal{B}^t} \left(1- l_m^t\right) \rho_m^t \mathbf{g}_m^t\right\rangle\right] ,
\end{align}
where $l_m^t = c^t(1 - y_m^t)$ and $y^t=\frac{1}{S} \sum_{m \in \mathcal{B}^t } y_m^t$.

As with \eqref{second of T11 new}, with $\kappa^t = c^t (1-w^t-y^t)\eta U$, the second term on the RHS of \eqref{T12 new} can be bounded as 
\begin{align} \label{second of T12 new}
    &\mathbb{E}\left[\left\langle\nabla f\left(\boldsymbol{\theta}^t\right), c^t (1-w^t-y^t) \left( \mathbf{r}^t + \eta U \nabla f\left(\boldsymbol{\theta}^t\right) \right) \right\rangle\right] \nonumber \\
    & \leq \frac{\kappa^t}{2}\mathbb{E}\!\left[\left\|\nabla f\left(\boldsymbol{\theta}^t\right)\right\|^2\right] \!\!+\!\frac{\kappa^t L^2}{2U} \!\!\sum_{u=0}^{U-1} \mathbb{E}\!\left[\left\|\boldsymbol{\theta}^{t, u}\!\!-\!\boldsymbol{\theta}^t\right\|^2\right],
\end{align}
based on the Cauchy-Schwarz inequality and \textbf{Assumption 1}.  

The third term on the RHS of \eqref{T12 new} is upper bounded as
\begin{align} \label{third of T12 new}
    &  \mathbb{E}\left[\left\langle\nabla f\left(\boldsymbol{\theta}^t\right), \frac{1}{S}\sum_{m \in \mathcal{B}^t} \left(1- l_m^t\right) \rho_m^t \mathbf{g}_m^t\right\rangle\right]  \nonumber \\
    & \overset{(a)}{=} \frac{-1}{b^t} \mathbb{E}\left[\left\langle  b^t \nabla f\left(\boldsymbol{\theta}^t\right), \frac{1}{ S}\sum_{m \in \mathcal{B}^t} b_m^t \sum_{u=0}^{U-1} \eta \nabla F_m\left(\boldsymbol{\theta}_m^{t, u}\right)\right\rangle\right] \nonumber \\
    & \overset{(b)}{\leq} \frac{-b^t \eta U}{2} \mathbb{E} \left[ \left\|\nabla f\left(\boldsymbol{\theta}^t\right)\right\|^2 \right] \nonumber\\
    &  + \frac{1}{2 b^t} \mathbb{E}\Bigg[ \Bigg\| \frac{1}{S} \sum_{m \in \mathcal{B}^t} b_m^t \frac{\sqrt{\eta}}{\sqrt{U}} \sum_{u=0}^{U-1} \nabla F_m(\boldsymbol{\theta}^t) - \nabla F_m\left(\boldsymbol{\theta}_m^{t, u}\right) \Bigg\|^2  \Bigg]  \nonumber \\
    & \overset{(c)}{\leq} \!\! \frac{\eta L^2}{2b^t S} \!\!\! \sum_{m \in \mathcal{B}^t} \!\! ( b_m^t )^2 \!\! \sum_{u=0}^{U-1} \! \mathbb{E} \!\! \left[ \! \left\| \boldsymbol{\theta}^{t} \!\!-\! \boldsymbol{\theta}_m^{t, u} \right\|^2 \! \right] \!\!-\!\! \frac{b^t \eta U}{2}  \!\! \left\|  \nabla \! f \! \left(  \boldsymbol{\theta}^t  \right) \! \right\|^2 ,
\end{align}
where $b_m^t = \left(1 - l_m^t\right) \rho_m^t $, and $b^t = \frac{1}{S}\sum_{m \in \mathcal{B}^t} b_m^t$;
$(a)$ is due to $\mathbf{g}_m^t = \boldsymbol{\theta}_m^{t,U}-\boldsymbol{\theta}^t$ and \textbf{Assumption 2}; $(b)$ comes from $-\!2\left\langle {\bf{x}},{\bf{y}} \right\rangle \!\leq\! -\| {\bf{x}} \|^2\! + \!\| {\bf{x}}\!-\!{\bf{y}} \|^2$; and $(c)$ is due to \textbf{Assumption~1} and the Cauchy-Schwarz inequality.


By the triangle inequality, $\| {\bf{v}}_m^t \| \leq \| \left(1-\lambda_m^t\right) \frac{\left\|\mathbf{r}^t\right\|}{\left\|\mathbf{g}_m^t\right\|} \mathbf{g}_m^t \| + \| \lambda_m^t \mathbf{r}^t \| \leq \left\|\mathbf{r}^t\right\|, \forall i \in \mathcal{S}^t$. Then, $T_3$ is upper bounded by
\begin{align}
    T_3 
    \overset{(a)}{\leq} \frac{1}{S} \sum_{m \in \mathcal{S}^t} \mathbb{E}\left[ \left\|\mathbf{r}^t\right\|^2 \right] = \mathbb{E}\left[\left\|\mathbf{r}^t\right\|^2\right], 
\end{align}
where $(a)$ is based on the Cauchy-Schwarz inequality.


By the definition of $\mathbf{r}^t$ in \eqref{trust}, it follows that 
\begin{align} \label{bound of r}
    & \mathbb{E}\left[\left\|\mathbf{r}^t\right\|^2\right] 
    \overset{(a)}{\leq} \mathbb{E}\left[\left\|\boldsymbol{\theta}^{t, U-1} \!-\! \boldsymbol{\theta}^t \!-\! \eta \nabla f\left(\boldsymbol{\theta}^{t, U-1}\right) \right\|^2\right] \!+\! \eta^2\sigma_L^2 \nonumber \\
    & \overset{(b)}{\leq} (1+\frac{1}{2U-1})\mathbb{E}\left[\left\|\boldsymbol{\theta}^{t, U-1}-\boldsymbol{\theta}^t \right\|^2\right] + \eta^2\sigma_L^2 \nonumber \\
    & + 2\eta^2 U \mathbb{E}\left[\left\|\nabla f\left(\boldsymbol{\theta}^{t,U-1}\right)-\nabla f\left(\boldsymbol{\theta}^{t}\right) + \nabla f\left(\boldsymbol{\theta}^{t}\right)\right\|^2\right]  \nonumber \\
     & \overset{(c)}{\leq} (1 + \frac{1}{2U-1}+4\eta^2 U L^2)\mathbb{E}\left[\left\|\boldsymbol{\theta}^{t, U-1}-\boldsymbol{\theta}^t \right\|^2\right] \nonumber \\
     & + 4\eta^2 U \mathbb{E}\left[\left\|\nabla f\left(\boldsymbol{\theta}^t\right)\right\|^2\right] + \eta^2 \sigma_L^2 \nonumber \\
     & \overset{(d)}{\leq}  \frac{U}{U \!-\! 1}\mathbb{E} \! \left[\left\|\boldsymbol{\theta}^{t, U\!-\!1} \!-\! \boldsymbol{\theta}^t \right\|^2\right] \!+\! 4\eta^2 U \mathbb{E}\!\left[\left\|\nabla f\left(\boldsymbol{\theta}^t\right)\right\|^2\right] \!+\! \eta^2 \sigma_L^2 \nonumber \\
     & \overset{(e)}{\leq} 8 \eta^2 U^2 \mathbb{E}\left[ \left\|\nabla f\left(\boldsymbol{\theta}^t\right)\right\|^2 \right] + 2\eta^2 U \sigma_L^2 ,
\end{align}
where $(a)$ is due to $\mathbb{E}\left[\|\boldsymbol{x}\|^2\right]=\mathbb{E}\left[\|\boldsymbol{x}-\mathbb{E}[\boldsymbol{x}]\|^2+\|\mathbb{E}[\boldsymbol{x}]\|^2\right]$ and \textbf{Assumption 2}; $(b)$ comes from Young's inequality; $(c)$ is based on the Cauchy-Schwarz inequality and \textbf{Assumption 1}; $(d)$ holds since $\frac{1}{2 U-1}+4\eta^2 U L^2 \leq \frac{1}{U-1} $ with $\eta \leq \frac{1}{3UL}$; and $(e)$ is due to $(1+ \frac{1}{U-1})^U \leq 2U$. 
It can also be verified that \eqref{bound of r} holds when $U=1$. For $1 \leq u \leq U$, we now have
\begin{align} \label{drift bound new}
    \mathbb{E} \! \left[ \left\| \boldsymbol{\theta}^{t, u} \!-\! \boldsymbol{\theta}^t \right\|^2 \right] \!\leq\! 8 \eta^2 U^2 \mathbb{E} 
 \! \left[ \left\|\nabla f\left(\boldsymbol{\theta}^t\right)\right\|^2 \right] \!+\! 2 \eta^2 U \sigma_L^2.
\end{align} 

Together with the upper bounds of $T_1$, $T_2$, and $T_3$ based on \eqref{second of T11 new}, \eqref{third of T11 new}, \eqref{second of T12 new}, \eqref{third of T12 new} and \eqref{drift bound new}, we can rewrite \eqref{one-step new} as
\begin{align} \label{one round convergence new}
    &\mathbb{E} \! \left[f\left(\boldsymbol{\theta}^{t+1}\right)\right] \leq f\left(\boldsymbol{\theta}^t\right) \!-\! \eta U D_1^t \mathbb{E} \! \left[ \left\|\nabla f\left(\boldsymbol{\theta}^t\right)\right\|^2 \right] \!+\! D_2^t ,
\end{align}
where $D_1^t = \frac{\kappa^t}{2\eta U} + \frac{b^t}{2} - 8\eta U - 4 (\delta^t+\kappa^t) \eta U L^2 -\frac{4 \eta^2 U^2 L^2}{ b^t S} \sum_{m \in \mathcal{B}^t }(b_m^t)^2 - 4\eta U \frac{ \| \sum_{m \in \mathcal{A}^t}(1- h_m^t) \|^2 }{ S^2 \delta^t}  $, and $D_2^t =2\eta^2 U \sigma_L^2 \Big( \frac{ (\delta^t + \kappa^t) L^2}{2} + \frac{  \left\| \sum_{m \in \mathcal{A}^t}\left(1- h_m^t\right) \right\|^2 }{2 S^2 \delta^t} + \frac{\eta U L^2}{2 b^t S} \sum_{m \in \mathcal{B}^t }\left(b_m^t\right)^2 + 1 \Big)$.

Since the malicious workers can potentially upload any local updates, we have $x^t  \in [-w^t,w^t)$ and $\delta^t   \in (0, 2c^t w^t \eta U]$. Hence,
$y^t  \in (0, 1-w^t ]$,
$(b_m^t)^2  \in [ 0, q^2 ]$,
$b^t  \in [ (1-w^t)(1-c^t)p,(1-w^t)q ]$,
and
$\kappa^t  \in [0,c^t(1-w^t) \eta U)$ under \textbf{Assumptions 3} and \textbf{4}.
Note that $ \left\| \sum_{m \in \mathcal{A}^t}\left(1- h_m^t\right) \right\|^2=0, \forall t$, when $c^t = \frac{w^t}{w^t - x^t} \in [\frac{1}{2},1]$. With the stpdfize $\eta \leq \frac{(1-w^t)(1-c^t)p}{32U+2U L^2(p+2q^2)}, \forall t$, and the above ranges of $\delta^t$, $\kappa^t$, $b^t$ and $b_m^t$, we can prove that there is a positive constant $\chi^t$ satisfying
\begin{align} \label{chi}
    \chi^t
    & <  \frac{(1-w^t)(1-c^t)p}{2} - 8\eta U - \frac{4\eta^2 U^2 L^2}{(1-c^t)p} q^2 \nonumber \\
    & 
    - 4c^t (1+ w^t) L^2 \eta^2 U^2 < D_1^t , \ \forall t,
\end{align}
Similarly, $D_2^t$ can be upper bounded as
\begin{align} \label{D_2^t}
    D_2^t \!\leq\! 2\eta^2 U \sigma_L^2 \! \left( \! \frac{(1 \!+\! w^t) c^t  \eta U L^2}{2} \!+\! \frac{ \eta U L^2 q^2}{2(1 \!-\! c^t) p}  \!+\! 1 \! \right) .
\end{align}

Define $\chi = \max_{t} \chi^t$, $D_2 = \max_{t} D_2^t$, $\tilde{c}=\max _t c^t$ and $\tilde{w}=\max _t w^t$. By reorganizing \eqref{one round convergence new} and summing it over $t=0,\cdots,T-1$, we finally obtain
\begin{align} \label{brdrag final conv}
    \frac{1}{T} \sum_{t=0}^{T-1} \mathbb{E} \left[ \left\|\nabla f\left(\boldsymbol{\theta}^t\right)\right\|^2 \right] \leq \frac{f\left(\boldsymbol{\theta}^0\right)-f^*}{\chi \eta U T} + W .
\end{align}
Here, $\eta \!\leq \!\frac{(1-\tilde{w})(1-\tilde{c})p}{32U\!+\!2U L^2(p\!+\!2q^2)}$; $W\!=\!\frac{D_2}{\chi  \eta U}$. 

Since the PS may be unaware of the malicious workers, $c^t = \frac{w^t}{w^t-x^t}$ may not hold. Nevertheless, with the stpdfize $\eta \leq \frac{2c^t(1-c^t)(1-w^t)(1-w^t-y^t)p}{32 p U+ c^t U L^2(p^2+q^2)} , \forall t$, if $c^t$ satisfies
\begin{align}
    \left( (\rho^t)^2 \!-\! 7w^t \!+\! 8x^t \!-\! 1 \right)(c^t)^2 \!+\! (15w^t \!+\! 1)c^t  \geq \frac{8(w^t)^2}{w^t-x^t} ,
\end{align}
we can prove that there also exists a positive constant $\chi^t \leq D_1^t$, and $D_2^t$ can be upper bounded as 
\begin{align} \label{D_2^t new}
    D_2^t \!\leq\! 2 \eta^2 U \sigma_L^2 \! \left( \! \frac{(1 \!+\! w^t) c^t  \eta U L^2}{2} \!\!+\!\! \frac{ \eta U L^2 q^2}{2(1 \!-\! c^t) p}  \!\!+\!\! H^t \!\!+\!\! 1 \! \right) ,
\end{align}
where $H^t = \frac{  \left\| \sum_{m \in \mathcal{A}^t}\left(1- h_m^t\right) \right\|^2 }{2 S^2 \delta^t}$ is bounded. Then, \eqref{brdrag final conv} can also be attained. 
This proof concludes.

\bibliographystyle{IEEEtran}
\bibliography{bib}

\end{document}